\documentclass[11pt,english]{article}
\usepackage[T1]{fontenc}
\usepackage[latin9]{inputenc}
\usepackage{verbatim}
\usepackage{amsmath}
\usepackage{amsthm}
\usepackage{amssymb}

\makeatletter
\numberwithin{equation}{section}
\numberwithin{figure}{section}
  \theoremstyle{remark}
  \newtheorem*{acknowledgement*}{\protect\acknowledgementname}
\theoremstyle{plain}
\newtheorem{thm}{\protect\theoremname}
  \theoremstyle{plain}
  \newtheorem*{conjecture*}{\protect\conjecturename}
  \theoremstyle{plain}
  \newtheorem{fact}[thm]{\protect\factname}
  \theoremstyle{definition}
  \newtheorem{defn}[thm]{\protect\definitionname}
  \theoremstyle{plain}
  \newtheorem{lem}[thm]{\protect\lemmaname}
  \theoremstyle{plain}
  \newtheorem{cor}[thm]{\protect\corollaryname}
  \theoremstyle{remark}
  \newtheorem*{rem*}{\protect\remarkname}
  \theoremstyle{remark}
  \newtheorem{rem}[thm]{\protect\remarkname}
  \theoremstyle{definition}
  \newtheorem{xca}[thm]{\protect\exercisename}

\usepackage{hyperref}
\hypersetup{
    colorlinks,
    allcolors=red
}
\usepackage[lined,boxed,ruled,norelsize,algo2e]{algorithm2e}

\usepackage[margin=1in]
           {geometry}

\usepackage[obeyDraft]{todonotes}

\usepackage{thmtools}
\usepackage{thm-restate}

\makeatother

\usepackage{babel}
  \providecommand{\acknowledgementname}{Acknowledgement}
  \providecommand{\conjecturename}{Conjecture}
  \providecommand{\corollaryname}{Corollary}
  \providecommand{\definitionname}{Definition}
  \providecommand{\exercisename}{Exercise}
  \providecommand{\factname}{Fact}
  \providecommand{\lemmaname}{Lemma}
  \providecommand{\remarkname}{Remark}
\providecommand{\theoremname}{Theorem}

\begin{document}
\global\long\def\defeq{\stackrel{\mathrm{{\scriptscriptstyle def}}}{=}}
\global\long\def\norm#1{\left\Vert #1\right\Vert }
\global\long\def\R{\mathbb{R}}
 \global\long\def\Rn{\mathbb{R}^{n}}
\global\long\def\tr{\mathrm{Tr}}
\global\long\def\diag{\mathrm{diag}}
\global\long\def\Diag{\mathrm{Diag}}
\global\long\def\C{\mathbb{C}}
 \global\long\def\E{\mathbb{E}}
\global\long\def\vol{\text{vol}}
\global\long\def\argmax{\text{argmax}}

\title{Geodesic Walks in Polytopes}

\author{Yin Tat Lee\thanks{Microsoft Research and University of Washington, yintat@uw.edu},
Santosh S. Vempala\thanks{Georgia Tech, vempala@gatech.edu}}
\maketitle
\begin{abstract}
We introduce the geodesic walk for sampling Riemannian manifolds and
apply it to the problem of generating uniform random points from the
interior of polytopes in $\R^{n}$ specified by $m$ inequalities.
The walk is a discrete-time simulation of a stochastic differential
equation (SDE) on the Riemannian manifold equipped with the metric
induced by the Hessian of a convex function; each step is the solution
of an ordinary differential equation (ODE). The resulting sampling
algorithm for polytopes mixes in $O^{*}(mn^{\frac{3}{4}})$ steps.
This is the first walk that breaks the quadratic barrier for mixing
in high dimension, improving on the previous best bound of $O^{*}(mn)$
by Kannan and Narayanan for the Dikin walk. We also show that each
step of the geodesic walk (solving an ODE) can be implemented efficiently,
thus improving the time complexity for sampling polytopes. Our analysis
of the geodesic walk for general Hessian manifolds does not assume
positive curvature and might be of independent interest.
\end{abstract}
\tableofcontents{}

\bigskip{}

\begin{acknowledgement*}
The authors thank Sébastien Bubeck, Ben Cousins, Ton Dieker, Chit
Yu Ng, Sushant Sachdeva, Nisheeth K. Vishnoi and Chi Ho Yuen for helpful
discussions. Part of this work was done while visiting the Simons
Institute for the Theory of Computing, UC Berkeley. The authors thank
Yan Kit Chim for making the illustrations. 

\pagebreak{}
\end{acknowledgement*}

\section{Introduction}

Sampling a high-dimensional polytope is a fundamental algorithmic
problem with many applications. The problem can be solved in randomized
polynomial time. Progress on the more general problem of sampling
a convex body given by a membership oracle \cite{DyerFK89,DyerF90,DyerFK91,LS90,LS92,LS93,KLS97,LV3,LV06,VemSurvey}
has lead to a set of general-purpose techniques, both for algorithms
and for analysis in high dimension. All known algorithms are based
on sampling by discrete-time Markov chains. These include the ball
walk \cite{L90}, hit-and-run \cite{Sm,LV3} and the Dikin walk \cite{KanNar2009},
the last requiring stronger access than a membership oracle. In each
case, the main challenge is analyzing the mixing time of the Markov
chain. For a polytope defined by $m$ inequalities in $\R^{n},$ the
current best complexity of sampling is roughly the minimum of $n^{3}\cdot mn\mbox{ and }mn\cdot mn^{\omega-1}$
where the first factor in each term is the mixing time and the second
factor is the time to implement one step. In fact, the bound of $n^{3}$
on the mixing time (achieved by the ball walk and hit-and-run) holds
for arbitrary convex bodies, and $O(mn)$ is just the time to implement
a membership oracle. The second term is for the Dikin walk, for which
Kannan and Narayanan showed a mixing time of $O(mn)$ for the Dikin
walk \cite{KanNar2009}, with each step implementable in roughly matrix
multiplication time. For general convex bodies given by membership
oracles, $\Omega(n^{2})$ is a lower bound on the number of oracle
calls for all known walks. A quadratic upper bound would essentially
follow from a positive resolution of the KLS hyperplane conjecture
(we mention that \cite{CV2014} show a mixing bound of $\tilde{O}(n^{2})$
for the ball walk for sampling from a Gaussian distribution restricted
to a convex body). The quadratic barrier seems inherent for sampling
convex bodies given by membership oracles, holding even for cubes
and cylinders for the known walks based on membership oracles. It
has not been surpassed thus far even for explicitly described polytopes.

For a polytope in $\R^{n},$ the Euclidean perspective is natural
and predominant. The approach so far has been to define a process
on the points of the polytope so that the distribution of the current
point quickly approaches the uniform (or a desired stationary distribution).
The difficulty is that for points near the boundary of a body, steps
are necessarily small due to the nature of volume distribution in
high dimension. The Dikin walk departs from the standard perspective
by making the distribution of the next step depend on the distances
of the current point to defining hyperplanes of the polytope. At each
step, the process picks a random point from a suitable ellipsoid that
is guaranteed to almost lie inside. This process adapts to the boundary,
but runs into similar difficulties \textemdash{} the ellipsoid has
to shrink as the point approaches the boundary in order to ensure
that (a) the stationary distribution is close to uniform and (b) the
$1$-step distribution is \emph{smooth}, both necessary properties
for proving rapid convergence to the uniform distribution. Even though
this walk has the appealing property of being affine-invariant, and
thus avoids explicitly \emph{rounding }the polytope, the current best
upper bound for mixing is still quadratic, even for cylinders.

An alternative approach for sampling is the simulation of Brownian
motion with boundary reflection \cite{Harrison85,Dieker10,dalalyan2014theoretical,bubeck2015sampling}.
While there has been much study of this process, several difficulties
arise in turning it into an efficient algorithm. In particular, if
the current point is close to the boundary of the polytope, extra
care is needed in simulation and the process effectively slows down.
However, if it is deep inside the polytope, we should expect that
Brownian motion locally looks like a Gaussian distribution and hence
it is easier to simulate. This suggests that the standard Euclidean
metric, which does not take into account distance to the boundary,
is perhaps not the right notion for getting a fast sampling algorithm. 

In this paper, we combine the use of Stochastic Differential Equations
(SDE) with non-Euclidean geometries (Riemannian manifolds) to break
the quadratic barrier for mixing in polytopes. As a result we obtain
significantly faster sampling algorithms. 

Roughly speaking, our work is based on three key conceptual ideas.
The first is the use of a Riemannian metric rather than the Euclidean
metric. This allows us to scale space as we get closer to the boundary
and incorporate boundary information much more smoothly. This idea
was already used by Narayanan \cite{Narayanan16} to extend the Dikin
walk to more general classes of convex bodies. The relevant metrics
are induced by Hessians of convex barrier functions, objects that
have been put to remarkable use for efficient linear and convex optimization
\cite{nesterov1994interior}. The second idea is to simulate an SDE
on the manifold corresponding to the metric, via its geodesics (shortest-path
curves, to be defined presently). Unlike straight lines, geodesics
bend away from the boundary and this allows us to take larger steps
while staying inside the polytope. The third idea is to use a modification
of standard Brownian motion via a \emph{drift }term, i.e., rather
than centering the next step at the current point, we first shift
the current point deterministically, then take a random step. This
drift term compensates the changes of the step size and this makes
the process closer to symmetric. Taken together, these ideas allow
us to simulate an SDE by a discrete series of ordinary differential
equations (ODE), which we show how to solve efficiently to the required
accuracy. In order to state our contributions and results more precisely,
we introduce some background, under three headings.

\paragraph{Riemannian Geometry. }

A manifold can be viewed as a surface embedded in a Euclidean space.
Each point in the manifold (on the surface), has a tangent space (the
linear approximate of the surface at that point) and a local metric.
For a point $x$ in a manifold $M$, the metric at $x$ is defined
by a positive definite matrix $g(x)$ and the length of a vector $u$
in the tangent space $T_{x}M$ is defined as $\norm u_{x}\defeq u^{T}g(x)u$.
By integration, the length of any curve on the manifold is defined
as $\int\norm{\frac{dc}{dt}}_{c(t)}$ . A basic fact about Riemannian
manifolds is that for any point in the manifold, in any direction
(from the tangent space), there is locally a shortest path (\emph{geodesic})\emph{
}starting in that direction. In Euclidean space, this is just a straight
line in the starting direction. Previous random walks involve generating
a random direction and going along a straight line in that direction.
However such straight lines do not take into account the local geometry,
while geodesics do. We give formal definitions in Section \ref{sec:RG}.

\paragraph{Hessian Manifolds. }

In this paper, we are concerned primarily with Riemannian manifolds
induced by Hessians of smooth (infinitely differentiable) strictly
convex functions. More precisely, for any such function $\phi$, the
local metric (of the manifold induced by $\phi$) at a point $x\in M$
is given by the Hessian of $\phi$ at $x$, i.e., $\nabla^{2}\phi(x).$
Since $\phi$ is strictly convex, its Hessian is positive definite
and hence the Riemannian manifold induced by $\phi$ is well-defined
and is called a Hessian manifold. In the context of convex optimization,
we are interested in a class of convex functions called self-concordant
barriers. Such convex functions are smooth in a precise sense and
blow up on the boundary of a certain convex set. The class of Hessian
manifolds corresponding to self-concordant barriers has been studied
and used to study interior-point methods (IPM) \cite{karmarkar1990riemannian,nesterov2002riemannian,NesterovN08}.

Two barriers of particular interest are the logarithmic barrier and
the Lee-Sidford (LS) barrier \cite{lee2014path}, both defined for
polytopes. For a polytope $Ax>b$, with $A\in\R^{m\times n}$ and
$b\in\R^{m}$, for any $x$ in the polytope, the logarithmic barrier
is given by 
\[
\phi(x)=-\sum_{i}\ln(Ax-b)_{i}.
\]
This barrier is efficient in practice and has a self-concordance parameter
$\nu\le m$. The latter controls the number of iterations of the IPM
for optimization as $O(\sqrt{\nu}).$ The best possible value of $\nu$
is $n$, the dimension. This is achieved up to a constant by the \emph{universal
barrier} \cite{nesterov1994interior}, the \emph{canonical barrier}
\cite{hildebrand2014canonical} and the \emph{entropic barrier} \cite{bubeck2014entropic}.
However, these barrier functions take longer to evaluate (currently
$\Omega(n^{5})$ or more). The LS barrier has been shown to be efficiently
implementable (in time $O\left(nnz(A)+n^{2}\right)$ \cite{lee2015efficient}),
while needing only $\tilde{O}(\sqrt{n})$ iterations \cite{lee2014path}. 

In this paper, we develop tools to analyze general Hessian manifolds
and show how to use them for the logarithmic barrier to obtain a faster
algorithm for sampling polytopes. 

\paragraph{Stochastic Differential Equations. }

Given a self-concordant barrier $\phi$ on a convex set $K$, there
is a unique Brownian motion with drift on the Hessian manifold $M$
induced by $\phi$ that has uniform stationary distribution. In the
Euclidean coordinate system, the SDE is given by
\begin{equation}
dx_{t}=\mu(x_{t})dt+\left(\nabla^{2}\phi(x_{t})\right)^{-1/2}dW_{t}\label{eq:SDE}
\end{equation}
where the first term, called \emph{drift, }is given by: 
\begin{equation}
\mu_{i}(x_{t})=\frac{1}{2}\sum_{j=1}^{n}\frac{\partial}{\partial x_{j}}\left(\left(\nabla^{2}\phi(x_{t})\right)^{-1}\right)_{ij}.\label{eq:drift_term_general}
\end{equation}
This suggests an approach for generating a random point in a polytope,
namely to simulate the SDE. The running time of such an algorithm
depends on the convergence rate of the SDE and the cost of simulating
the SDE in discrete time steps.

Since the SDE is defined on the Riemannian manifold $M$, it is natural
to consider the following geodesic walk: 
\begin{equation}
x^{(j+1)}=\exp_{X^{(j)}}(\sqrt{h}w+\frac{h}{2}\mu(x^{(j)}))\label{eq:geo_walk}
\end{equation}
where $\exp_{x^{(j)}}$ is a map from $T_{x^{(j)}}M$ back to the
manifold, $w$ is a random Gaussian vector on $T_{x^{(j)}}M$, $\mu(x^{(j)})\in T_{x^{(j)}}M$
is the drift term and $h$ is the step size. The coefficient of the
drift term depends on the coordinate system we use; as we show in
Lemma \ref{lem:half-drift}, for the coordinate system induced by
a geodesic, the drift term is $\mu/2$ instead of $\mu$ as in (\ref{eq:SDE}).
The Gaussian vector $w$ has mean $0$ and variance $1$ in the metric
at $x$, i.e. for any $u$, $\E_{w}\langle w,u\rangle_{x}^{2}=\norm u^{2}.$
We write it as $w\sim N_{x}(0,I)$.

It can be shown that this discrete walk converges to (\ref{eq:SDE})
as $h\rightarrow0$ and it converges in a rate faster than the walk
suggested by Euclidean coordinates, namely, $x^{(j+1)}=x^{(j)}+\sqrt{h}w+h\mu(x^{(j)})$.
(Note the drift here is proportional to $h$ and not $h/2$.) This
is the reason we study the geodesic walk.

\subsection{Algorithm}

The algorithm is a discrete-time simulation of the geodesic process
(\ref{eq:geo_walk}). For step-size $h$ chosen in advance, let $p(x\overset{w}{\rightarrow}y)$
be the probability density (in Euclidean coordinates) of going from
$x$ to $y$ using the local step $w$. In general, the stationary
distribution of the geodesic process is not uniform and it is difficult
to analyze the stationary distribution unless $h$ is very small,
which would lead to a high number of steps. To get around this issue,
we use the standard method of rejection sampling to get a uniform
stationary distribution. We call this the geodesic walk (see Algo.
\ref{algo:geodesic_detailed} for full details).

\begin{algorithm2e}

\caption{Geodesic Walk}

\SetAlgoLined

Pick a Gaussian random vector $w\sim N_{x}(0,I)$, i.e. $\E_{w}\langle w,u\rangle_{x}^{2}=\norm u^{2}.$ 

Compute $y=\exp_{x}(\sqrt{h}w+\frac{h}{2}\mu(x))$ where $\mu(x)$
is given by (\ref{eq:drift_term_general}).

Let $p(x\overset{w}{\rightarrow}y)$ be the probability density of
going from $x$ to $y$ using the above step $w$.

Compute a corresponding $w'$ s.t. $x=\exp_{y}(\sqrt{h}w'+\frac{h}{2}\mu(y)).$

With probability $\min\left(1,\frac{p(y\overset{w'}{\rightarrow}x)}{p(x\overset{w}{\rightarrow}y)}\right)$,
go to $y$;

Otherwise, stay at $x$.

\end{algorithm2e}

We show how to implement this in Section \ref{subsec:implementation_general}.
Each iteration of the geodesic walk only uses matrix multiplication
and matrix inverse for $O(\log^{O(1)}m)$ many $O(m)\times O(m)$-size
matrices, and the rejection probability is small, i.e., acceptance
probability is at least a constant in each step.

To implement the geodesic walk, we need to compute the exponential
map $\exp_{x}$, the vector $w'$ and the probability densities $p(x\overset{w}{\rightarrow}y),p(y\overset{w'}{\rightarrow}x)$
efficiently. These computational problems turn out to be similar \textemdash{}
all involve solving ordinary differential equations (ODEs) to accuracy
$1/n^{\Theta(1)}$. Hence, one can view the geodesic walk as reducing
the problem of simulating an SDE (\ref{eq:SDE}) to solving a sequence
of ODEs. Although solving ODEs is well-studied, existing literature
seems quite implicit about the dependence on the dimension and hence
it is difficult to apply it directly. In Section \ref{subsec:CollocationMethod},
we rederive some existing results about solving ODEs efficiently,
but with quantitative estimates of the dependence on the dimension
and desired error. 

\subsection{Main result}

In this paper, we analyze the geodesic walk for the logarithmic barrier.
The convergence analysis will need tools from Riemannian geometry
and stochastic calculus, while the implementation uses efficient (approximate)
solution of ODEs. Both aspects appear to be of independent interest.
For the reader unfamiliar with these topics, we include an exposition
of the relevant background.

We analyze the geodesic walk in general and give a bound on the mixing
time in terms of a set of manifold parameters (Theorem \ref{thm:gen-convergence}).
Applying this to the logarithmic barrier, we obtain a faster sampling
algorithm for polytopes, going below the $mn$ mixing time of the
Dikin walk, while maintaining the same per-step complexity.
\begin{thm}[Sampling with logarithmic barrier]
\label{thm:Log-barrier-sampling}For any polytope $\{Ax\geq b\}$
with $m$ inequalities in $\R^{n}$, the geodesic walk with the logarithmic
mixes in $\tilde{O}\left(mn^{\frac{3}{4}}\right)$ steps from a warm
start, with each step taking $\tilde{O}\left(mn^{\omega-1}\right)$
time to implement. 
\end{thm}
The implementation of each step of sampling is based on an efficient
algorithm for solving high-dimensional ODEs (Theorem \ref{thm:Solving_ODE}).
We state the implementation as a general theorem below, and expect
it will have other applications. As an illustration, we show how Physarum
dynamics (studied in \cite{straszak2015natural}) can be computed
more efficiently (Section \ref{subsec:physarum}). 
\begin{thm}
Let $u(t)\in\Rn$ be the solution of the ODE $\frac{d}{dt}u(t)=F(u(t),t),u(0)=v$.
Suppose we are given $\varepsilon>0$ and $1\leq p\leq\infty$ such
that

\begin{enumerate}
\item There is a degree $d$ polynomial $q$ from $\R$ to $\Rn$ such that
$q(0)=v$ and $\norm{\frac{d}{dt}u(t)-\frac{d}{dt}q(t)}_{p}\leq\varepsilon$
for all $0\leq t\leq1$.
\item For some $L\geq1$, we have that $\norm{F(x,t)-F(y,t)}_{p}\leq L\norm{x-y}_{p}$
for all $x,y$ and $0\leq t\leq1$.
\end{enumerate}
Then, we can compute $u$ such that $\norm{u-u(1)}_{p}=O(\varepsilon)$
in $O(ndL^{3}\log^{2}(dK/\varepsilon))$ time and $O(dL^{2}\log(K/\varepsilon))$
evaluations of $F$ where $K=\max_{x,0\leq t\leq1}\norm{F(x,t)}_{p}$. 
\end{thm}
For the application to the geodesic walk, $L=O(1)$ and we did not
optimize over the dependence on $L$.

\subsection{Discussion and open problems}

At a high level, our algorithm is based on the following sequence
of fairly natural choices. First, we consider a Brownian motion that
gives uniform stationary distribution and such Brownian motion is
unique for a given metric (via the Fokker-Planck equation, \ref{thm:Fokker-Planck}).
Since the set we sample is a polytope, we use the metric given by
the Hessian of a self-concordant barrier, a well-studied class of
metrics in convex optimization. This allows us to reduce the sampling
problem to the problem of simulating an SDE.\footnote{Coincidentally, when we use the best known self-concordant barrier,
the canonical barrier, our SDE becomes a Brownian motion with drift
on an Einstein manifold. This is similar to how physical particles
move under general relativity (an algorithm that has been executed
for over 10 billion years!).} To simulate an SDE, we apply the Milstein method, well-known in that
field. To implement the Milstein method, we perform a change of coordinates
to make the metric locally constant, which greatly simplifies the
Milstein approximation. These coordinates are called normal coordinates
and can be calculated by geodesics (Lemma \ref{lem:half-drift}).
This gives the step of our walk (\ref{eq:geo_walk}).

There are two choices which are perhaps not the most natural. First,
it is unclear whether Hessians of self-concordant barriers are the
best metrics for the sampling problem; after all, these barriers were
designed for solving linear programs. Second, the Milstein method
may not be the best choice for discrete approximation. There are other
numerical methods with better convergence rates for solving SDEs,
such as higher-order Runge-Kutta schemes. However, the latter methods
are very complex and it is not clear how to implement them in $\tilde{O}\left(mn^{\omega-1}\right)$
time. 

There are several avenues for improving the sampling complexity further.
One is to take longer steps and use a higher-order simulation of the
SDE that is accurate up to a larger distance. Another part that is
not tight in the analysis is the isoperimetry. Our analysis incurs
a linear factor in the mixing time due to the isoperimetry. As far
as we know, this factor might be as small as $O(1).$ We make this
precise via the following, admittedly rash generalization of the KLS
hyperplane conjecture. If true, it would directly improve our mixing
time bound to $\tilde{O}\left(n^{\frac{3}{4}}\right).$
\begin{conjecture*}
For any Hessian manifold $M$ with metric $d$ induced by a convex
body K, let the isoperimetric ratio of a subset $S$ w.r.t. $d$ be
defined as 

\[
\psi_{d}(S)=\inf_{\varepsilon>0}\frac{\vol(\{x\in K\setminus S,d(x,S)\le\varepsilon\})}{\varepsilon\cdot\min\{\vol(S),\vol(K\setminus S)\}}
\]
and the isoperimetric ratio of $d$ as $\psi=\inf_{S\subset K}\psi_{d}(S)$.
Then there is a subset $S$ defined as a halfspace intersected with
$K$ with $\psi_{d}(S)=O(\psi).$
\end{conjecture*}
We note that we are not simultaneously experts on Riemannian geometry,
numerical SDE/ODE and convex geometry; we view our paper as a sampling
algorithm that connects different areas while improving the state-of-the-art.
Although the sampling problem is a harder problem than solving linear
programs, the step size of geodesic walk we use is larger than that
of the short-step interior point method. Unlike many papers on manifolds,
we do not assume positive curvature everywhere. For example, in recent
independent work, Mangoubi and Smith \cite{mangoubi2016rapid} assumed
positive sectional curvature bounded between $0<\mathfrak{m}$ and
$\mathfrak{M}<\infty$, and an oracle to compute geodesics, and analyzed
a geodesic walk for sampling from the uniform distribution on a Riemannian
manifold; while their mixing time bound $O((\mathfrak{M}/m)^{3})$
is dimension-independent, as far as we know, any manifold that approximates
(the double cover of) a polytope well enough would have $\mathfrak{M}/m=\Omega(n)$
and hence does not yield an improvement. 

We hope that the connections revealed in this paper might be useful
for further development of samplers, linear programming solvers and
the algorithmic theory of manifolds. 

\subsection{Outline}

In the next section, we recall the basic definitions and properties
of Riemannian geometry and Hessian manifolds, and relevant facts from
stochastic calculus and complex analysis. We also derive the discrete-time
geodesic walk (in Section \ref{subsec:Derivation-of-GW}), showing
how the formula naturally arises. Key concepts that we use repeatedly
include geodesics, curvature and Jacobi fields. In Section \ref{sec:Convergence},
we prove the convergence guarantee for general Hessian manifolds.
The analysis needs three high-level ingredients: (1) the rejection
probability of the filter is small (2) two points that are close in
manifold distance have the property that their next step distributions
have large overlap (bounded total variation distance) (3) a geometric
isoperimetric inequality that shows that large subsets have large
boundaries. Of these the last is relatively straightforward, relying
on a comparison with the Hilbert metric and existing isoperimetic
inequalities for the latter. For the first two, we first derive an
explicit formula for the one-step probability distribution (Lemma
\ref{lem:prob_formula}). For bounding the rejection probability,
we need to show that this probability is comparable going forwards
and going backwards as computed in the algorithm (Theorem \ref{thm:rejectionprob}).
The one-step overlap is also derived by comparing the transition probabilities
from two nearby points to the same point (Theorem \ref{thm:dTV}).
This comparison and resulting bounds depend on several smoothness
and stability parameters of the manifold. This part also needs an
auxiliary function that controls the change of geodesics locally.
An important aspect of the analysis is understanding how this probability
changes via the Jacobi fields induced by geodesics. Given these ingredients,
the proof of mixing and conductance follows a fairly standard path
(Section \ref{sec:mixing}). As a warm-up, in Section \ref{subsec:Warmup},
we work out the mixing for a hypercube with the logarithmic barrier
\textemdash{} the mixing time is $\tilde{O}\left(n^{\frac{1}{3}}\right)$. 

In subsequent sections, we apply this general theorem to the logarithmic
barrier for a polytope, to prove Theorem \ref{thm:Log-barrier-sampling}.
We bound each of the parameters and use an explicit auxiliary function
that is just a combination of the infinity norm and the $\ell_{4}$
norm. 

The algorithm for solving ODEs (collocation method) is presented and
analyzed in Section \ref{subsec:CollocationMethod}, and this is used
to compute the geodesic and transition probabilities. The main idea
of the analysis is to show that the ODE can be approximated by a low-degee
polynomial, depending on bounds on the derivatives of the solution.
To bound these derivatives, we give some general relations for bounding
higher derivatives (Section \ref{sec:est_var}). As a simple application
of the collocation method, we give a faster convergence bound for
discretized Physarum dynamics. To apply this method for the log barrier,
in Section \ref{sec:Implementation-log-barrier}, we show that the
functions we wish to compute (geodesic, transition probability) are
complex analytic, then apply the derivative estimates of the previous
section and finally bound the time complexity of the collocoation
method.

\pagebreak{}
\section{Background and notation}

Throughout the paper, we use lowercase letter for vectors and vector
fields and uppercase letter for matrices and tensors (this is not
the convention used in Riemannian geometry). We use $e_{k}$ to denote
coordinate vectors. We use $\frac{d}{dt}$ for the usual derivative,
e.g. $\frac{df(c(t))}{dt}$ is the derivative of some function $f$
along a curve $c$ parametrized by $t$, we use $\frac{\partial}{\partial v}$
for the usual partial derivative. We use $D^{k}f(x)[v_{1},v_{2},\cdots,v_{k}]$
for the $k^{th}$ directional derivative of $f$ at $x$ along $v_{1},v_{2},\cdots,v_{k}$.
We use $\nabla$ for the connection (manifold derivative, defined
below which takes into account the local metric), $D_{v}$ for the
directional derivative of vector wrt to the vector (or vector field)
$v$ (again, defined below which takes into account the local metric),
and $D_{t}$ if the parametrization is clear from the context. We
use $g$ for the local metric. Given a point $x\in M$, $g$ is a
matrix with entries $g_{ij}.$ Its inverse has entries $g^{ij}.$
Also $n$ is the dimension, $m$ the number of inequalities, $\gamma$
is a geodesic, and $\phi$ is a smooth convex function. 

\subsection{Basic Definitions of Riemannian geometry \label{sec:RG}}

Here we first introduce basic notions of Riemannian geometry. One
can think of a manifold $M$ as a $n$-dimensional ``surface'' in
$\R^{k}$ for some $k\geq n$. In this paper, we only use a special
kind of manifolds, called Hessian manifold. For these manifolds, many
definition can be defined directly by some mysterious formulas. If
it helps, the reader can use this section merely to build intuition
and use Lemma \ref{lem:Hessian_formula} instead as the formal definition
of various concepts defined here.

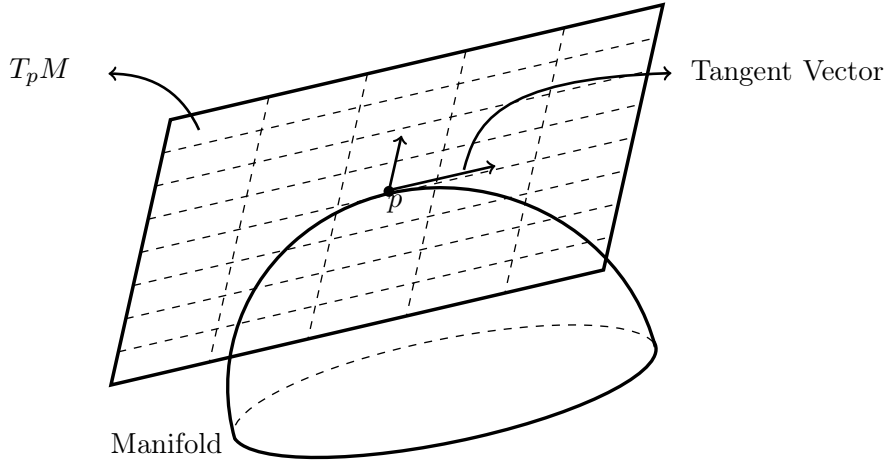
\begin{figure}
\begin{centering}
\begin{tikzpicture}[y=0.80pt, x=0.80pt, yscale=-0.45, xscale=0.5, inner sep=0pt, outer sep=0pt]   \path[draw,fill,line width=0.800pt]     (365.4455,424.6981) ellipse (0.1154cm and 0.1298cm);   \path[draw,line width=1.000pt, <->] (377.3214,366.1266) -- (365.6040,423.4098) -- (466.1347,397.8656);   \path[fill,line width=0.800pt] (364.3816,447.4461) node[above right] (text4209) {$p$};   \path[draw,dash pattern=on 3.00pt off 3.00pt,line width=0.500pt] (251.9807,325.2050) -- (195.7380,604.1067);   \path[draw,dash pattern=on 3.00pt off 3.00pt,line     width=0.500pt] (345.1308,301.0125) -- (288.8882,579.9142);   \path[draw,dash pattern=on 3.00pt off 3.00pt,line     width=0.500pt] (438.2810,276.8200) -- (382.0384,555.7217);   \path[draw,dash pattern=on 3.00pt off 3.00pt,line     width=0.500pt] (531.4312,252.6276) -- (475.1885,531.5293);   \path[draw,dash pattern=on 3.00pt off 3.00pt,line     width=0.500pt] (144.7698,419.1229) -- (610.5207,298.1605);   \path[draw,dash pattern=on 3.00pt off 3.00pt,line     width=0.500pt] (130.7092,488.8483) -- (596.4600,367.8859);   \path[draw,dash pattern=on 3.00pt off 3.00pt,line     width=0.500pt] (116.6485,558.5737) -- (582.3994,437.6114);   \path[draw,dash pattern=on 3.00pt off 3.00pt,line     width=0.500pt] (151.8002,384.2602) -- (617.5510,263.2978);   \path[draw,dash pattern=on 3.00pt off 3.00pt,line     width=0.500pt] (137.7395,453.9856) -- (603.4904,333.0232);   \path[draw,dash pattern=on 3.00pt off 3.00pt,line     width=0.500pt] (123.6789,523.7110) -- (589.4297,402.7486);   \path[draw,dash pattern=on 3.00pt off 3.00pt,line     width=0.500pt] (109.6182,593.4365) -- (575.3691,472.4741);   \path[cm={{0.96789,-0.25138,-0.19768,0.98027,(0.0,0.0)}},draw,line width=1.318pt,rounded corners=0.0000cm] (249.9886,420.5369)     rectangle (731.1910,705.0530); \path[draw,line width=1.000pt,<-] (100.2209,300.8954) .. controls (146.7120,301.3088)     and (171.4692,325.1671) .. (185.9142,360.1494); \path[draw,line width=1.000pt,->] (436.3334,402.0543) .. controls (447.3331,356.6007)     and (463.7234,306.0006) .. (632.6714,300.4063) node[pos=1,right] {\ \ Tangent Vector}; \path[line width=0.800pt] (6.6531,314.0834) node[above right] (text11999) {$T_{p}M$};    \path[line width=0.800pt] (102.2222,700.1400) node[above right] (text4263) {Manifold};   \path[cm={{0.9737,-0.22782,-0.22782,-0.9737,(0.0,0.0)}},draw,line width=1.250pt]     (467.1272,-713.3978)arc(-0.000:60.352:204.634410 and     214.532)arc(60.352:120.704:204.634410 and     214.532)arc(120.704:181.056:204.634410 and 214.532);   \path[cm={{0.96926,-0.24602,0.49025,0.87158,(0.0,0.0)}},draw,line width=1.272pt]     (260.1574,742.2313)arc(-0.000:89.464:204.788350 and     57.803)arc(89.464:178.929:204.788350 and 57.803);   \path[cm={{-0.96926,0.24602,-0.49025,-0.87158,(0.0,0.0)}},draw,dash     pattern=on 3.05pt off 3.05pt,line     width=0.509pt] (149.8904,-745.6280)arc(0.000:89.464:204.788350 and     57.803)arc(89.464:178.929:204.788350 and 57.803); \end{tikzpicture} 
\par\end{centering}
\centering{}\caption{Riemannian Manifold and Tangent Space}
\end{figure}

\begin{enumerate}
\item Tangent space $T_{p}M$: For any point $p$, the tangent space $T_{p}M$
of $M$ at point $p$ is a linear subspace of $\R^{k}$ of dimension
$n$. Intuitively, $T_{p}M$ is the vector space of possible directions
that are tangential to the manifold at $x$. Equivalently, it can
be thought as the first-order linear approximation of the manifold
$M$ at $p$. For any curve $c$ on $M$, the direction $\frac{d}{dt}c(t)$
is tangent to $M$ and hence lies in $T_{c(t)}M$. When it is clear
from context, we define $c'(t)=\frac{dc}{dt}(t)$. For any open subset
$M$ of $\Rn$, we can identify $T_{p}M$ with $\Rn$ because all
directions can be realized by derivatives of some curves in $\Rn$.
\item Riemannian metric: For any $v,u\in T_{p}M$, the inner product (Riemannian
metric) at $p$ is given by $\left\langle v,u\right\rangle _{p}$
and this allows us to define the norm of a vector $\norm v_{p}=\sqrt{\left\langle v,v\right\rangle _{p}}$.
We call a manifold a Riemannian manifold if it is equipped with a
Riemannian metric. When it is clear from context, we define $\left\langle v,u\right\rangle =\left\langle v,u\right\rangle _{p}$.
In $\Rn$ , $\left\langle v,u\right\rangle _{p}$ is the usual $\ell_{2}$
inner product.
\item Differential (Pushforward) $d$: Given a function $f$ from a manifold
$M$ to a manifold $N$, we define $df(x)$ as the linear map from
$T_{x}M$ to $T_{f(x)}N$ such that
\[
df(x)(c'(0))=(f\circ c)'(0)
\]
for any curve $c$ on $M$ starting at $x=c(0)$. When $M$ and $N$
are Euclidean spaces, $df(x)$ is the Jacobian of $f$ at $x$. We
can think of pushforward as a manifold Jacobian, i.e., the first-order
approximation of a map from a manifold to a manifold.
\item Hessian manifold: We call $M$ a Hessian manifold (induced by $\phi$)
if $M$ is an open subset of $\Rn$ with the Riemannian metric at
any point $p\in M$ defined by
\[
\left\langle v,u\right\rangle _{p}=v^{T}\nabla^{2}\phi(p)u
\]
where $v,u\in T_{p}M$ and $\phi$ is a smooth convex function on
$M$.
\item Length: For any curve $c:[0,1]\rightarrow M$, we define its length
by
\[
L(c)=\int_{0}^{1}\norm{\frac{d}{dt}c(t)}_{c(t)}dt.
\]
\item Distance: For any $x,y\in M$, we define $d(x,y)$ be the infimum
of the lengths of all paths connecting $x$ and $y$. In $\Rn$ ,
$d(x,y)=\norm{x-y}_{2}$.
\item Geodesic: We call a curve $\gamma(t):[a,b]\rightarrow M$ a geodesic
if it satisfies both of the following conditions:

\begin{enumerate}
\item The curve $\gamma(t)$ is parameterized with constant speed. Namely,
$\norm{\frac{d}{dt}\gamma(t)}_{\gamma(t)}$ is constant for $t\in[a,b]$.
\item The curve is the locally shortest length curve between $\gamma(a)$
and $\gamma(b)$. Namely, for any family of curve $c(t,s)$ with $c(t,0)=\gamma(t)$
and $c(a,s)=\gamma(a)$ and $c(b,s)=\gamma(b)$, we have that $\left.\frac{d}{ds}\right|_{s=0}\int_{a}^{b}\norm{\frac{d}{dt}c(t,s)}_{c(t,s)}dt=0$.
\end{enumerate}
Note that, if $\gamma(t)$ is a geodesic, then $\gamma(\alpha t)$
is a geodesic for any $\alpha$. Intuitively, geodesics are local
shortest paths. In $\Rn$, geodesics are straight lines. 
\item Exponential map: The map $\exp_{p}:T_{p}M\rightarrow M$ is defined
as
\[
\exp_{p}(v)=\gamma_{v}(1)
\]
where $\gamma_{v}$ is the unique geodesic starting at $p$ with initial
velocity $\gamma_{v}'(0)$ equal to $v$. The exponential map takes
a straight line $tv\in T_{p}M$ to a geodesic $\gamma_{tv}(1)=\gamma_{v}(t)\in M$.
Note that $\exp_{p}$ maps $v$ and $tv$ to points on the same geodesic.
Intuitively, the exponential map can be thought as point-vector addition
in a manifold. In $\R^{n}$, we have $\exp_{p}(v)=p+v$. 
\item Parallel transport: Given any geodesic $c(t)$ and a vector $v$ such
that $\left\langle v,c'(0)\right\rangle _{c(0)}=0$, we define the
parallel transport of $v$ along $c(t)$ by the following process:
Take $h$ to be infinitesimally small and $v_{0}=v$. For $i=1,2,\cdots,1/h$,
we let $v_{ih}$ be the vector orthogonal to $c'(ih)$ that minimizes
the distance on the manifold between $\exp_{c(ih)}(hv_{ih})$ and
$\exp_{c((i-1)h)}(hv_{(i-1)h})$. Intuitively, the parallel transport
finds the vectors on the curve such that their end points are closest
to the end points of $v$. For general vector $v\in T_{c'(0)}$, we
write $v=\alpha c'(0)+w$ and we define the parallel transport of
$v$ along $c(t)$ is the sum of $\alpha c'(t)$ and the parallel
transport of $w$ along $c(t)$. For non-geodesic curve, see the definition
in Fact \ref{fact:basic_RG}. 
\item Orthonormal frame: Given vector fields $v_{1},v_{2},\cdots,v_{n}$
on a subset of $M$, we call $\{v_{i}\}_{i=1}^{n}$ is an orthonormal
frame if $\left\langle v_{i},v_{j}\right\rangle _{x}=\delta_{ij}$
for all $x$. Given a curve $c(t)$ and an orthonormal frame at $c(0)$,
we can extend it on the whole curve by parallel transport and it remains
orthonormal on the whole curve.
\item Directional derivatives and the Levi-Civita connection: For a vector
$v\in T_{p}M$ and a vector field $u$ in a neighborhood of $p$,
let $\gamma_{v}$ be the unique geodesic starting at $p$ with initial
velocity $\gamma_{v}'(0)=v$. Define 
\[
\nabla_{v}u=\lim_{h\rightarrow0}\frac{u(h)-u(0)}{h}
\]
where $u(h)\in T_{p}M$ is the parallel transport of $u(\gamma(h))$
from $\gamma(h)$ to $\gamma(0)$. Intuitively, Levi-Civita connection
is the directional derivative of $u$ along direction $v$, \emph{taking
the metric into account}. In particular, for $\Rn$, we have $\nabla_{v}u(x)=\frac{d}{dt}u(x+tv)$.
When $u$ is defined on a curve $c$, we define $D_{t}u=\nabla_{c'(t)}u$.
In $\Rn$, we have $D_{t}u(\gamma(t))=\frac{d}{dt}u(\gamma(t))$.
We reserve $\frac{d}{dt}$ for the usual derivative with Euclidean
coordinates. 
\end{enumerate}
We list some basic facts about the definitions introduced above that
are useful for computation and intuition.
\begin{fact}
\label{fact:basic_RG}Given a manifold $M$, a curve $c(t)\in M$,
a vector $v$ and vector fields $u,w$ on $M$, we have the following:

\begin{enumerate}
\item (alternative definition of parallel transport) $v(t)$ is the parallel
transport of $v$ along $c(t)$ if and only if $\nabla_{c'(t)}v(t)=0$.
\item (alternative definition of geodesic) $c$ is a geodesic if and only
if $\nabla_{c'(t)}c'(t)=0$. 
\item (linearity) $\nabla_{v}(u+w)=\nabla_{v}u+\nabla_{v}w$.
\item (product rule) For any scalar-valued function f, $\nabla_{v}(f\cdot u)=\frac{\partial f}{\partial v}u+f\cdot\nabla_{v}u$. 
\item (metric preserving) $\frac{d}{dt}\left\langle u,w\right\rangle _{c(t)}=\left\langle D_{t}u,w\right\rangle _{c(t)}+\left\langle u,D_{t}w\right\rangle _{c(t)}$.
\item (torsion free-ness) For any map $c(t,s)$ from a subset of $\R^{2}$
to $M$, we have that $D_{s}\frac{\partial c}{\partial t}=D_{t}\frac{\partial c}{\partial s}$
where $D_{s}=\nabla_{\frac{\partial c}{\partial s}}$ and $D_{t}=\nabla_{\frac{\partial c}{\partial t}}$.
\item (alternative definition of Levi-Civita connection) $\nabla_{v}u$
is the unique linear mapping from the product of vector and vector
field to vector field that satisfies (3), (4), (5) and (6).
\end{enumerate}
\end{fact}

\subsubsection{Curvature}

Here, we define various notions of curvature. Roughly speaking, they
measure the amount by which a manifold deviates from Euclidean space. 

Given vector $u,v\in T_{p}M$, in this section, we define $uv$ be
the point obtained from moving from $p$ along direction $u$ with
distance $\norm u_{p}$ (using geodesic), then moving along direction
``$v$'' with distance $\norm v_{p}$ where ``$v$'' is the parallel
transport of $v$ along the path $u$. In $\Rn$, $uv$ is exactly
$p+u+v$ and hence $uv=vu$, namely, parallelograms close up. For
a manifold, parallelograms almost close up, namely, $d(uv,vu)=o(\norm u\norm v)$.
This property is called being \emph{torsion-free. }

\begin{figure}
\begin{centering}
\hfill{}\begin{tikzpicture}[y=0.80pt, x=0.80pt, yscale=-0.500000, xscale=0.500000] \path[cm={{1.0,0.0,-0.75034,0.66105,(0.0,0.0)}},draw,line width=1pt] (822.3431,751.9528) rectangle (1035.9486,982.7806);  \path[draw,dash pattern=on 3pt off 1pt,line width=1pt] (84.2620,650.0331) -- (83.9450,426.2917)node[pos=1,above]{$w$};  \path[draw,dash pattern=on 3pt off 1pt,line width=1pt] (257.9406,497.6764) -- (218.2036,276.8737);  \path[draw,dash pattern=on 3pt off 1pt,line width=1pt](83.6129,425.7508) -- (217.7519,278.1314);  \path[draw,dash pattern=on 3pt off 1pt,line width=1pt] (336.4017,428.4870) -- (84.0520,426.3433);  \path[draw,dash pattern=on 3pt off 1pt,line width=1pt] (471.5379,495.9237) -- (393.2490,285.3465);  \path[draw,dash pattern=on 3pt off 1pt,line width=1pt] (337.0091,428.1680) -- (298.6542,649.2150);  \path[draw,dash pattern=on 3pt off 1pt,line width=1pt] (547.9405,283.9299) -- (472.4714,495.2052) node[pos=0,above]{$uvw$};  \path[draw,dash pattern=on 3pt off 1pt,line width=1pt] (336.9567,428.4480) -- (546.8298,284.5220);  \path[draw,dash pattern=on 3pt off 1pt,line width=1pt] (393.9144,285.7962) -- (218.4810,277.6440) node[pos=0,above]{$vuw$};  \path[draw,line width=1pt,->] (119.3991,619.2995) -- (116.8834,548.4308);  \path[draw,line width=1pt,->] (154.5278,588.5411) -- (149.5330,517.8039);  \path[draw,line width=1pt,->] (189.3156,557.6874) -- (181.8457,487.1686);  \path[draw,line width=1pt,->] (224.1014,526.8036) -- (214.2523,456.5776);  \path[draw,line width=1pt,->] (257.8481,497.5127) -- (245.8148,430.1965)  node[pos=0,below]{$v$};  \path[draw,line width=1pt,->] (84.2972,650.0581) -- (84.3150,579.1448) node[pos=0,below]{$p$};  \path[draw,line width=1pt,->] (127.1485,649.9085) -- (129.6845,579.0406);  \path[draw,line width=1pt,->] (170.6047,649.9987) -- (175.5314,579.2568);  \path[draw,line width=1pt,->] (255.2158,649.6706) -- (265.0865,579.4476);  \path[draw,line width=1pt,->] (298.7792,649.4337) -- (311.0888,579.5970) node[pos=0,below]{$u$};  \path[draw,line width=1pt,->] (213.2858,649.4170) -- (220.7998,578.9029);  \path[draw,line width=1pt,<-] (286.1249,427.4137) -- (300.9963,496.7501);  \path[draw,line width=1pt,<-] (325.7382,428.2247) -- (343.2012,496.9542);  \path[draw,line width=1pt,<-] (366.7538,429.0345) -- (386.4546,497.1562);  \path[draw,line width=1pt,<-] (447.2251,430.3964) -- (471.8095,496.9118);  \path[draw,line width=1pt,<-] (407.8654,429.3206) -- (429.6094,496.8180);  \path[draw,line width=1pt,->] (334.1954,618.2156) -- (348.8285,548.8285);  \path[draw,line width=1pt,->] (368.5949,587.8981) -- (385.8054,519.1049);  \path[draw,line width=1pt,->] (402.9377,557.6584) -- (422.4474,489.4817);  \path[draw,line width=1pt,->] (437.9636,526.8325) -- (459.9720,459.4209);  \path[draw,line width=1pt,->] (472.2228,496.6352) -- (496.5315,430.0186) node[pos=0,right]{$vu \sim uv$};
\end{tikzpicture}\hfill{}\begin{tikzpicture}[y=0.80pt, x=0.80pt, yscale=-0.350000, xscale=0.350000]   \path[draw=black,dash pattern=on 4.80pt off 4.80pt,line width=1pt] (57.4689,451.8057)     .. controls (144.5748,662.0838) and (233.2419,723.7170) ..     (343.5827,815.6919);   \path[draw=black,dash pattern=on 4.80pt off 4.80pt,line width=1pt]     (260.8495,204.7092) .. controls (415.6174,263.7272) and (482.7235,343.1355) ..     (566.1984,697.0648);   \path[draw=black,dash pattern=on 4.80pt off 4.80pt,line width=1pt]     (195.2339,267.8478) .. controls (342.2029,327.2485) and (453.0993,613.5705) ..     (453.0993,613.5705);   \path[cm={{0.82919,-0.55897,0.08056,0.99675,(0.0,0.0)}},draw=black,fill=black!20,line width=1pt]     (275.2935,611.9233) ellipse (5.5444cm and 1.6844cm)  node{$S(0)$};   \path[cm={{0.78678,-0.61724,0.14521,0.9894,(0.0,0.0)}},draw=black,fill=black!20,line width=1pt]     (359.8904,908.7310) ellipse (4.3573cm and 0.9148cm) node{$S(t)$};   \path[draw=black,dash pattern=on 4.80pt off 4.80pt,line width=1pt]     (110.5263,384.6930) .. controls (128.3208,432.2078) and (146.1153,460.4194) ..     (163.9098,486.1968);   \path[draw=black,dash pattern=on 4.80pt off 4.80pt,line width=1pt]     (212.0300,556.8728) .. controls (254.8015,622.4465) and (297.5731,656.1454) ..     (340.3447,705.7458);   \path[draw=black,dash pattern=on 4.80pt off 4.80pt,line width=1pt]     (375.6279,742.0559) .. controls (422.9963,793.1837) and (422.9963,793.1837) ..     (422.9963,793.1837);   \path[draw=black,dash pattern=on 4.80pt off 4.80pt,line width=1pt]     (475.1798,661.8767) .. controls (490.6728,697.0880) and (490.6728,697.0880) ..     (490.6728,697.0880);   \path[draw,line width=2pt,->]     (116.1725,572.8793) .. controls (188.7667,685.5475) and (188.7667,685.5475) ..     (188.7667,685.5475) node[left]{$v\ $};   \path[draw,line width=2pt,->]     (212.0418,557.4682) .. controls (299.5113,668.3450) and (299.5113,668.3450) ..     (299.5113,668.3450) node[left]{$v\ $};   \path[draw,line width=2pt,->]     (316.8103,369.1200) .. controls (394.5018,485.0058) and (394.5018,485.0058) ..     (394.5018,485.0058) node[left]{$v\ $};   \path[draw,line width=2pt,->]     (438.8704,337.6744) .. controls (512.3536,451.2406) and (512.3536,451.2406) ..     (512.3536,451.2406) node[left]{$v\ $}; \end{tikzpicture}\hfill{}
\par\end{centering}
\centering{}\caption{Riemann curvature tensor measures the deviation of parallelepipeds
(\ref{eq:Ruvw_explain}) and Ricci curvature measures the change of
volumes (\ref{eq:Ric_explain}). The illustrations are inspired by
\cite{ollivier2013visual}.}
\end{figure}
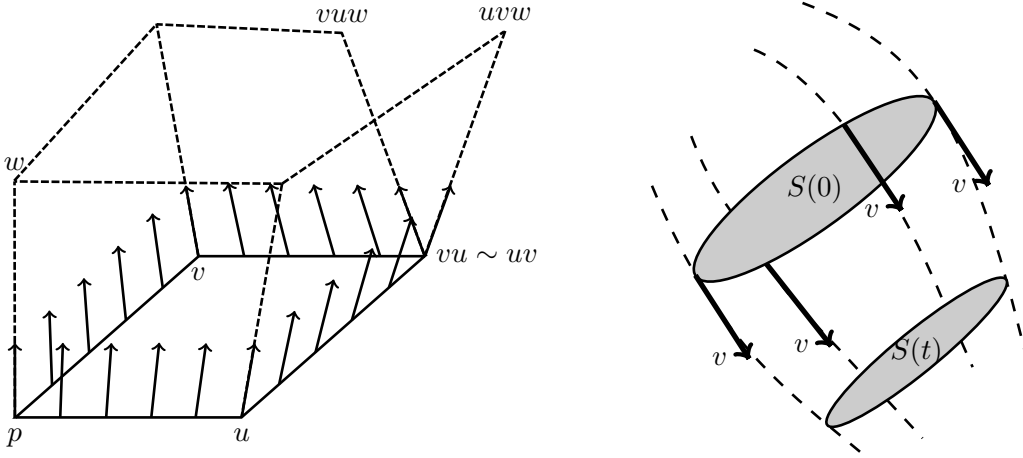

\begin{enumerate}
\item Riemann curvature tensor: Three-dimensional parallelepipeds might
not close up, and the curvature tensor measures how far they are from
closing up. Given vector $u,v,w\in T_{p}M$, we define $uvw$ as the
point obtained by moving from $uv$ along direction ``$w$'' for
distance $\norm w_{p}$ where ``$w$'' is the parallel transport
of $w$ along the path $uv$. In a manifold, parallelepipeds do not
close up and the Riemann curvature tensor how much $uvw$ deviates
from $vuw$. Formally, for vector fields $v$, $w$, we define $\tau_{v}w$
be the parallel transport of $w$ along the vector field $v$ for
one unit of time. Given vector field $v,w,u$, we define the Riemann
curvature tensor by 
\begin{equation}
R(u,v)w=\left.\frac{d}{ds}\frac{d}{dt}\tau_{su}^{-1}\tau_{tv}^{-1}\tau_{su}\tau_{tv}w\right|_{t,s=0}.\label{eq:Ruvw_explain}
\end{equation}
Riemann curvature tensor is a tensor, namely, $R(u,v)w$ at point
$p$ depends only on $u(p)$, $v(p)$ and $w(p)$. 
\item Ricci curvature: Given a vector $v\in T_{p}M$, the Ricci curvature
$\text{Ric}(v)$ measures if the geodesics starting around $p$ with
direction $v$ converge together. Positive Ricci curvature indicates
the geodesics converge while negative curvature indicates they diverge.
Let $S(0)$ be a small shape around $p$ and $S(t)$ be the set of
point obtained by moving $S(0)$ along geodesics in the direction
$v$ for $t$ units of time. Then, 
\begin{equation}
\text{vol}S(t)=\text{vol}S(0)(1-\frac{t^{2}}{2}\text{Ric}(v)+\text{smaller terms}).\label{eq:Ric_explain}
\end{equation}
Formally, we define
\[
\text{Ric}(v)=\sum_{u_{i}}\left\langle R(v,u_{i})u_{i},v\right\rangle 
\]
where $u_{i}$ is an orthonormal basis of $T_{p}M$. Equivalently,
we have $\text{Ric}(v)=\E_{u\sim N(0,I)}\left\langle R(v,u)u,v\right\rangle $.
For $\R^{n}$, $\text{Ric}(v)=0$. For a sphere in $n+1$ dimension
with radius $r$, $\text{Ric}(v)=\frac{n-1}{r^{2}}\norm v^{2}$.
\end{enumerate}
\begin{fact}[Alternative definition of Riemann curvature tensor]
\label{fact:formula_R} %
Given any $M$-valued function $c(t,s)$, we have vector fields $\frac{\partial c}{\partial t}$
and $\frac{\partial c}{\partial s}$ on $M$. Then, for any vector
field $z$, 
\[
R(\frac{\partial c}{\partial t},\frac{\partial c}{\partial s})z=\nabla_{\frac{\partial c}{\partial t}}\nabla_{\frac{\partial c}{\partial s}}z-\nabla_{\frac{\partial c}{\partial s}}\nabla_{\frac{\partial c}{\partial t}}z.
\]
Equivalently, we write $R(\partial_{t}c,\partial_{s}c)z=D_{t}D_{s}z-D_{s}D_{t}z$.
\end{fact}

\begin{fact}
\label{fact:sym_cuv}Given vector fields $v,u,w,z$ on $M$,
\[
\left\langle R(v,u)w,z\right\rangle =\left\langle R(w,z)v,u\right\rangle =-\left\langle R(u,v)w,z\right\rangle =-\left\langle R(v,u)z,w\right\rangle .
\]
\end{fact}

\subsubsection{Jacobi field\label{subsec:Jacobi-field}}

In this paper, we often study the behavior of a family of geodesics.
One crucial fact we use is that the change of a family of geodesics
satisfies the following equation, called the Jacobi equation:
\begin{thm}[{\cite[Thm 4.2.1]{jost2008riemannian}}]
\label{thm:Jacobi_equation}Let $c:[0,\ell]\rightarrow M$ be a geodesic
and $c(t,s)$ be a variation of $c(t)$ (i.e. $c(t,0)=c(t)$) such
that $c_{s}(t)\defeq c(t,s)$ is a geodesic for all $s$. Then $u(t)\defeq\left.\frac{\partial}{\partial s}c(t,s)\right|_{s=0}$
satisfies the following equation
\[
D_{t}D_{t}u+R(u,\frac{dc}{dt})\frac{dc}{dt}=0
\]
where $R(\cdot,\cdot)\cdot$ is Riemann curvature tensor defined before.
Conversely, any vector field $V$ on $c(t)$ satisfying the equation
\[
D_{t}D_{t}V+R(V,\frac{dc}{dt})\frac{dc}{dt}=0
\]
can be obtained by a variation of $c(t)$ through geodesics. We call
any vector field satisfying this equation a Jacobi field.
\end{thm}
See figure \ref{fig:Jacobifield} for an illustration.

\begin{figure}
\begin{centering}
\begin{tikzpicture}[y=0.80pt, x=0.80pt, yscale=-0.5, xscale=0.6]   \path[draw,dashed,line width=1pt] (0.1622,605.7811) .. controls (114.4521,535.6969) and (325.4295,429.6695) ..   (486.1622,517.7811) node[right,scale=1] {Geodesic $c(t,0)$};   \path[draw,dashed,line width=1pt] (0.4189,576.9433) .. controls (229.6782,381.5019) and (356.4545,330.9561) ..  (484.4189,410.9433) node[right,scale=1] {Geodesic $c(t,1/2)$};   \path[draw,dashed,line width=1pt] (0.7094,541.5266) .. controls (94.9993,427.6772) and (297.9767,198.2086) ..  (482.7094,274.6337) node[right,scale=1] {Geodesic $c(t,1)$};   \path[draw,->,line width=1pt] (0.0522,605.8188)--(0.0522,591.7023);   \path[draw,->,line width=1pt] (99.0387,551.2472)--(99.0387,526.2842);   \path[draw,->,line width=1pt] (196.7748,511.9232)-- (196.7748,471.9628);   \path[draw,->,line width=1pt] (293.0893,488.1887)-- (293.0893,438.6656);   \path[draw,->,line width=1pt] (396.3352,488.2392)-- (396.3352,420.2521) ;   \path[draw,->,line width=1pt] (485.7360,517.9052) -- (485.7360,459.3281) node[right,scale=1] {Jacobi field $V(t)=\frac{\partial c(t,s)}{\partial s}$};   \path[draw,->,line width=1pt] (9.0641,383.9608) -- (9.0641,335.2655) node[above,scale=1] {$s$};   \path[draw,->,line width=1pt] (8.2630,383.9716) -- (56.9583,383.9716) node[right,scale=1] {$t$};
\end{tikzpicture}
\par\end{centering}
\centering{}\caption{\label{fig:Jacobifield}A Jacobi field describes the difference between
geodesic.}
\end{figure}
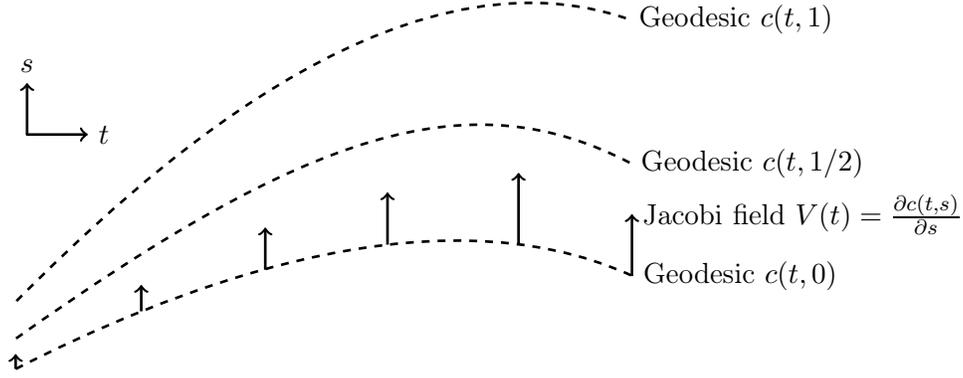

In the appendix, we prove the first part of the theorem above as an
illustration of the Jacobi equation. In $\Rn$, geodesics are straight
lines and a Jacobi field is linear, namely, $u(t)=u(0)+u'(0)t$. A
similar decomposition holds for any Jacobi field.
\begin{fact}
\label{fact:Jacobi_field_split}Given a unit speed geodesic $c(t)$,
every Jacobi field $u(t)$ on $c(t)$ can be split into a tangential
part $u_{1}$ and a normal part $u_{2}$ such that 

\begin{enumerate}
\item $u=u_{1}+u_{2}$,
\item $u_{1}$ and $u_{2}$ are Jacobi fields, namely, $D_{t}D_{t}u_{1}+R(u_{1},\frac{dc}{dt})\frac{dc}{dt}=0$
and $D_{t}D_{t}u_{2}+R(u_{2},\frac{dc}{dt})\frac{dc}{dt}=0$,
\item $u_{1}$ is parallel to $\frac{dc}{dt}$ and is linear, namely, $u_{1}(t)=\left(\left\langle u(0),\frac{dc}{dt}(0)\right\rangle _{c(0)}+\left\langle D_{t}u(0),\frac{dc}{dt}(0)\right\rangle _{c(0)}t\right)\frac{dc}{dt}(t)$,
\item $u_{2}$ is orthogonal to $\frac{dc}{dt}$, namely, $\left\langle u_{2}(t),\frac{dc}{dt}(t)\right\rangle _{c(t)}=0$.
\end{enumerate}
\end{fact}
\begin{defn}
\label{def:Rt_definition}Given a geodesic $\gamma(t)$, we define
a linear map $R(t):T_{\gamma(t)}M\rightarrow T_{\gamma(t)}M$ by

\[
R(t)u=R(u,\gamma'(t))\gamma'(t).
\]
In particular, the Jacobi equation on $\gamma(t)$ can be written
as $D_{t}^{2}u+R(t)u=0$. 

Given a coordinate system $x_{i}$, the linear map $R(t)$ can be
written as a symmetric matrix
\[
R(t)_{ij}=\left\langle R(x_{i}(t),\gamma'(t))\gamma'(t),x_{j}(t)\right\rangle .
\]
\end{defn}
To analysis the Jacobi equation, it is convenient to adopt the follow
matrix notation.
\begin{defn}
Given a linear map $A:T_{x}M\rightarrow T_{x}M$, we define $\norm A_{2}=\max_{\norm v_{x}=1}\norm{Av}_{x}$,
$\norm A_{F}=\sum_{i,j}\left\langle v_{i},Av_{j}\right\rangle _{x}^{2}$
and $\tr A=\sum_{i}v_{i}Av_{i}$ where $\{v_{i}\}_{i=1}^{n}$ is some
arbitrary orthonormal basis of $T_{x}M$.
\end{defn}
In particular, we have that $Ric(\gamma'(t))=\tr R(t)$.

\subsubsection{Hessian manifolds}

Recall that a manifold is called Hessian if it is a subset of $\Rn$
and its metric is given by $g_{ij}=\frac{\partial^{2}}{\partial x^{i}\partial x^{j}}\phi$
for some smooth convex function $\phi$. We let $g^{ij}$ be entries
of the inverse matrix of $g_{ij}$. For example, we have $\sum_{j}g^{ij}g_{jk}=\delta_{ik}$.
We use $\phi_{ij}$ to denote $\frac{\partial^{2}}{\partial x^{i}\partial x^{j}}\phi$
and $\phi_{ijk}$ to denote $\frac{\partial^{3}}{\partial x^{i}\partial x^{j}\partial x^{k}}\phi$. 

Since a Hessian manifold is a subset of Euclidean space, we identify
tangent spaces $T_{p}M$ by Euclidean coordinates. The following lemma
gives formulas for the Levi-Civita connection and curvature under
Euclidean coordinates. 
\begin{lem}[\cite{totaro2004curvature}]
\label{lem:Hessian_formula}Given a Hessian manifold $M$, vector
fields $v,u,w,z$ on $M$, we have the following:

\begin{enumerate}
\item (Levi-Civita connection) $\nabla_{v}u=\sum_{ik}v_{i}\frac{\partial u_{k}}{\partial x_{i}}e_{k}+\sum_{ijk}v_{i}u_{j}\Gamma_{ij}^{k}e_{k}$
where $e_{k}$ are coordinate vectors and the Christoffel symbol 
\[
\Gamma_{ij}^{k}=\frac{1}{2}\sum_{l}g^{kl}\phi_{ijl}.
\]
\item (Riemann curvature tensor) $\left\langle R(u,v)w,z\right\rangle =\sum_{ijlk}R_{klij}u_{i}v_{j}w_{l}z_{k}$
where 
\[
R_{klij}=\frac{1}{4}\sum_{pq}g^{pq}\left(\phi_{jkp}\phi_{ilq}-\phi_{ikp}\phi_{jlq}\right).
\]
\item (Ricci curvature) $Ric(v)=\frac{1}{4}\sum_{ijlkpq}g^{pq}g^{jl}\left(\phi_{jkp}\phi_{ilq}-\phi_{ikp}\phi_{jlq}\right)v_{i}v_{k}$.
\end{enumerate}
\end{lem}
As an exercise, the reader can try to prove Fact \ref{fact:basic_RG}
for Hessian manifolds using the above lemma as a definition. The proof
is given in the appendix. 

In this paper, geodesics are everywhere and we will be using the following
lemma in all of our calculations. 
\begin{lem}[{\cite[Cor 3.1]{nesterov2002riemannian}}]
If $\phi$ is a self-concordant function, namely, $\left|D^{3}f(x)[h,h,h]\right|\leq2(D^{2}f(x)[h,h])^{3/2}$,
then the corresponding Hessian manifold $M$ is geodesically complete,
namely, for any $p\in M$, the exponential map is defined on the entire
tangent space $T_{p}M$ and for any two points $p,q\in M$, there
is a length minimizing geodesic connecting them.

In particular, for a polytope $M=\{x\,:\,Ax>b\}$, the Hessian manifold
induced by the function $\phi(x)=-\sum_{i}\log(a_{i}^{T}x-b_{i})$
is geodesically complete.
\end{lem}

\subsubsection{Normal coordinates\label{subsec:normal_coordinate}}

For any manifold $M$, and any $p\in M$, the exponential map $\ \exp_{x}$
maps from $T_{x}M$ to $M$. Since $T_{x}M$ is isomorphic to $\Rn$,
$\exp_{x}^{-1}$ gives a local coordinate system of $M$ around $x$.
We call this system \emph{the normal coordinates} at $x$. In a normal
coordinate system, the metric is locally constant.
\begin{lem}
\label{lem:normal_coordinate_constant}In normal coordinates, we have
\[
g_{ij}(x)=\delta_{ij}-\frac{1}{3}\sum_{kl}R_{ikjl}(x)x^{k}x^{l}+O(|x|^{3}).
\]
\end{lem}
For a Hessian manifold, one can do a linear transformation to make
the normal coordinates coincide with Euclidean coordinates up to the
first order. 
\begin{lem}
\label{lem:Hessian_normal_map}Given a Hessian manifold $M$ and any
point $x\in M$. We pick a basis of $T_{x}M$ such that 
\[
\exp_{x}(tv)=x+tv+O(t^{2}).
\]
Let $F:M\rightarrow\Rn$ be the normal coordinates defined by $F(y)=\exp_{x}^{-1}$.
Then, we 
\[
DF[h]=h\text{ and }D^{2}F_{k}(x)[h,h]=h^{T}\Gamma^{k}h
\]
where $F_{k}$ is the $k^{th}$ coordinate of $F$ and $\Gamma^{k}$
is the matrix with entries $\Gamma_{ij}^{k}$ defined in Lemma \ref{lem:Hessian_formula}.
\end{lem}

\subsection{Stochastic calculus \label{subsec:Stochastic-Calculus}}

A stochastic differential equation (SDE) describes a stochastic process
over a domain $\Omega.$ It has the form $dx_{t}=\mu(x_{t},t)dt+\sigma(x_{t},t)dW_{t}$
where $x_{t}$ is the current point at time $t,$ $W_{t}$ is a standard
Brownian motion, and $\mu(x_{t},t),\sigma(x_{t},t)$ are the mean
and covariance of the next infinitesimal step at time $t$. 
\begin{lem}[It\={o}'s lemma]
\label{lem:Ito} Given a SDE $dx_{t}=\mu(x_{t})dt+\sigma(x_{t})dW_{t}$
and any smooth function $f$, we have 
\[
df(t,x_{t})=\left\{ \frac{\partial f}{\partial t}+\left\langle \nabla f,\mu\right\rangle +\frac{1}{2}\tr\left[\sigma^{T}\left(\nabla^{2}f\right)\sigma\right]\right\} dt+\left(\nabla f\right)^{T}\sigma dW_{t}.
\]
\end{lem}
SDEs are closely related to diffusion equations: 
\begin{eqnarray*}
\frac{\partial}{\partial t}p(x,t) & = & \frac{1}{2}\nabla\cdot\left(A(x,t)\nabla p(x,t)\right)
\end{eqnarray*}
where $p(x,t)$ is the density at point $x$ and time $t,$ $\nabla\cdot$
is the usual divergence operator, $\nabla p$ is the gradient of $p$
and the matrix $A(x,t)$ represents the diffusion coefficient at point
$x$ and time $t$. When $A(x,t)=I$, we get the familiar heat equation:

\[
\frac{\partial}{\partial t}p(x,t)=\frac{1}{2}\Delta p(x,t).
\]
In this paper, the diffusion coefficient will be a symmetric positive
definite matrix given by the Hessian $\left(\nabla^{2}\phi(x)\right)^{-1}$
of a convex function $\phi(x).$ 

The Fokker-Planck equation connects an SDE to a diffusion equation.
\begin{thm}[Fokker-Planck equation]
\label{thm:Fokker-Planck}For any stochastic differential equation
(SDE) of the form
\[
dx_{t}=\mu(x_{t},t)dt+\sqrt{A(x_{t},t)}dW_{t},
\]
the probability density of the SDE is given by the diffusion equation
\[
\frac{\partial}{\partial t}p(x,t)=-\sum_{i=1}^{n}\frac{\partial}{\partial x_{i}}[\mu_{i}(x,t)p(x,t)]+\frac{1}{2}\sum_{i=1}^{n}\sum_{j=1}^{n}\frac{\partial^{2}}{\partial x_{i}x_{j}}[A_{ij}(x,t)p(x,t)].
\]
\end{thm}

\subsubsection{\label{subsec:Derivation-of-GW}Derivation of the Geodesic walk}

Given a smooth convex function $\phi$ on the convex domain $M$,
namely that it is convex and is infinitely differentiable at every
interior point of $M$, we consider the corresponding diffusion equation
\begin{eqnarray*}
\frac{\partial}{\partial t}p(x,t) & = & \frac{1}{2}\nabla\cdot\left(\nabla^{2}\phi\right)^{-1}\nabla p.
\end{eqnarray*}
We can expand it by
\begin{eqnarray*}
\frac{\partial}{\partial t}p(x,t) & = & \frac{1}{2}\sum_{i=1}^{n}\frac{\partial}{\partial x_{i}}\left(\sum_{j=1}^{n}\left(\left(\nabla^{2}\phi\right)^{-1}\right)_{ij}\frac{\partial}{\partial x_{j}}p(x,t)\right)\\
 & = & \frac{1}{2}\sum_{i=1}^{n}\sum_{j=1}^{n}\frac{\partial^{2}}{\partial x_{i}x_{j}}\left(\left(\left(\nabla^{2}\phi\right)^{-1}\right)_{ij}p(x,t)\right)-\frac{1}{2}\sum_{i=1}^{n}\sum_{j=1}^{n}\frac{\partial}{\partial x_{i}}\left(\frac{\partial}{\partial x_{j}}\left(\left(\nabla^{2}\phi\right)^{-1}\right)_{ij}p(x,t)\right).
\end{eqnarray*}
The uniform distribution is the stationary distribution of this diffusion
equation. Now applying the Fokker-Planck equation (Theorem \ref{thm:Fokker-Planck})
with $A=\left(\nabla^{2}\phi\right)^{-1}$, the SDE for the above
diffusion is given by: 

\[
dx_{t}=\mu(x_{t})dt+\left(\nabla^{2}\phi(x_{t})\right)^{-1/2}dW_{t}.
\]
This explains the definition of (\ref{eq:SDE}). To simplify the notation,
we write the SDE as 
\begin{equation}
dx_{t}=\mu(x_{t})dt+\sigma(x_{t})dW_{t}\label{eq:Dikin_SDE}
\end{equation}
where $\sigma(x_{t})=\left(\nabla^{2}\phi(x_{t})\right)^{-1/2}$.
One way to simulate this is via the Euler\textendash Maruyama method,
namely
\[
x_{(t+1)h}=x_{th}+\mu(x_{th})h+\sigma(x_{th})w_{th}\sqrt{h}
\]
where $w_{th}\sim N_{x_{th}}(0,I)$. We find the direction we are
heading and take a small step along that direction. However, if we
view $M$ as a manifold, then directly adding the direction $\mu(x_{th})h+\sigma(x_{th})w_{th}\sqrt{h}$
to $x_{th}$ is not natural; the Euclidean coordinate is just an arbitrary
coordinate system and we could pick any other coordinate systems and
add the direction into $x_{th}$, giving a different step. Instead,
we take the step in normal coordinates (Section \ref{subsec:normal_coordinate}). 

In particular, given an initial point $x_{0}$, we define $F=\exp_{x_{0}}^{-1}$
and we note that $F(x_{t})$ is another SDE. To see the defining equation
of this transformed SDE, we use It\={o}'s lemma (Lemma \ref{lem:Ito})
to show that the transformed SDE looks the same but with half the
drift term. This explains the formulation of geodesic walk: $x^{(j+1)}=\exp_{x^{(j)}}(\sqrt{h}w+\frac{h}{2}\mu(x^{(j)}))$.
\begin{lem}
\label{lem:half-drift}Let $F=\exp_{x_{0}}^{-1}$ and $x_{t}$ satisfies
the SDE (\ref{eq:Dikin_SDE}) Then we have
\begin{equation}
dF(x_{0})=\frac{1}{2}\mu(x_{0})dt+\sigma(x_{0})dW_{0}.\label{eq:geodesic_SDE}
\end{equation}
\end{lem}
\begin{proof}
It\={o}'s lemma (Lemma \ref{lem:Ito}) shows that 
\[
dF_{k}(x_{t})=\left\{ \left\langle \nabla F_{k},\mu\right\rangle +\frac{1}{2}\tr\left[\sigma^{T}\left(\nabla^{2}F_{k}\right)\sigma\right]\right\} dt+\left(\nabla F_{k}\right)^{T}\sigma dW_{t}
\]
where $F_{k}$ indicates the $k^{th}$ coordinate of $F$. 

From Lemma \ref{lem:Hessian_normal_map}, we have that $\left\langle \nabla F_{k}(x_{0}),\mu\right\rangle =\mu_{k}$,
$\left(\nabla F_{k}(x_{0})\right)^{T}\sigma=e_{k}^{T}\sigma$ and
\begin{eqnarray*}
\tr\left[\sigma(x_{0})^{T}\left(\nabla^{2}F_{k}(x_{0})\right)\sigma(x_{0})\right] & = & \sum_{i}D^{2}F_{k}[\sigma e_{i},\sigma e_{i}]\\
 & = & \sum_{i}e_{i}^{T}\sigma^{T}\Gamma^{k}\sigma e_{i}\\
 & = & \tr\left(\sigma^{T}\Gamma^{k}\sigma\right)\\
 & = & \sum e_{i}^{T}\Gamma^{k}\left(\nabla^{2}\phi\right)^{-1}e_{i}.
\end{eqnarray*}
Now, using Lemma \ref{lem:Hessian_formula}, we have that
\[
\tr\left[\sigma(x_{0})^{T}\left(\nabla^{2}F_{k}(x_{0})\right)\sigma(x_{0})\right]=\frac{1}{2}\sum_{ijl}g^{kl}\phi_{ijl}g^{ji}
\]
Hence, we have that
\[
dF_{k}(x_{0})=\left\{ \mu_{k}+\frac{1}{4}\sum_{ijl}g^{kl}\phi_{ijl}g^{ji}\right\} dt+e_{k}^{T}\sigma dW_{t}
\]
Recall that the drift term (\ref{eq:drift_term_general}) is given
by 
\begin{eqnarray}
\mu_{k} & = & \frac{1}{2}\sum_{i=1}^{n}\frac{\partial}{\partial x_{i}}\left(\left(\nabla^{2}\phi\right)^{-1}\right)_{ki}\nonumber \\
 & = & -\frac{1}{2}\sum_{i}e_{k}^{T}\left(\nabla^{2}\phi\right)^{-1}\frac{\partial}{\partial x_{i}}\nabla^{2}\phi\left(\nabla^{2}\phi\right)^{-1}e_{i}\nonumber \\
 & = & -\frac{1}{2}\sum_{l,i,j}g^{kl}\frac{\partial}{\partial x_{i}}\phi_{lj}g^{ji}\nonumber \\
 & = & -\frac{1}{2}\sum_{ijl}g^{kl}\phi_{ijl}g^{ji}.\label{eq:formula_mu_cts}
\end{eqnarray}
Therefore, we have the result.
\end{proof}
To understand why Euler-Maruyama method works especially better on
normal coordinates, we recall the following theorem:
\begin{thm}[Euler-Maruyama method]
Given a SDE $dx_{t}=\mu(x_{t})dt+\sigma(x_{t})dW_{t}$ where both
$\mu(x_{t})\in\R^{d}$ and $\sigma(x_{t})\in\R^{d\times m}$ are Lipschitz
smooth functions, consider the algorithm 
\begin{align*}
\overline{x}_{(i+1)h} & =\overline{x}_{ih}+\mu(\overline{x}_{ih})h+\sigma(\overline{x}_{ih})(W_{(i+1)h}-W_{ih})\text{ for all }i\geq0\\
\overline{x}_{0} & =x_{0}.
\end{align*}
For some small enough $h>0$, we have that
\[
\E\norm{\overline{x}_{1}-x_{1}}=O(\sqrt{h}).
\]
\end{thm}
As a comparison, it is known that there is better method such as the
following that gives better error.
\begin{thm}[Milstein method]
\label{thm:milstein_method} Given a SDE $dx_{t}=\mu(x_{t})dt+\sigma(x_{t})dW_{t}$
where both $\mu(x_{t})\in\R^{d}$ and $\sigma(x_{t})\in\R^{d\times m}$
are Lipschitz smooth functions, consider the algorithm 
\begin{align*}
\overline{x}_{(i+1)h} & =\overline{x}_{ih}+\mu(\overline{x}_{ih})h+\sigma(\overline{x}_{ih})(W_{(i+1)h}-W_{ih})+\sum_{j_{1},j_{2}=1}^{m}L_{i}^{j_{1}}\sigma^{k,j_{2}}(\overline{x}_{ih})I_{j_{1}j_{2},i}\text{ for all }i\geq0\\
\overline{x}_{0} & =x_{0}.
\end{align*}
where $L_{i}^{j}=\sum_{k=1}^{d}\sigma^{k,j}(\overline{x}_{ih})\frac{\partial}{\partial x^{k}}$
and $I_{j_{1}j_{2},i}\sim\int_{0}^{h}\int_{ih}^{ih+t_{1}}dW_{t_{2}}^{j_{2}}dW_{t_{1}}^{j_{1}}$.
For some small enough $h>0$, we have that
\[
\E\norm{\overline{x}_{1}-x_{1}}=O(h).
\]
\end{thm}
Note that under normal coordinates, the metric $\sigma$ is locally
constant (Lemma \ref{lem:normal_coordinate_constant}). Due to this,
the term $\sum L^{j_{1}}\sigma^{k,j_{2}}I_{j_{1}j_{2}}$ in the Milstein
method vanishes. Hence, Euler-Maruyama method is equivalent to Milstein
method under normal coordinates. This is one of the reasons we use
geodesic instead of straight line as in Dikin walk. We remark that
this theorem is used for conveying intuition only. 

\subsection{Complex analysis \label{subsec:Complex-analysis}}

A complex function is said to be \emph{(complex) analytic} (equivalently,
\emph{holomorphic}) if it is locally defined by a convergent power
series. Hartog's theorem shows that a complex function in several
variables $f:\C^{n}\rightarrow\C$ is holomorphic iff it is analytic
in each variable (while fixing all the other variables). For any power
series expansion, we define the radius of convergence at $x$ as the
largest number $r$ such that the series converges on the sphere with
radius $r$ centered at $x$. In this paper, we use the fact that
complex analytic functions behave very nicely up to the radius of
convergence, and one can avoid complicated and tedious computations
by using general convergence theorems.
\begin{thm}
[Cauchy's Estimates]\label{thm:cauchy_estimate}Suppose $f$ is holomorphic
on a neighborhood of the ball $B\defeq\{z\in\mathbb{C}\ :\ \left|z-z_{0}\right|\leq r\}$,
then we have that
\[
\left|f^{(k)}(z_{0})\right|\leq\frac{k!}{r^{k}}\sup_{z\in B}\left|f(z)\right|.
\]
In particular, for any rational function $f(z)=\frac{\prod_{i=1}^{\alpha}(z-a_{i})}{\prod_{j=1}^{\beta}(z-b_{j})}$
and any ball $B\defeq\{z\in\mathbb{C}\ :\ \left|z-z_{0}\right|\leq r\}$
such that $b_{j}\notin B$, we have that
\[
\left|f^{(k)}(z_{0})\right|\leq\frac{k!}{r^{k}}\sup_{z\in B}\left|f(z)\right|.
\]
\end{thm}
A similar estimate holds for analytic functions\emph{ }in\emph{ }several
variables. We will also use the following classical theorem.
\begin{thm}[Simplified Version of Cauchy\textendash Kowalevski theorem]
\label{thm:cauchy_kowalevski}If $f$ is a complex analytic function
defined on a neighborhood of $(z_{0},\alpha)\in\mathbb{C}^{n+1}$,
then the problem
\[
\frac{dw}{dz}=f(z,w),\quad w(z_{0})=\alpha,
\]
has a unique complex analytic solution $w$ defined on a neighborhood
around $z_{0}$. 

Similarly, for a complex analytic function $f$ defined in a neighborhood
of $(z_{0},\alpha,\beta)\in\mathbb{C}^{2n+1}$, the ODE
\[
\frac{d^{2}w}{dz^{2}}=f(z,w,\frac{dw}{dz}),\quad w(z_{0})=\alpha,\quad\frac{dw}{dz}(z_{0})=\beta
\]
has a unique complex analytic solution $w$ defined in a neighborhood
around $z_{0}$. 
\end{thm}

\pagebreak{}
\section{Convergence of the Geodesic Walk}

\label{sec:Convergence}The geodesic walk is a Metropolis-filtered
Markov chain, whose stationary distribution is the uniform distribution
over the polytope to be sampled. We will prove that the conductance
of this chain is large with an appropriate choice of the step-size
parameter. Therefore, its mixing time to converge to the stationary
distribution will be small. The proof of high conductance involves
showing (a) the acceptance probability of the Metropolis filter is
at least a constant (b) the induced metric satisfies a strong isoperimetric
inequality (c) two points that are close in metric distance are also
close in probabilistic distance, namely, the one-step distributions
from them have large overlap. Besides bounding the number of steps,
we also have to show that each step of the Markov chain can be implemented
efficiently. We do this in later sections via an efficient algorithm
for approximately solving ODEs.

In this section, we present the general conductance bound for Hessian
manifolds. The bound on the conductance will use several parameters
determined by the specific barrier function. In Section \ref{sec:Logarithmic-Barrier},
we bound these parameters for the logarithmic barrier. 

For a Markov chain with state space $M$, stationary distribution
$Q$ and next step distribution $P_{u}(\cdot)$ for any $u\in M$,
the conductance of the Markov chain is 

\[
\phi=\inf_{S\subset M}\frac{\int_{S}P_{u}(M\setminus S)dQ(u)}{\min\left\{ Q(S),Q(M\setminus S)\right\} }.
\]

The conductance of an ergodic Markov chain allows us to bound its
mixing time, i.e., the rate of convergence to its stationary distribution,
e.g., via the following theorem of Lovász and Simonovits.
\begin{thm}[\cite{LS93}]
Let $Q_{t}$ be the distribution of the current point after $t$
steps of a Markov chain with stationary distribution $Q$ and conductance
at least $\phi,$ starting from initial distribution $Q_{0}.$ Then, 
\end{thm}
\[
d_{TV}(Q_{t},Q)\le\sqrt{d_{0}}\left(1-\frac{\phi^{2}}{2}\right)^{t}
\]
where $d_{0}=\E_{Q_{0}}(dQ_{0}(u)/dQ(u))$ is a measure of the distance
of the starting distribution from the stationary and $d_{TV}$ is
the total variation distance.

\subsection{Hessian parameters\label{subsec:Hessian_parameters}}

The mixing of the walk depends on the maximum values of several smoothness
parameters of the manifold. Since each step of our walk involves a
Gaussian vector which can be large with some probability, many smoothness
parameters inevitably depend on this Gaussian vector. Formally, let
$\gamma$ be the geodesic used in a step of the geodesic walk with
the parameterization $\gamma:[0,\ell]\rightarrow M$ where $\ell\defeq\sqrt{nh}$.
Note that $\ell$ is not exactly the length of the geodesic step,
but it is close with high probability due to Gaussian concentration.
Rather than using supremum bounds for our smoothness parameters, it
suffices to use large probability bounds, where the probability is
over the choice of geodesic at any point $x\in\Omega$. To capture
this notion that ``most geodesics are good'', we allow the use of
an auxiliary function $V(\gamma)\geq0$ to measure how good a geodesic
is. Several of the smoothness parameters assume that this function
is bounded and Lipshitz for a sufficiently large step size $h$. More
precisely, viewing geodesics as maps $\gamma:[0,\ell]\rightarrow M,$
we assume that there exists an auxiliary real function on the tangent
bundle (union of tangent manifolds for all points in $M$), $V:TM\rightarrow\R_{+}$,
satisfying
\begin{enumerate}
\item For $h\le H$ and any variation of geodesics $\gamma_{s}$ with $V(\gamma_{s})\leq V_{0}$,
there is a $V_{1}\ge V_{0}$ s.t. $\left|\frac{d}{ds}V(\gamma_{s})\right|\leq V_{1}\left(\norm{\frac{d}{ds}\gamma_{s}(0)}_{\gamma_{s}(0)}+\ell\norm{D_{s}\gamma_{s}}_{\gamma_{s}(0)}\right)$ 
\item For any $x\in M$,
\begin{equation}
\mathbb{P}_{\text{geodesic }\gamma\text{ from }x}(V(\gamma)\leq\frac{1}{2}V_{0})\geq1-\frac{V_{0}}{100V_{1}}.\label{eq:gamma_good}
\end{equation}
\end{enumerate}
\begin{defn}
\label{def:hessian_parameter}Given a Hessian manifold $M$, maximum
step size $H$ and auxiliary function $V$ with parameters $V_{0},V_{1},$
we define the smoothness parameters $D_{0},D_{1},D_{2,}G_{1},G_{2},R_{1},R_{2}$
depending only on $M$ and the step size $h\le H$ as follows:
\begin{enumerate}
\item The maximum norm of the drift in the local metric, $D_{0}=\sup_{x\in M}\|\mu(x)\|_{x}.$
\item The smoothness of the norm of the drift, $D_{1}=\sup_{h\leq H,V(\gamma)\leq V_{0},0\leq t\leq\ell}\frac{d}{dt}\norm{\mu(\gamma(t))}_{\gamma(t)}^{2}$.
\item The smoothness of the drift, $D_{2}=\sup_{x\in M,\norm s_{x}\leq1}\norm{\nabla_{s}\mu(x)}_{x}$.
\item The smoothness of the local volume, $G_{1}=\sup_{h\leq H,V(\gamma)\leq V_{0},0\leq t\leq\ell}\left|\log\det(g(\gamma(t))))'''\right|$
where $g(x)$ is the metric at $x$.
\item The smoothness of the metric, $G_{2}=\sup\frac{d(x,y)}{d_{H}(x,y)}$
where $d_{H}$ is the Hilbert distance (defined in Section \ref{subsec:Isoperimetry})
and $d$ is the shortest path distance in $M$. 
\item The stability of the Jacobian field, $R_{1}=\sup_{h\leq H,V(\gamma)\leq V_{0},0\leq t\leq\ell}\norm{R(t)}_{F}$
where $R(t)$ is defined in Definition \ref{def:Rt_definition}.
\item The smoothness of the Ricci curvature, $R_{2}=\sup_{h\leq H,V(\gamma)\leq V_{0}}\left|\frac{d}{ds}Ric(\gamma_{s}'(t))\right|$
(see Definition \ref{def:R2}).
\end{enumerate}
\end{defn}
We refer to these as the parameters of a Hessian manifold. Our main
theorem for convergence can be stated as follows.
\begin{thm}
\label{thm:gen-convergence}On a Hessian manifold of dimension $n$
with an auxiliary function, step-size upper bound $H$ and parameters
$D_{0},D_{1},D_{2,}G_{1},G_{2},R_{1},R_{2}$, the geodesic walk with
step size 
\[
h\le\Theta(1)\min\left\{ \frac{1}{n^{1/3}D_{1}^{2/3}},\frac{1}{D_{2}},\frac{1}{nR_{1}},\frac{1}{(nD_{0}R_{1})^{2/3}},\frac{1}{\left(nR_{2}\right)^{2/3}},\frac{1}{nG_{1}^{2/3}},H\right\} 
\]
has conductance $\Omega(\sqrt{h}/G_{2})$ and mixing time $O(G_{2}^{2}/h).$
\end{thm}
In the rest of this section, we prove this theorem. It can be sepecialized
to any Hessian manifold by bounding the parameters. In later sections,
we do this for the log barrier, by defining the auxiliary function
and bounding the manifold parameters. 

\subsection{Isoperimetry}

\label{subsec:Isoperimetry}For a convex body $K$, the \emph{cross-ratio
distance} of $x$ and $y$ is
\[
d_{K}(x,y)=\frac{|x-y||p-q|}{|p-x||y-q|}
\]
where $p$ and $q$ are on the boundary of $K$ such that $p,x,y,q$
are on the straight line $\overline{xy}$ and are in order. In this
section, we show that if the distance $d(x,y)$ induced by the Riemannian
metric is upper bounded by the cross-ratio distance, then the body
has good isoperimetric constant in terms of $d$. We note that although
the cross-ratio distance is not a metric, the closely-related Hilbert
distance is a metric:

\[
d_{H}(x,y)=\log\left(1+\frac{|x-y||p-q|}{|p-x||y-q|}\right).
\]

\begin{thm}
\label{thm:iso}For a Hessian manifold $M$ with smoothness parameters
$G_{2}$, for any partition of $M$ into three measurable subsets
$S_{1},S_{2},S_{3}$, we have that

\[
\vol(S_{3})\ge\frac{d(S_{1},S_{2})}{G_{2}}\min\{\vol(S_{1}),\vol(S_{2})\}.
\]
\end{thm}
The theorem follows from the following isoperimetric inequality from
\cite{Lovasz1998}, the definition of $G_{2}$ and the fact $d_{H}\leq d_{K}$.
\begin{thm}[\cite{Lovasz1998}]
\label{thm:d_K}For any convex body $K$ and any partition of $K$
into disjoint measurable subsets $S_{1},S_{2},S_{3}$

\[
\vol(S_{3})\ge d_{K}(S_{1},S_{2})\vol(S_{1})\vol(S_{2}).
\]
\end{thm}

\subsection{$1$-step distribution}

We first derive a formula for the drift term \textemdash{} it is in
fact a classical Newton step of the volumetric barrier function $\log\det\nabla^{2}\phi(x)$.
\begin{lem}
\label{lem:drift-Newton-step}We have 
\[
\mu(x)=-\left(\nabla^{2}\phi(x)\right)^{-1}\nabla\psi(x)
\]
where $\psi(x)=\frac{1}{2}\log\det\nabla^{2}\phi(x)$. 
\end{lem}
\begin{proof}
We note that
\begin{eqnarray*}
\frac{\partial}{\partial x_{j}}\log\det\left(\nabla^{2}\phi\right)^{-1} & = & \tr\left(\left(\nabla^{2}\phi\right)\frac{\partial}{\partial x_{j}}\left(\nabla^{2}\phi\right)^{-1}\right)\\
 & = & -\tr\left(\left(\nabla^{2}\phi\right)\left(\nabla^{2}\phi\right)^{-1}\left(\frac{\partial}{\partial x_{j}}\nabla^{2}\phi\right)\left(\nabla^{2}\phi\right)^{-1}\right)\\
 & = & -\sum_{k}e_{k}^{T}\left(\frac{\partial}{\partial x_{j}}\nabla^{2}\phi\right)\left(\nabla^{2}\phi\right)^{-1}e_{k}.
\end{eqnarray*}
Hence, we have
\begin{eqnarray*}
\frac{1}{2}e_{i}^{T}\left(\nabla^{2}\phi\right)^{-1}\nabla\log\det\left(\nabla^{2}\phi\right)^{-1} & = & -\frac{1}{2}\sum_{jk}\left(\nabla^{2}\phi\right)_{ij}^{-1}e_{k}^{T}\left(\frac{\partial}{\partial x_{j}}\nabla^{2}\phi\right)\left(\nabla^{2}\phi\right)^{-1}e_{k}\\
 & = & -\frac{1}{2}\sum_{jk}e_{i}^{T}\left(\nabla^{2}\phi\right)^{-1}\left(\frac{\partial}{\partial x_{k}}\nabla^{2}\phi\right)\left(\nabla^{2}\phi\right)^{-1}e_{k}
\end{eqnarray*}
On the other hand, we have
\begin{eqnarray*}
\mu_{i} & = & \frac{1}{2}\sum_{k}\frac{\partial}{\partial x_{k}}\left(\left(\nabla^{2}\phi\right)^{-1}\right)_{ik}\\
 & = & -\frac{1}{2}\sum_{k}e_{i}^{T}\left(\nabla^{2}\phi\right)^{-1}\left(\frac{\partial}{\partial x_{k}}\nabla^{2}\phi\right)\left(\nabla^{2}\phi\right)^{-1}e_{k}.
\end{eqnarray*}
\end{proof}
To have a uniform stationary distribution, the geodesic walk uses
a Metropolis filter. The transition probability before applying the
filter is given as follows in Euclidean coordinates.
\begin{lem}
\label{lem:prob_formula} For any $x\in M$ and $h>0$, the probability
density of the 1-step distribution from $x$ (before applying the
Metropolis filter) is given by 
\begin{equation}
p_{x}(y)=\sum_{v_{x}:\exp_{x}(v_{x})=y}\det(d\exp_{x}(v_{x}))^{-1}\sqrt{\frac{\det\left(g(y)\right)}{\left(2\pi h\right)^{n}}}\exp\left(-\frac{1}{2}\norm{\frac{v_{x}-\frac{h}{2}\mu(x)}{\sqrt{h}}}_{x}^{2}\right)\label{eq:1_step_prof}
\end{equation}
where $y=\exp_{x}(v_{x})$ and $d\exp_{x}$ is the differential of
the exponential map at $x$\textup{.}
\end{lem}
\begin{proof}
We prove the formula by separately considering each $v_{x}\in T_{x}M$
s.t. $\exp_{x}(v_{x})=y$, then summing up. In the tangent space $T_{x}M$,
the point $v_{x}$ is a Gaussian step. Therefore, the probability
density of $v_{x}$ in $T_{x}M$ as follows. 
\[
p_{x}^{T_{x}M}(v_{x})=\frac{1}{\left(2\pi h\right)^{n/2}}\exp\left(-\frac{1}{2}\norm{\frac{v_{x}-\frac{h}{2}\mu(x)}{\sqrt{h}}}_{x}^{2}\right).
\]
Note that $v_{x},\mu(x)\in T_{x}M.$ Let $y=\exp_{x}(v_{x}).$ In
the tangent space $T_{y}M,$ we have that $y$ maps to $0$. Let $F:T_{x}M\rightarrow K$
defined by $F(v)=\text{id}_{M\rightarrow K}\circ\exp_{x}(v)$. Here
$K$ is the same set as $M$ but endowed with the Euclidean metric.
Hence, we have
\[
dF(v_{x})=d\text{id}_{M\rightarrow K}(y)d\exp_{x}(v_{x}).
\]
The result follows from $p_{x}(y)=\det(dF(v_{x}))^{-1}p_{x}^{T_{x}M}(v_{x})$
and
\begin{eqnarray*}
\det dF(v_{x}) & = & \det\left(d\text{id}_{M\rightarrow K}(y)\right)\det\left(d\exp_{x}(v_{x})\right)\\
 & = & \det(g(y))^{-1/2}\det\left(d\exp_{x}(v_{x})\right).
\end{eqnarray*}
\end{proof}
In Section \ref{subsec:reject_prob}, we bound the acceptance probability
of the Metropolis filter. This is a crucial aspect of the analysis.

\begin{restatable}{thm}{rejectionprob}

\label{thm:rejectionprob}Given a geodesic $\gamma$ with $\gamma(0)=x$,
$\gamma'(0)=v_{x}$, $\gamma(\ell)=y$, $\gamma'(\ell)=-v_{y}$ with
$\ell=\sqrt{nh}$. Suppose that $h\leq\min(H,\frac{1}{nR_{1}})$ and
$V(\gamma)\leq V_{0}$, then we have that 
\[
\left|\log\left(\frac{p(x\overset{v_{x}}{\rightarrow}y)}{p(y\overset{v_{y}}{\rightarrow}x)}\right)\right|=O\left(\sqrt{n}h^{3/2}D_{1}+(nh)^{3/2}G_{1}+(nhR_{1})^{2}\right).
\]

\end{restatable}

In Section \ref{sec:Smoothness-of-P}, we bound the overlap of one-step
distributions from nearby points.

\begin{restatable}{thm}{dTV}

\label{thm:dTV}For $h\leq\min\left(H,\frac{1}{10^{6}nR_{1}}\right)$,
then the one-step distributions $P_{x},P_{z}$ from $x,z$ satisfy 

\[
d_{TV}(P_{x},P_{z})=O\left(nhR_{2}+D_{2}\sqrt{h}+\frac{1}{\sqrt{h}}+nhD_{0}R_{1}\right)d(x,z)+\frac{1}{20}.
\]

\end{restatable}

Combining the above two theorems lets us bound the conductance and
mixing time of the walk, as we show in the next section.

\subsection{Conductance and mixing time}

\label{sec:mixing}
\begin{proof}[Proof of Theorem \ref{thm:gen-convergence}]
The proof follows the standard outline for geometric random walks
(see e.g., \cite{VemSurvey}). Let $Q$ be the uniform distribution
over $M$ and $S$ be any measurable subset of $M$. Then our goal
is to show that  

\[
\frac{\int_{S}P_{x}(M\setminus S)\,dQ(x)}{\min\left\{ Q(S),Q(M\setminus S)\right\} }=\Omega\left(\frac{\sqrt{h}}{G_{2}}\right).
\]
Since the Markov chain is time-reversible (For any two subsets $A,B$,
$\int_{A}P_{x}(B)\,dx=\int_{B}P_{x}(A)\,dx$), we can write the numerator
of the LHS above as 

\[
\frac{1}{2}\left(\int_{S}P_{x}(M\setminus S)\,dQ(x)+\int_{M\setminus S}P_{x}(S)\,dQ(x)\right).
\]
Define

\begin{align*}
S_{1} & =\{x\in S\,:\,P_{x}(M\setminus S)<0.05\}\\
S_{2} & =\{x\in M\setminus S\,:\,P_{x}(S)<0.05\}\\
S_{3} & =M\setminus S_{1}\setminus S_{2}.
\end{align*}
We can assume wlog that $Q(S_{1})\ge(1/2)Q(S)$ and $Q(S_{2})\ge(1/2)Q(M\setminus S)$
(if not, the conductance is $\Omega(1)$). 

Next, we note that for any two points $x\in S_{1}$ and $y\in S_{2}$,
$d_{TV}(P_{x},P_{y})>0.9$. Therefore, by Theorem \ref{thm:d_K},
we have that $d(x,y)=\Omega(\sqrt{h})$ and hence $d(S_{1},S_{2})=\Omega(\sqrt{h})$.
 Therefore, using Theorem \ref{thm:iso}, 

\[
\vol(S_{3})=\Omega\left(\frac{\sqrt{h}}{G_{2}}\right)\min\{\vol(S_{1}),\vol(S_{2})\}.
\]

Going back to the conductance, 

\begin{align*}
\frac{1}{2}\left(\int_{S}P_{x}(M\setminus S)\,dQ(x)+\int_{M\setminus S}P_{x}(S)\,dQ(x)\right) & \ge\frac{1}{2}\int_{S_{3}}(0.05)dQ(x)\\
 & =\Omega\left(\frac{\sqrt{h}}{G_{2}}\right)\min\{\vol(S_{1}),\vol(S_{2})\}\frac{1}{\vol(M)}\\
 & =\Omega\left(\frac{\sqrt{h}}{G_{2}}\right)\min\left\{ \frac{\vol(S)}{\vol(M)},\frac{\vol(M\setminus S)}{\vol(M)}\right\} \\
 & =\Omega\left(\frac{\sqrt{h}}{G_{2}}\right)\min\{Q(S),Q(M\setminus S)\}
\end{align*}
Therefore, $\phi(S)\ge\Omega\left(\frac{\sqrt{h}}{G_{2}}\right)$.
\end{proof}
\begin{cor}
\label{thm:mixing}Let $K$ be a polytope. Let $Q$ be the uniform
distribution over $K$ and $Q_{t}$ be the distribution obtained after
$t$ steps of the geodesic walk started from a distribution $Q_{0}$
with $d_{0}=\sup_{K}\frac{dQ_{0}}{dQ}$. Then after $t>O(G_{2}^{2}/h)\log\left(\frac{d_{0}}{\epsilon}\right)$
steps, with probability at least $1-\delta$, we have $d_{TV}(Q_{t},Q)\le\epsilon.$ 
\end{cor}

\subsection{Warm-up: Interval and hypercube\label{subsec:Warmup}}

The technical core of the proof is in Theorems \ref{thm:rejectionprob}
and \ref{thm:dTV}. Before we get to those, in this section, we analyze
the geodesic walk for one-dimensional Hessian manifolds and the hypercube
(product of intervals), which we can do by more elementary methods. 
\begin{lem}
Given a barrier function $\phi$ on one dimension interval $(\alpha,\beta)$
such that $\left|\phi^{(3)}(x)\right|=O(\phi''(x)^{3/2})$, $\left|\phi^{(4)}(x)\right|=O(\phi''(x)^{4/2})$
and $\left|\phi^{(5)}(x)\right|=O(\phi''(x)^{5/2})$ for $x\in(\alpha,\beta)$.
For step size $h$ smaller than some constant, the geodesic walk on
the Hessian manifold $(\alpha,\beta)$ induced by $\phi$ satisfies
that
\[
\left|\log\left(\frac{p(x\rightarrow y)}{p(y\rightarrow x)}\right)\right|=O(h^{3/2})
\]
with constant probability.
\end{lem}
The first assumption above $\left|\phi^{(3)}(x)\right|=O(\phi''(x)^{3/2})$
is called self-concordance, and the others can be viewed as its extensions
to higher derivatives. It is easy to check that the logarithmic barrier
$\phi(x)=-\log(1+x)-\log(1-x)$ on $(-1,1)$ satisfies all the assumptions.
Our main theorem for general manifolds (Theorem \ref{thm:rejectionprob})
implies the same bound after substituting the parameters for the log
barrier on the interval with $n=1$ and $m=2$. 
\begin{proof}
Since the walk is affine-invariant and the conditions are scale-invariant,
we can assume the domain is the interval $(-1,1)$. Let the metric
$p(x)=\phi''(x)$. The drift is given by (\ref{eq:drift_term_general}):
\[
\mu(x)=\frac{1}{2}\frac{d}{dx}\left(\frac{1}{p(x)}\right)=\frac{-p'(x)}{2p^{2}(x)}.
\]
Let $f(x)\defeq\int_{0}^{x}\sqrt{p(t)}dt$ be the mapping from $(-1,1)$
to $\R$ such that for any $x,y\in(-1,1)$, we have that $|f(x)-f(y)|$
is the manifold distance on $(-1,1)$ using the metric $p$. In particular,
we have that
\[
\exp_{x}(v)=f^{-1}\left(f(x)+p^{1/2}(x)v\right).
\]
Therefore, the probability density on $\R$ (as in Lemma \ref{lem:prob_formula}
for the one-dimensional case) is given by
\[
p^{\R}(x\rightarrow y)=\frac{1}{\sqrt{2\pi h}}\exp\left[-\frac{1}{2h}\left(f(x)-\frac{hp'(x)}{4p^{3/2}(x)}-f(y)\right)^{2}\right].
\]
Hence, 
\[
p(x\rightarrow y)=\sqrt{\frac{p(y)}{2\pi h}}\exp\left[-\frac{1}{2h}\left(f(x)-\frac{hp'(x)}{4p^{3/2}(x)}-f(y)\right)^{2}\right].
\]
Hence, we have that
\begin{eqnarray*}
 &  & \frac{p(x\rightarrow y)}{p(y\rightarrow x)}\\
 & = & \sqrt{\frac{p(y)}{p(x)}}\exp\left[\frac{1}{2h}\left(f(y)-f(x)-\frac{hp'(y)}{4p^{3/2}(y)}\right)^{2}-\frac{1}{2h}\left(f(y)-f(x)+\frac{hp'(x)}{4p^{3/2}(x)}\right)^{2}\right]\\
 & = & \exp\left[\frac{\log(p(y))-\log(p(x))}{2}-\frac{f(y)-f(x)}{4}\left(\frac{p'(x)}{p^{3/2}(x)}+\frac{p'(y)}{p^{3/2}(y)}\right)+\frac{h}{32}\left(\frac{(p'(y))^{2}}{p^{3}(y)}-\frac{(p'(x))^{2}}{p^{3}(x)}\right)\right].
\end{eqnarray*}
Note that
\begin{eqnarray*}
\log(p(y))-\log(p(x)) & = & \log(p(f^{-1}(f(y))))-\log(p(f^{-1}(f(x))))\\
 & = & \int_{f(x)}^{f(y)}\frac{p'(f^{-1}(t))}{p(f^{-1}(t))f'(f^{-1}(t))}dt\\
 & = & \int_{f(x)}^{f(y)}\frac{p'(f^{-1}(t))}{p^{3/2}(f^{-1}(t))}dt.
\end{eqnarray*}
Note that for any second differentiable function $\phi$, Lemma \ref{lem:trapezoidal}
shows that that
\[
\left|\int_{\beta}^{\alpha}\phi(t)dt-(\alpha-\beta)\left(\phi(\alpha)+\phi(\beta)\right)\right|\leq\frac{\left|\alpha-\beta\right|^{3}}{12}\max_{\alpha\leq t\leq\beta}\left|\phi''(t)\right|.
\]
Hence, 
\begin{eqnarray*}
 &  & \left|\log(p(y))-\log(p(x))-\frac{f(y)-f(x)}{2}\left(\frac{p'(x)}{p^{3/2}(x)}+\frac{p'(y)}{p^{3/2}(y)}\right)\right|\\
 & \leq & \frac{\left|f(x)-f(y)\right|^{3}}{12}\max_{f(x)\leq t\leq f(y)}\left|\frac{d^{2}}{dt^{2}}\frac{p'(f^{-1}(t))}{p^{3/2}(f^{-1}(t))}\right|.
\end{eqnarray*}

For the other term, we note that
\[
\left|\frac{(p'(y))^{2}}{p^{3}(y)}-\frac{(p'(x))^{2}}{p^{3}(x)}\right|\leq\left|y-x\right|\max_{x\leq t\leq y}\left|\frac{d}{dt}\frac{(p'(t))^{2}}{p^{3}(t)}\right|.
\]
Hence, 
\begin{eqnarray*}
 &  & \left|\log\left(\frac{p(x\rightarrow y)}{p(y\rightarrow x)}\right)\right|\\
 & \leq & O\left(\left|f(x)-f(y)\right|^{3}\max_{f(x)\leq t\leq f(y)}\left|\frac{d^{2}}{dt^{2}}\frac{p'(f^{-1}(t))}{p^{3/2}(f^{-1}(t))}\right|\right)+O\left(h\left|y-x\right|\max_{x\leq t\leq y}\left|\frac{d}{dt}\frac{(p'(t))^{2}}{p^{3}(t)}\right|\right).
\end{eqnarray*}

Note that
\[
\frac{d}{dt}\frac{p'(f^{-1}(t))}{p^{3/2}(f^{-1}(t))}=-\frac{3}{2}\frac{(p'(f^{-1}(t)))^{2}}{p^{3}(f^{-1}(t))}+\frac{p''(f^{-1}(t))}{p^{2}(f^{-1}(t))}
\]
and
\[
\frac{d^{2}}{dt^{2}}\frac{p'(f^{-1}(t))}{p^{3/2}(f^{-1}(t))}=\frac{9}{2}\frac{(p'(f^{-1}(t)))^{3}}{p^{9/2}(f^{-1}(t))}-\frac{3p'(f^{-1}(t))p''(f^{-1}(t))}{p^{7/2}(f^{-1}(t))}+\frac{p'''(f^{-1}(t))}{p^{5/2}(f^{-1}(t))}.
\]
Also, we have that
\[
\frac{d}{dt}\frac{(p'(t))^{2}}{p^{3}(t)}=\frac{2p'(t)p''(t)}{p^{3}(t)}-\frac{3(p'(t))^{3}}{p^{4}(t)}.
\]
Hence, we have that
\begin{eqnarray*}
\left|\log\left(\frac{p(x\rightarrow y)}{p(y\rightarrow x)}\right)\right| & = & O\left(\left|f(x)-f(y)\right|^{3}\max_{x\leq t\leq y}\left(\left|\frac{(p'(t))^{3}}{p^{9/2}(t)}\right|+\left|\frac{p'(t)p''(t)}{p^{7/2}(t)}\right|+\left|\frac{p'''(t)}{p^{5/2}(t)}\right|\right)\right)\\
 &  & +O\left(h\left|y-x\right|\max_{x\leq t\leq y}\left(\left|\frac{p'(t)p''(t)}{p^{3}(t)}\right|+\left|\frac{(p'(t))^{3}}{p^{4}(t)}\right|\right)\right).
\end{eqnarray*}
Since $\phi$ is self-concordant, i.e.,$\left|\phi^{(3)}(x)\right|=O(\phi''(x)^{3/2})$,
we know that $p(t)\leq O(1)p(s)$ for all $s,t\in[x,y]$ if $p(x)|x-y|$
is smaller than a constant. Noting that $p(x)|x-y|^{2}=\Theta(h)$
with constant probability, we have that $p(t)\leq O(1)p(s)$ for all
$s,t\in[x,y]$. Therefore, we have that
\begin{align*}
\left|f(x)-f(y)\right| & =O(|x-y|\sqrt{p(x)})=O(\sqrt{h}),\\
|x-y| & =O(p^{-1/2}(x)\sqrt{h}).
\end{align*}
Hence, we have
\begin{eqnarray*}
 &  & \left|\log\left(\frac{p(x\rightarrow y)}{p(y\rightarrow x)}\right)\right|\\
 & = & O\left(h^{3/2}\max_{x\leq t\leq y}\left(\left|\frac{(p'(t))^{3}}{p^{9/2}(t)}\right|+\left|\frac{p'(t)p''(t)}{p^{7/2}(t)}\right|+\left|\frac{p'''(t)}{p^{5/2}(t)}\right|\right)\right)+O\left(h^{3/2}\max_{x\leq t\leq y}\left(\left|\frac{p'(t)p''(t)}{p^{7/2}(t)}\right|+\left|\frac{(p'(t))^{3}}{p^{9/2}(t)}\right|\right)\right).
\end{eqnarray*}
Using our assumption on $\phi$, we get our result.
\end{proof}
\begin{rem*}
The hypercube in $\R^{n}$ is a product of intervals. In fact, the
next-step density function is a product function, and the quantity
$\log\left(\frac{p(x\rightarrow y)}{p(y\rightarrow x)}\right)$ is
a sum over the same quantity in each coordinate. Viewing the coordinates
as independent processes, this is a sum of independent random variables,
w.h.p. of magnitude $O(h^{\frac{3}{2}}).$ The signs are random since
ratio is symmetric in $x,y.$ Thus, the overall log ratio is $O\left(\sqrt{n}h^{\frac{3}{2}}\right).$
To keep this bounded, it suffices to choose $h=O\left(n^{-\frac{1}{3}}\right).$
Moreover, since the isoperimetric ratio ($G_{2}$ in our parametrization)
is $O\left(1\right)$ for the hypercube, this gives an overall mixing
time of $O\left(n^{\frac{1}{3}}\right)$, much lower than the mixing
time of $n^{2}$ for the ball walk and hit-and-run, or the current
bound for the Dikin walk.
\end{rem*}

\subsection{Rejection probability\label{subsec:reject_prob}}

The goal of this section is to prove Theorem \ref{thm:rejectionprob},
i.e., the rejection probability of the Metropolis filter is small.
For a transition from $x$ to $y$ on the manifold, the filter is
applied with respect to a randomly chosen $v_{x},$ i.e., for one
geodesic from $x$ to $y$. We will bound the ratio of the transition
probabilities (without the filter) as follows: 
\begin{eqnarray*}
\log\left(\frac{p(x\overset{v_{x}}{\rightarrow}y)}{p(y\overset{v_{y}}{\rightarrow}x)}\right) & = & \log\left(\frac{\det\left(d\exp_{x}(v_{x})\right)^{-1}}{\det\left(d\exp_{y}(v_{y})\right)^{-1}}\right)+\frac{1}{2}\log\det\left(g(y)\right)-\frac{1}{2}\log\det\left(g(x)\right)\\
 &  & -\frac{1}{2}\norm{\frac{v_{x}-\frac{h}{2}\mu(x)}{\sqrt{h}}}_{x}^{2}+\frac{1}{2}\norm{\frac{v_{y}-\frac{h}{2}\mu(y)}{\sqrt{h}}}_{y}^{2}.
\end{eqnarray*}
Since a geodesic has constant speed, we have $\norm{v_{x}}_{x}=\norm{v_{y}}_{y}$.
Therefore, we have that
\begin{eqnarray}
\log\left(\frac{p(x\overset{v_{x}}{\rightarrow}y)}{p(y\overset{v_{y}}{\rightarrow}x)}\right) & = & \log\left(\frac{\det\left(d\exp_{y}(v_{y})\right)}{\det\left(d\exp_{x}(v_{x})\right)}\right)+\frac{1}{2}\log\det\left(g(y)\right)-\frac{1}{2}\log\det\left(g(x)\right)\label{eq:formula_ratio}\\
 &  & +\frac{1}{2}\left\langle v_{x},\mu(x)\right\rangle _{x}-\frac{h}{8}\norm{\mu(x)}_{x}^{2}-\frac{1}{2}\left\langle v_{y},\mu(y)\right\rangle _{y}+\frac{h}{8}\norm{\mu(y)}_{y}^{2}.\nonumber 
\end{eqnarray}
We separate the proof into three parts:
\begin{itemize}
\item $\left|\norm{\mu(x)}_{x}^{2}-\norm{\mu(y)}_{y}^{2}\right|\le\sqrt{nh}D_{1}$,
immediate from the definition of $D_{1}$ and the parameterization
$\gamma:[0,\ell\defeq\sqrt{nh}]\rightarrow M$. 
\item Sec \ref{sec:Trapezoidal_rule_volumetric_barrier}: $\left|\log\det\left(g(y)\right)-\log\det\left(g(x)\right)+\left\langle v_{x},\mu(x)\right\rangle _{x}-\left\langle v_{y},\mu(y)\right\rangle _{y}\right|=O((nh)^{3/2}G_{1})$,
\item Sec \ref{sec:smooth_exp}: $\left|\log\left(\frac{\det\left(d\exp_{y}(v_{y})\right)}{\det\left(d\exp_{x}(v_{x})\right)}\right)\right|=O((nhR_{1})^{2})$.
\end{itemize}
Together, these facts imply Theorem \ref{thm:rejectionprob}.

\subsubsection{Trapezoidal rule }

\label{sec:Trapezoidal_rule_volumetric_barrier}Recall that the trapezoidal
rule is to approximate $\int_{0}^{h}f(t)dt$ by $\frac{h}{2}\left(f(0)+f(h)\right)$.
The nice thing about this rule is that the error is $O(h^{3})$ instead
of $O(h^{2})$ because the second order term cancels by symmetry.
Our main observation here is that the geodesic walk implicitly follows
a trapezoidal rule on the metric and hence it has a small error. We
include the proof of the trapezoidal rule error for completeness.
\begin{lem}
\label{lem:trapezoidal}We have that
\[
\left|\int_{0}^{\ell}f(t)dt-\frac{\ell}{2}\left(f(0)+f(\ell)\right)\right|\leq\frac{\ell^{3}}{12}\max_{0\leq t\leq\ell}\left|f''(t)\right|.
\]
\end{lem}
\begin{proof}
Note that
\begin{eqnarray*}
\int_{0}^{\ell}f(t)dt-\frac{\ell}{2}\left(f(0)+f(\ell)\right) & = & \int_{0}^{\ell}\left(f(0)+\int_{0}^{t}f'(s)ds\right)dt-\ell f(0)-\frac{\ell}{2}\int_{0}^{\ell}f'(s)ds\\
 & = & \int_{0}^{\ell}\int_{0}^{t}f'(s)dsdt-\frac{\ell}{2}\int_{0}^{\ell}f'(s)ds\\
 & = & \int_{0}^{\ell}(\frac{\ell}{2}-s)f'(s)ds\\
 & = & \int_{0}^{\ell}(\frac{\ell}{2}-s)\left(f'(\frac{\ell}{2})+\int_{\ell/2}^{s}f''(t)dt\right)ds\\
 & = & \int_{0}^{\ell}(\frac{\ell}{2}-s)\int_{\ell/2}^{s}f''(t)dtds\\
 & \leq & \frac{\ell^{3}}{12}\max_{0\leq t\leq\ell}\left|f''(t)\right|.
\end{eqnarray*}
\end{proof}
We apply this to the logdet function.
\begin{lem}
Let $f(t)=\log\det(g(\gamma(t))$ where $\gamma(t)=\exp_{x}(\frac{t}{\ell}v_{x})$.
If $h\leq H$ and $V(\gamma)\leq V_{0}$, we have 
\[
\left|\log\det(g(y))-\log\det(g(x))+\left\langle v_{x},\mu(x)\right\rangle _{x}-\left\langle v_{y},\mu(y)\right\rangle _{y}\right|=O((nh)^{3/2}G_{1}).
\]
\end{lem}
\begin{proof}
Let $f(t)=\log\det g(\gamma(t))$. By Lemma \ref{lem:drift-Newton-step},
$\mu(\gamma(t))=-\frac{1}{2}g(\gamma(t))^{-1}\nabla f(\gamma(t)).$
Using this, 

\begin{align*}
f'(t) & =\langle\nabla_{\gamma(t)}\log\det g(\gamma(t)),\gamma'(t)\rangle_{2}\\
 & =\langle g(\gamma(t))^{-1}\nabla_{\gamma(t)}\log\det g(\gamma(t)),\gamma'(t)\rangle_{\gamma(t)}\\
 & =-\langle2\mu(\gamma(t)),\gamma'(t)\rangle_{\gamma(t)}.
\end{align*}
Noting that $v_{x}=\ell\gamma'(0)$ and $v_{y}=-\ell\gamma'(\ell)$,
and using Lemma \ref{lem:trapezoidal}, we have
\begin{eqnarray*}
 &  & \left|\log\det(g(y))-\log\det(g(x))+\left\langle v_{x},\mu(x)\right\rangle _{x}-\left\langle v_{y},\mu(y)\right\rangle _{y}\right|\\
 & = & \left|\log\det(g(y))-\log\det(g(x))+\ell\left(\left\langle \gamma'(0),\mu(x)\right\rangle _{x}+\left\langle \gamma'(\ell),\mu(y)\right\rangle _{y}\right)\right|\\
 & = & \left|\int_{0}^{\ell}f'(t)dt-\frac{\ell}{2}\left(f'(0)+f'(\ell)\right)\right|\\
 & \leq & \frac{\ell^{3}}{12}\max_{0\leq t\leq\ell}\left|f'''(t)\right|=O\left((nh)^{3/2}G_{1}\right).
\end{eqnarray*}
\end{proof}

\subsubsection{Smoothness of exponential map\label{sec:smooth_exp}}

First, we show the relation between the differential of exponential
map $d\exp_{x}(v_{x})$ and the Jacobi field along the geodesic $\exp_{x}(tv_{x})$.
This can be viewed as the fundamental connection between geodesics
and the Jacobi field in matrix notation, which will be convenient
for our purpose. 
\begin{lem}
\label{lem:formula_dexp}Given a geodesic $\gamma(t)=\exp_{x}(\frac{t}{\ell}v_{x})$,
let $\{X_{i}(t)\}_{i=1}^{n}$ be the parallel transport of some orthonormal
frame along $\gamma(t)$. Then, for any $w\in\Rn$, we have that 
\[
d\exp_{x}(v_{x})\left(\sum w_{i}X_{i}(0)\right)=\sum_{i}\psi_{w}(\ell)_{i}X_{i}(\ell)
\]
 where $\psi_{w}$ satisfies the following Jacobi equation along $\gamma(t)$:
\begin{eqnarray}
\frac{d^{2}}{dt^{2}}\psi_{w}(t)+R(t)\psi_{w}(t) & = & 0\text{ for all }0\leq t\leq\ell,\nonumber \\
\frac{d}{dt}\psi_{w}(0) & = & \frac{w}{\ell},\label{eq:jacobian}\\
\psi_{w}(0) & = & 0,\nonumber 
\end{eqnarray}
and $R(t)$ is defined in Definition \ref{def:Rt_definition}.
\end{lem}
\begin{proof}
We want to compute $d\exp_{x}(v_{x})(w)$ for some $w\in T_{y}M$.
By definition, we have that
\[
d\exp_{x}(v_{x})(w)=\frac{d}{ds}\gamma(t,s)|_{t=\ell,s=0}
\]
where $\gamma(t,s)=\exp_{x}(tv_{x}/\ell+sw)$. It is known that $\eta_{w}(t)=\frac{d}{ds}\gamma(t,s)|_{s=0}$
is a Jacobi field given by the formula (Sec \ref{subsec:Jacobi-field})
\begin{eqnarray*}
D_{t}D_{t}\eta_{w}+R(\eta_{w},\gamma'(t))\gamma'(t) & = & 0,\text{ for }0\leq t\leq\ell\\
D_{t}\eta_{w}(0) & = & w/\ell,\\
\eta_{w}(0) & = & 0.
\end{eqnarray*}

Recall that the parallel transport of an orthonormal frame remains
orthonormal because 
\[
\frac{d}{dt}\left\langle X_{i}(t),X_{j}(t)\right\rangle _{\gamma(t)}=\left\langle D_{t}X_{i}(t),X_{j}(t)\right\rangle _{\gamma(t)}+\left\langle X_{i}(t),D_{t}X_{j}(t)\right\rangle _{\gamma(t)}=0.
\]
Since $X_{i}(t)$ is an orthonormal basis at $T_{\gamma(t)}M$, we
can write
\[
\eta_{w}(t)=\sum_{i}\psi_{w,i}(t)X_{i}(t).
\]
Now, we need to verify $\psi_{w}$ satisfies (\ref{eq:jacobian}).
The second and last equation is immediate. 

Since $\eta_{w}$ satisfies the ODE above, we have 
\[
D_{t}D_{t}\sum_{i}\psi_{w,i}(t)X_{i}(t)+R(\sum_{i}\psi_{w,i}(t)X_{i}(t),\gamma'(t))\gamma'(t)=0.
\]
Since $X_{i}(t)$ is a parallel transport, we have that $D_{t}D_{t}\left(\psi_{w,i}(t)X_{i}(t)\right)=\frac{d^{2}}{dt^{2}}\psi_{w,i}(t)X_{i}(t)$
and hence 
\[
\sum_{i}\frac{d^{2}}{dt^{2}}\psi_{w,i}(t)X_{i}(t)+\sum_{i}\psi_{w,i}(t)R(X_{i}(t),\gamma'(t))\gamma'(t)=0
\]
Let $R_{ij}(t)=\left\langle R(X_{i}(t),\gamma'(t))\gamma'(t),X_{j}(t)\right\rangle $.
Since $R(t)$ is a symmetric matrix (Fact \ref{fact:sym_cuv}), we
have
\[
\frac{d^{2}}{dt^{2}}\psi_{w}(t)+R(t)\psi_{w}(t)=0.
\]
This verifies the first equation in (\ref{eq:jacobian}).
\end{proof}
Next, we have an elementary lemma about the determinant.
\begin{lem}
\label{lem:logdet_est}Suppose that $E$ is a matrix (not necessarily
symmetric) with $\norm E_{2}\leq\frac{1}{4}$, we have
\[
\left|\log\det(I+E)-\tr E\right|\leq\norm E_{F}^{2}.
\]
\end{lem}
\begin{proof}
Let $f(t)=\log\det(I+tE)$. Then, by Jacobi's formula, we have
\begin{eqnarray*}
f'(t) & = & \tr\left((I+tE)^{-1}E\right),\\
f''(t) & = & -\tr((I+tE)^{-1}E(I+tE)^{-1}E).
\end{eqnarray*}
Since $\norm E_{2}\leq\frac{1}{4}$, we have that $\norm{(I+tE)^{-1}}_{2}\leq\frac{4}{3}$
and hence
\begin{eqnarray*}
\left|f''(t)\right| & = & \left|\tr((I+tE)^{-1}E(I+tE)^{-1}E)\right|\\
 & \leq & \left|\tr(E^{T}\left((I+tE)^{-1}\right)^{T}(I+tE)^{-1}E)\right|\\
 & \leq & 2\left|\tr(E^{T}E)\right|=2\norm E_{F}^{2}.
\end{eqnarray*}
The result follows from 
\begin{eqnarray*}
f(1) & = & f(0)+f'(0)+\int_{0}^{1}(1-s)f''(s)ds\\
 & = & \tr(E)+\int_{0}^{1}(1-s)f''(s)ds.
\end{eqnarray*}
\end{proof}
Using lemma \ref{lem:logdet_est} along with bounds on the solution
to the Jacobi field equation (Lemma \ref{lem:jacobi_approx_sol_mat}
in Sec. \ref{subsec:Approx_Jacobi}) , we have the following: 
\begin{lem}
\label{lem:Jac-approx}Given a geodesic walk $\gamma(t)=\exp_{x}(\frac{t}{\ell}v_{x})$
with step size $h$ satisfying $0<h\leq\frac{1}{nR}$ where $R=\max_{0\leq t\leq\ell}\norm{R(t)}_{F}$.
We have that $d\exp_{x}(v_{x})$ is invertible,
\begin{equation}
\left|\log\det\left(d\exp_{x}(v_{x})\right)-\int_{0}^{\ell}\frac{s(\ell-s)}{\ell}\text{Ric}(\gamma'(s))ds\right|\leq\frac{(nhR)^{2}}{6}\label{eq:Jac_approx_1}
\end{equation}
and
\[
\left|\log\det\left(d\exp_{x}(v_{x})\right)-\log\det\left(d\exp_{y}(v_{y})\right)\right|\leq\frac{(nhR)^{2}}{3}.
\]
If we further assume $V(\gamma)\leq V_{0}$ and $h\leq H$, then we
have that $R\leq R_{1}$.
\end{lem}
\begin{proof}
Let $\Psi$ be the solution of the ODE $\Psi''(t)+R(t)\Psi(t)=0$,
$\Psi'(0)=I/\ell$ and $\Psi(0)=0$. We know that $\norm{R(t)}_{2}\leq\norm{R(t)}_{F}\leq R$
for all $0\leq t\leq\ell$. Hence, Lemma \ref{lem:jacobi_approx_sol_mat}
shows that 
\begin{eqnarray}
\norm{\Psi(t)-I}_{F} & \leq & \max_{0\leq s\leq\ell}\norm{R(s)}_{F}\left(\frac{\ell^{3}}{5}\norm{I/\ell}_{2}\right)\nonumber \\
 & \leq & \frac{1}{5}nhR\leq\frac{1}{5}\label{eq:estimate_psi_error}
\end{eqnarray}
By the Lemma \ref{lem:logdet_est}, we have that
\begin{equation}
\left|\log\det(\Psi(\ell))-\tr\left(\Psi(\ell)-I\right)\right|\leq\left(\frac{1}{5}nhR\right)^{2}.\label{eq:log_psi_est_1}
\end{equation}

Now, we need to estimate $\tr(\Psi(\ell)-I)$. Note that
\begin{eqnarray*}
\Psi(\ell) & = & \Psi(0)+\Psi'(0)\ell+\int_{0}^{\ell}(\ell-s)R(s)\Psi(s)ds\\
 & = & I+\int_{0}^{\ell}(\ell-s)R(s)\Psi(s)ds.
\end{eqnarray*}
Hence, we have
\begin{eqnarray*}
\Psi(\ell)-I-\int_{0}^{\ell}\frac{s(\ell-s)}{\ell}R(s)ds & = & \int_{0}^{\ell}(\ell-s)R(s)\left(\Psi(s)-\frac{s}{\ell}I\right)ds.
\end{eqnarray*}
Using Lemma (\ref{eq:estimate_psi_error}), we have
\begin{eqnarray}
\left|\tr\left(\Psi(\ell)-I-\int_{0}^{\ell}\frac{s(\ell-s)}{\ell}R(s)ds\right)\right| & \leq & \int_{0}^{\ell}(\ell-s)\left|\tr R(s)\left(\Psi(s)-\frac{s}{\ell}I\right)\right|ds\nonumber \\
 & \leq & \int_{0}^{\ell}(\ell-s)\norm{R(s)}_{F}\norm{\Psi(s)-\frac{s}{\ell}I}_{F}ds\nonumber \\
 & \leq & \frac{\ell^{2}}{2}\cdot R\cdot\frac{1}{5}Rnh\leq\frac{(nhR)^{2}}{10}.\label{eq:log_psi_est_2}
\end{eqnarray}
Combining (\ref{eq:log_psi_est_1}) and (\ref{eq:log_psi_est_2}),
we have that
\[
\left|\log\det(\Psi(\ell))-\int_{0}^{\ell}\frac{s(\ell-s)}{\ell}R(s)ds\right|\leq\left(\frac{1}{5}Rnh\right)^{2}+\frac{(nhR)^{2}}{10}\leq\frac{\left(nhR\right)^{2}}{6}.
\]

By Lemma \ref{lem:formula_dexp}, for any $w$, we have that $d\exp_{x}(v_{x})\left(\sum w_{i}X_{i}(0)\right)=\sum_{i}\psi_{w}(\ell)_{i}X_{i}(\ell)=\sum_{i}(\Psi(\ell)w)_{i}X_{i}(\ell)$.
Since $\{X_{i}(t)\}_{i=1}^{n}$ are orthonormal, this shows that 
\[
d\exp_{x}(v_{x})=X(\ell)\Psi(\ell)X(0)^{T}
\]
where $X$ is the matrix $[X_{1},X_{2},\cdots,X_{n}]$. Since $\norm{\Psi(\ell)-I}_{2}\leq\frac{1}{5}$
(\ref{eq:estimate_psi_error}), we have that $\Psi(\ell)$ is invertible
and so is $d\exp_{x}(v_{x})$. 

Since $X(\ell)$ and $X(0)$ are orthonormal, we have that
\[
\log\det\left(d\exp_{x}(v_{x})\right)=\log\det\Psi(\ell).
\]
Therefore, this gives 
\[
\left|\log\det\left(d\exp_{x}(v_{x})\right)-\int_{0}^{\ell}\frac{s(\ell-s)}{\ell}\tr R(s)ds\right|\leq\frac{(nhR)^{2}}{6}.
\]
By the definition of Ricci curvature and Fact \ref{fact:sym_cuv},
we have that 
\begin{eqnarray*}
\tr R(s) & = & \sum_{i}\left\langle R(X_{i}(s),\gamma'(s))\gamma'(s),X_{j}(s)\right\rangle \\
 & = & \text{Ric}(\gamma'(s)).
\end{eqnarray*}
This gives the result (\ref{eq:Jac_approx_1}).

Since the geodesic $\exp_{x}(tv_{x})$ is the same as the geodesic
$\exp_{y}(tv_{y})$ except for swapping the parameterization, and
since $\int_{0}^{\ell}\frac{s(\ell-s)}{\ell}\tr R(s)ds$ is invariant
under this swap, (\ref{eq:Jac_approx_1}) implies that $\log\det\left(d\exp_{x}(v_{x})\right)$
is close to $\log\det\left(d\exp_{y}(v_{y})\right)$.
\end{proof}

\subsection{Smoothness of one-step distributions}

\label{sec:Smoothness-of-P}Here we prove Theorem \ref{thm:d_K}.
Recall that the probability density of going from $x$ to $y$ is
given by the following formula:
\[
p_{x}(y)=\sum_{v_{x}:exp_{x}(v_{x})=y}\det\left(d\exp_{x}(v_{x})\right)^{-1}\sqrt{\frac{\det\left(g(y)\right)}{\left(2\pi h\right)^{n}}}\exp\left(-\frac{1}{2}\norm{\frac{v_{x}-\frac{h}{2}\mu(x)}{\sqrt{h}}}_{x}^{2}\right)
\]
To simplify the calculation, we apply Lemma \ref{lem:Jac-approx}
and consider the following estimate of $p_{x}(y)$ instead
\begin{equation}
\tilde{p}_{x}(y)=\sum_{v_{x}:exp_{x}(v_{x})=y}\sqrt{\frac{\det\left(g(y)\right)}{\left(2\pi h\right)^{n}}}\exp\left(-\int_{0}^{\ell}\frac{t(\ell-t)}{\ell}\text{Ric}(\gamma_{v_{x}}'(t))dt-\frac{1}{2}\norm{\frac{v_{x}-\frac{h}{2}\mu(x)}{\sqrt{h}}}_{x}^{2}\right)\label{eq:p_tilde_x}
\end{equation}
where $\gamma_{v_{x}}(t)$ be the geodesic from $\gamma_{v_{x}}(0)=x$
with $\gamma_{v_{x}}'(0)=v_{x}$. Lemma \ref{lem:Jac-approx} shows
that $\left|\log\left(\tilde{p}_{x}(y)/p_{x}(y)\right)\right|$ is
small and hence it suffices to prove the smoothness of $\tilde{p}_{x}(y)$.

Let $c(s)$ be an unit speed geodesic going from $x$ to some point
$z$ very close to $x$. Lemma \ref{lem:one_one_cor} shows that there
is an unique vector field $v(s)$ on $c(s)$ such that $\exp_{c(s)}(v(s))=y$
and $v(0)=v_{x}$. Now, we define 
\[
\zeta(v,s)=\sqrt{\frac{\det\left(g(y)\right)}{\left(2\pi h\right)^{n}}}\exp\left(-\int_{0}^{\ell}\frac{t(\ell-t)}{\ell}\text{Ric}(\gamma_{s}'(t))dt-\frac{1}{2}\norm{\frac{v(s)-\frac{h}{2}\mu(c(s))}{\sqrt{h}}}_{c(s)}^{2}\right)
\]
where $\gamma_{s}(t)=\exp_{c(s)}(\frac{t}{\ell}v(s))$. Then, we have
that
\[
\tilde{p}_{c(s)}(y)=\sum_{v:\exp_{c(s)}(v(s))=y}\zeta(v,s)
\]
and hence
\begin{equation}
\frac{d}{ds}\tilde{p}_{c(s)}(y)=\sum_{v:\exp_{c(s)}(v(s))=y}\left(-\int_{0}^{\ell}\frac{t(\ell-t)}{\ell}\frac{d}{ds}\text{Ric}(\gamma'_{s}(t))dt-\frac{1}{2}\frac{d}{ds}\norm{\frac{v(s)-\frac{h}{2}\mu(c(s))}{\sqrt{h}}}_{c(s)}^{2}\right)\zeta(v,s).\label{eq:foruma_dpds}
\end{equation}
Hence, it suffices to bound the terms in parenthesis. 

In Lemma \ref{lem:smoothness_logdetJ}, we analyze $\frac{d}{ds}\text{Ric}(\gamma'_{s}(t))$
and prove that
\[
\left|\int_{0}^{\ell}\frac{t(\ell-t)}{\ell}\frac{d}{ds}\text{Ric}(\gamma'_{s}(t))dt\right|=O\left(nhR_{2}\right).
\]
In Lemma \ref{lem:normvmuds}, we analyze $\frac{d}{ds}\norm{\frac{v(s)-\frac{h}{2}\mu(c(s))}{\sqrt{h}}}_{c(s)}^{2}$
and prove that 
\[
\E_{V(\gamma)\leq V_{0}}\left|\left.\frac{d}{ds}\norm{\frac{v(s)-\frac{h}{2}\mu(c(s))}{\sqrt{h}}}_{c(s)}^{2}\right|_{s=0}\right|=O\left(D_{2}\sqrt{h}+\frac{1}{\sqrt{h}}+nhD_{0}R_{1}\right)
\]
where $\gamma$ is a random geodesic walk starting from $x$ (before
the filtering step). This implies Thm. \ref{thm:dTV}, restated below
for convenience. 

\dTV*
\begin{proof}
Given $z$ such that $d(x,z)<\frac{V_{0}}{4V_{1}}$. Let $c(s)$ be
an unit speed minimal length geodesic connecting $x$ and $z$, $\tilde{p}_{c(s)}$
is defined by (\ref{eq:p_tilde_x}). 

By (\ref{eq:gamma_good}), with probability $1-\frac{V_{0}}{100V_{1}}$
in $y$, we have that $V(\gamma)\leq\frac{1}{2}V_{0}$. Let $Y$ be
the set of $y$ such that $V(\gamma)\leq\frac{1}{2}V_{0}$. Since
the distance from $x$ to $z$ is less than $\frac{V_{0}}{4V_{1}}$
and $V(\gamma)\leq\frac{1}{2}V_{0}$ for those $y$, Lemma \ref{lem:Jac-approx}
shows there is a family of geodesic $\gamma_{s}$ which connects $c(s)$
to $y$. Furthermore, we have that $V(\gamma_{s})\leq V_{0}$. 

For any $y\in Y$, we have that $V(\gamma_{s})\leq V_{0}.$ By Lemma
\ref{lem:Jac-approx}, we have that
\[
\exp\left(\frac{-1}{6}(nhR_{1})^{2}\right)p_{c(s)}(y)\leq\tilde{p}_{c(s)}(y)\leq\exp\left(\frac{1}{6}(nhR_{1})^{2}\right)p_{c(s)}(y).
\]
Using our assumption on $h$, we have that 
\[
\exp\left(3(nhR_{1})^{2}\right)\leq C\defeq1.01.
\]
Therefore, we have that
\begin{align*}
p_{x}(y)-p_{z}(y)\leq & (1-C^{-2})p_{x}(y)+C^{-2}p_{x}(y)-C^{-1}\tilde{p}_{z}(y)\\
\leq & (1-C^{-2})p_{x}(y)+C^{-1}(\tilde{p}_{x}(y)-\tilde{p}_{z}(y)).
\end{align*}
Similarly, we have
\begin{align*}
p_{x}(y)-p_{z}(y)\geq & (1-C^{2})p_{x}(y)+C(\tilde{p}_{x}(y)-\tilde{p}_{z}(y)).
\end{align*}
Since $\int_{Y}p_{x}(y)dy\geq1-\frac{V_{0}}{100V_{1}}$, we have that
\begin{align}
d_{TV}(p_{x},p_{z}) & \leq\frac{V_{0}}{100V_{1}}+\int_{Y}\left|p_{x}(y)-p_{z}(y)\right|dy\nonumber \\
 & \leq\frac{V_{0}}{100V_{1}}+\frac{1}{50}\int_{Y}\left|p_{x}(y)\right|dy+2\int\left|\tilde{p}_{x}(y)-\tilde{p}_{z}(y)\right|dy\nonumber \\
 & \leq\frac{V_{0}}{20V_{1}}+2\int\int_{Y}\left|\frac{d}{ds}\tilde{p}_{c(s)}(y)\right|dyds.\label{eq:d_TV_integrate}
\end{align}
Recall from (\ref{eq:foruma_dpds}) that
\[
\frac{d}{ds}\tilde{p}_{c(s)}(y)=\sum_{v:\exp_{c(s)}(v(s))=y}\left(-\int_{0}^{\ell}\frac{t(\ell-t)}{\ell}\frac{d}{ds}\text{Ric}(\gamma'_{s}(t))dt-\frac{1}{2}\frac{d}{ds}\norm{\frac{v(s)-\frac{h}{2}\mu(c(s))}{\sqrt{h}}}_{c(s)}^{2}\right)\zeta(v,s).
\]
Using Lemma \ref{lem:Jac-approx} again, we have that $\zeta(v,s)\leq C\cdot p(c(s)\overset{v(s)}{\rightarrow}y)\leq2p(c(s)\overset{v(s)}{\rightarrow}y)$
and hence
\[
\left|\frac{d}{ds}\tilde{p}_{c(s)}(y)\right|\leq2\sum_{v:\exp_{c(s)}(v(s))=y}\left(\left|\int_{0}^{\ell}\frac{t(\ell-t)}{\ell}\frac{d}{ds}\text{Ric}(\gamma'_{s}(t))dt\right|+\left|\frac{d}{ds}\norm{\frac{v(s)-\frac{h}{2}\mu(c(s))}{\sqrt{h}}}_{c(s)}^{2}\right|\right)p(c(s)\overset{v(s)}{\rightarrow}y).
\]

Since $V(\gamma_{s})\leq V_{0}$, we can use Lemma \ref{lem:smoothness_logdetJ}
to get
\[
\left|\int_{0}^{\ell}\frac{t(\ell-t)}{\ell}\frac{d}{ds}\text{Ric}(\gamma'_{s}(t))dt\right|\leq O\left(nhR_{2}\right).
\]
Hence, we have that
\begin{align*}
 & \int_{Y}\left|\frac{d}{ds}\tilde{p}_{c(s)}(y)\right|dy\\
\leq & nhR_{2}\int_{Y}\sum_{v:\exp_{c(s)}(v(s))=y}p(c(s)\overset{v(s)}{\rightarrow}y)dy+2\int_{Y}\sum_{v:\exp_{c(s)}(v(s))=y}\left|\frac{d}{ds}\norm{\frac{v(s)-\frac{h}{2}\mu(c(s))}{\sqrt{h}}}_{c(s)}^{2}\right|p(c(s)\overset{v(s)}{\rightarrow}y)dy.
\end{align*}

For the first term, we note that $\int_{Y}\sum_{v:\exp_{c(s)}(v(s))=y}p(c(s)\overset{v(s)}{\rightarrow}y)dy\leq1$. 

For the second term, we note that $y\in Y$ implies $V(\gamma_{s})\leq V_{0}$.
Hence, we have that
\begin{align*}
 & \int_{Y}\sum_{v:\exp_{c(s)}(v(s))=y}\left|\frac{d}{ds}\norm{\frac{v(s)-\frac{h}{2}\mu(c(s))}{\sqrt{h}}}_{c(s)}^{2}\right|p(c(s)\overset{v(s)}{\rightarrow}y)dy\\
\leq & \int_{V(\gamma_{s})\leq V_{0}}\sum_{v:\exp_{c(s)}(v(s))=y}\left|\frac{d}{ds}\norm{\frac{v(s)-\frac{h}{2}\mu(c(s))}{\sqrt{h}}}_{c(s)}^{2}\right|p(c(s)\overset{v(s)}{\rightarrow}y)dy\\
= & \E_{V(\gamma_{s})\leq V_{0}}\left|\frac{d}{ds}\norm{\frac{v(s)-\frac{h}{2}\mu(c(s))}{\sqrt{h}}}_{c(s)}^{2}\right|\\
= & O\left(D_{2}\sqrt{h}+\frac{1}{\sqrt{h}}+nhD_{0}R_{1}\right)
\end{align*}
 where we used Lemma \ref{lem:normvmuds} at the end. 

Hence, we have that
\[
\int_{Y}\left|\frac{d}{ds}\tilde{p}_{c(s)}(y)\right|dy=O\left(nhR_{2}+D_{2}\sqrt{h}+\frac{1}{\sqrt{h}}+nhD_{0}R_{1}\right).
\]

Putting this into (\ref{eq:d_TV_integrate}), we get
\[
d_{TV}(p_{x},p_{z})=O\left(nhR_{2}+D_{2}\sqrt{h}+\frac{1}{\sqrt{h}}+nhD_{0}R_{1}\right)d(x,z)+\frac{V_{0}}{20V_{1}}
\]
for any $d(x,z)<\frac{V_{0}}{4V_{1}}$. By summing along a path, for
any $x$ and $z$ and using that $V_{1}\geq V_{0}$, we have that
\[
d_{TV}(p_{x},p_{z})=O\left(nhR_{2}+D_{2}\sqrt{h}+\frac{1}{\sqrt{h}}+nhD_{0}R_{1}+1\right)d(x,z)+\frac{1}{20}.
\]
\end{proof}
\begin{defn}
\label{def:R2}Given a manifold $M$ and a family of geodesic $\gamma_{s}(t)$
with step size $h$ such that $h\leq H$ and $V(\gamma_{0})\leq V_{0}$.
Let $R_{2}$ be a constant depending on the manifold $M$ and the
step size $h$ such that for any $t$ such that $0\leq t\leq\ell$,
any curve $c(s)$ starting from $\gamma(t)$ and any vector field
$v(s)$ on $c(s)$ with $v(0)=\gamma'(t)$, we have that

\[
\left|\frac{d}{ds}\left.Ric(v(s))\right|_{s=0}\right|\leq\left(\norm{\left.\frac{dc}{ds}\right|_{s=0}}+\ell\norm{\left.D_{s}v\right|_{s=0}}\right)R_{2}.
\]
\end{defn}
\begin{lem}
\label{lem:smoothness_logdetJ}For $h\leq\min(H,\frac{1}{2nR_{1}})$
and $V(\gamma_{s})\leq V_{0}$, we have 
\[
\left|\int_{0}^{\ell}\frac{t(\ell-t)}{\ell}\frac{d}{ds}\text{Ric}(\gamma'_{s}(t))dt\right|\leq O\left(nhR_{2}\right)
\]
where $\gamma_{s}$ is as defined in the beginning of Section \ref{sec:Smoothness-of-P}.
\end{lem}
\begin{proof}
By Definition \ref{def:R2}, we have that
\begin{equation}
\left|\frac{d}{ds}Ric(\gamma'_{s}(s))\right|\leq\left(\norm{\frac{d}{ds}\gamma_{s}(0)}+\ell\norm{D_{s}\gamma'_{s}(0)}\right)R_{2}.\label{eq:dRicus}
\end{equation}
By definition of $\gamma_{s}$, we have that $\frac{d}{ds}\gamma_{s}(0)=\frac{d}{ds}\exp_{c(s)}(0)=\frac{d}{ds}c(s)$
is an unit vector and hence $\norm{\frac{d}{ds}\gamma_{s}(0)}=1$.
To bound the second term, we note that $\psi(t)=\frac{\partial}{\partial s}\gamma_{s}$
is a Jacobi field and Lemma \ref{lem:jacobi_approx_sol} shows that
\[
\norm{\psi(t)-\psi(0)-tD_{t}\psi(0)}\leq\lambda t^{2}\norm{\psi(0)}+\frac{\lambda t^{3}}{5}\norm{D_{t}\psi(0)}
\]
where $\lambda=\max_{0\leq s\leq\ell}\norm{R(s)}_{F}\leq R_{1}$.
Putting $t=\ell$, $\psi(\ell)=0$ and $\lambda\ell^{2}\leq R_{1}nh\leq\frac{1}{2}$,
we have
\[
\norm{\psi(0)+\ell D_{t}\psi(0)}\leq\frac{1}{2}\norm{\psi(0)}+\frac{\ell}{10}\norm{D_{t}\psi(0)}.
\]
Hence, we have that
\[
\norm{D_{t}\psi(0)}\leq\frac{3}{\ell}\norm{\psi(0)}=\frac{3}{\ell}.
\]
Using these, we have that $\norm{\psi(t)}=O(1)$ and $\norm{D_{t}\psi(t)}=O(\frac{1}{\ell})$
for $0\leq t\leq\ell$. 

Fix any $0\leq t\leq\ell$, we define $c(s)=\gamma_{s}(t)$ and $v(s)=\gamma_{s}'(t)$.
Using $\psi(t)=\frac{dc}{ds}$ and $D_{t}\psi(t)=D_{t}\frac{d}{ds}\gamma_{s}=D_{s}v$,
(\ref{eq:dRicus}) shows that
\[
\left|\frac{d}{ds}Ric(\gamma'_{s}(s))\right|\leq O\left(R_{2}\right).
\]
\end{proof}
\begin{lem}
\label{lem:normvmuds}Given an unit geodesic $c(s)$ from $x$ to
$z$. For any geodesic $\gamma(t)$ from $x$, we define a corresponding
vector field $v(s)$ on $c(s)$ such that $\gamma(t)=\exp_{x}(\frac{t}{\ell}v(0))$
and $\exp_{c(s)}(v(s))=\gamma(\ell)$ for all $s$ near $0$. Suppose
that $h\leq H$, then we have that
\[
\E_{V(\gamma)\leq V_{0}}\left|\left.\frac{d}{ds}\norm{\frac{v(s)-\frac{h}{2}\mu(c(s))}{\sqrt{h}}}_{c(s)}^{2}\right|_{s=0}\right|=O\left(D_{2}\sqrt{h}+\frac{1}{\sqrt{h}}+nhD_{0}R_{1}\right)
\]
where $\gamma$ is a random geodesic walk starting from $x$ (before
the filtering step).
\end{lem}
\begin{proof}
Note that
\begin{equation}
\left.\frac{d}{ds}\norm{\frac{v(s)-\frac{h}{2}\mu(c(s))}{\sqrt{h}}}_{c(s)}^{2}\right|_{s=0}=\frac{2}{h}\left\langle \left.D_{s}v\right|_{s=0}-\frac{h}{2}\left.D_{s}\mu\right|_{s=0},v-\frac{h}{2}\mu\right\rangle .\label{eq:normvmuds1}
\end{equation}
Since $v-\frac{h}{2}\mu$ is a random Gaussian vector $N(0,nhI)$
in $T_{x}M$ independent of $D_{s}\mu$, we have that 
\begin{equation}
\E_{\gamma}\left|\left\langle \frac{h}{2}\left.D_{s}\mu\right|_{s=0},v-\frac{h}{2}\mu\right\rangle \right|\leq O\left(\frac{\norm{\frac{h}{2}\left.D_{s}\mu\right|_{s=0}}\sqrt{nh}}{\sqrt{n}}\right)\leq O\left(h^{3/2}D_{2}\right).\label{eq:normvmuds2}
\end{equation}
By (\ref{eq:gamma_good}), we have that $\mathbb{P}(V(\gamma)<V_{0})\geq\frac{1}{2}$
and hence 
\[
\E_{V(\gamma)\leq V_{0}}\left|\left\langle \frac{h}{2}\left.D_{s}\mu\right|_{s=0},v-\frac{h}{2}\mu\right\rangle \right|=O\left(h^{3/2}D_{2}\right).
\]

By Lemma \ref{lem:smoothness_dx}, we know that $D_{s}v|_{s=0}=-c'+\zeta$
where $\zeta\perp v(0)$ and $\norm{\zeta}\leq\frac{3}{2}nhR_{1}$
when $V(\gamma)\leq V_{0}$. Since $v-\frac{h}{2}\mu$ is a random
Gaussian vector independent of $c'$, we have that
\begin{eqnarray}
\E_{V(\gamma)\leq V_{0}}\left|\left\langle D_{s}v,v-\frac{h}{2}\mu\right\rangle \right| & \leq & \E_{V(\gamma)\leq V_{0}}\left|\left\langle c',v-\frac{h}{2}\mu\right\rangle \right|+\frac{h}{2}\E\left|\left\langle \zeta,\mu\right\rangle \right|\nonumber \\
 & \leq & O\left(\frac{\norm{c'}\sqrt{nh}}{\sqrt{n}}\right)+nh^{2}R_{1}\norm{\mu}\nonumber \\
 & \leq & O\left(\sqrt{h}+nh^{2}D_{0}R_{1}\right).\label{eq:normvmuds3}
\end{eqnarray}

Combining the bounds (\ref{eq:normvmuds1}), (\ref{eq:normvmuds2})
and (\ref{eq:normvmuds3}), we have that
\begin{align*}
\E_{V(\gamma)\leq V_{0}}\left|\left.\frac{d}{ds}\norm{\frac{v(s)-\frac{h}{2}\mu(c(s))}{\sqrt{h}}}_{c(s)}^{2}\right|_{s=0}\right| & \leq O\left(D_{2}\sqrt{h}+\frac{1}{\sqrt{h}}+nhD_{0}R_{1}\right).
\end{align*}
\end{proof}

\subsection{Approximate solution of Jacobi field equations \label{subsec:Approx_Jacobi}}

Let $\gamma(t)$ be a geodesic and $\{X_{i}(t)\}_{i=1}^{n}$ be the
parallel transport of some orthonormal frame along $\gamma(t)$. As
we demonstrated in the proof of Lemma \ref{lem:formula_dexp}, Jacobi
fields can be expressed as linear combinations of $X_{i}$ and the
coefficients are given by the following matrix ODE:

\begin{eqnarray}
\frac{d^{2}}{dt^{2}}\psi(t)+R(t)\psi(t) & = & 0,\nonumber \\
\frac{d}{dt}\psi(0) & = & \beta,\label{eq:jacobi_general}\\
\psi(0) & = & \alpha\nonumber 
\end{eqnarray}
where $\psi(t),\alpha,\beta\in\Rn$ and $R(t)$ is defined in Definition
\ref{def:Rt_definition}.

In this section, we give estimates for Jacobi equations (\ref{eq:jacobi_general}).
The estimates we get can be viewed as small variants of the Rauch
comparison theorem (See \cite[Sec 4.5]{jost2008riemannian}). The
Rauch comparison theorem gives upper and lower bound on the magnitude
of Jacobi field. Our bounds instead show how fast Jacobi field deviates
from its linear approximation.

First, we give a basic bound on the solution in terms of hyperbolic
sine and cosine functions, which is a direct consequence of the Rauch
comparison theorem. We include a direct proof of this for completeness.
\begin{lem}
\label{lem:stable_Jacobi}Let $\psi$ be the solution of (\ref{eq:jacobi_general}).
Suppose that $\norm{R(t)}_{2}\leq\lambda$ for all $0\leq t\leq\ell$.
Then, we have that
\[
\norm{\psi(t)}_{2}\leq\norm{\alpha}_{2}\cosh(\sqrt{\lambda}t)+\frac{\norm{\beta}_{2}}{\sqrt{\lambda}}\sinh(\sqrt{\lambda}t)
\]
for all $0\leq t\leq\ell$. 
\end{lem}
\begin{proof}
Note that
\begin{eqnarray*}
\psi(t) & = & \psi(0)+\psi'(0)t+\int_{0}^{t}(t-s)\psi''(s)ds\\
 & = & \alpha+\beta t-\int_{0}^{t}(t-s)R(s)\psi(s)ds.
\end{eqnarray*}
Let $a(t)=\norm{\psi(t)}_{2}$, then we have that
\[
a(t)\leq\norm{\alpha}_{2}+\norm{\beta}_{2}t+\lambda\int_{0}^{t}(t-s)a(s)ds.
\]
Let $\overline{a}(t)$ be the solution of the integral equation
\[
\overline{a}(t)=\norm{\alpha}_{2}+\norm{\beta}_{2}t+\lambda\int_{0}^{t}(t-s)\overline{a}(s)ds.
\]
By induction, we have that $a(t)\leq\overline{a}(t)$ for all $t\geq0$.
By taking derivatives on both sides, we have that
\[
\overline{a}''(t)=\lambda\overline{a}(t),\ \overline{a}(0)=\norm{\alpha}_{2},\ \overline{a}'(0)=\norm{\beta}_{2}.
\]
Solving these equations, we have
\[
\norm{\psi(t)}_{2}=a(t)\leq\overline{a}(t)=\norm{\alpha}_{2}\cosh(\sqrt{\lambda}t)+\frac{\norm{\beta}_{2}}{\sqrt{\lambda}}\sinh(\sqrt{\lambda}t)
\]
for all $t\geq0$.
\end{proof}
Next, we give an approximate solution of (\ref{eq:jacobi_general}).
\begin{lem}
\label{lem:jacobi_approx_sol}Let $\psi$ be the solution of (\ref{eq:jacobi_general}).
Suppose that $\norm{R(t)}_{2}\leq\lambda$ for all $0\leq t\leq\frac{1}{\sqrt{\lambda}}$.
For any $0\leq t\leq\frac{1}{\sqrt{\lambda}}$, we have that
\[
\norm{\psi(t)-\alpha-\beta t}_{2}\leq\lambda t^{2}\norm{\alpha}_{2}+\frac{\lambda t^{3}}{5}\norm{\beta}_{2}
\]
and
\[
\norm{\psi'(t)-\beta}_{2}\leq2\lambda t\norm{\alpha}_{2}+\frac{3\lambda t^{2}}{5}\norm{\beta}_{2}
\]
\end{lem}
\begin{proof}
Note that
\begin{eqnarray*}
\psi(t) & = & \alpha+\beta t-\int_{0}^{t}(t-s)R(s)\psi(s)ds.
\end{eqnarray*}
Using Lemma \ref{lem:stable_Jacobi} and $\norm{R(t)}_{2}\leq\lambda$,
we have that
\begin{eqnarray*}
\norm{\psi(t)-\alpha-\beta t}_{2} & \leq & \lambda\int_{0}^{t}(t-s)\norm{\psi(s)}_{2}ds\\
 & \leq & \lambda\int_{0}^{t}(t-s)\left(\norm{\alpha}_{2}\cosh(\sqrt{\lambda}s)+\frac{\norm{\beta}_{2}}{\sqrt{\lambda}}\sinh(\sqrt{\lambda}s)\right)ds\\
 & = & \norm{\alpha}_{2}\left(\cosh(\sqrt{\lambda}t)-1\right)+\frac{\norm{\beta}_{2}}{\sqrt{\lambda}}\left(\sinh(\sqrt{\lambda}t)-\sqrt{\lambda}t\right).
\end{eqnarray*}
Since $0\leq t\leq\frac{1}{\sqrt{\lambda}}$, we have that $\left|\cosh(\sqrt{\lambda}t)-1\right|\leq\lambda t^{2}$
and $\left|\sinh(\sqrt{\lambda}t)-\sqrt{\lambda}t\right|\leq\frac{\lambda^{3/2}t^{3}}{5}$.
This gives the result.

Similarly, we have that
\[
\psi'(t)=\psi'(0)+\int_{0}^{t}\psi''(s)ds=\psi'(0)+\int_{0}^{t}R(s)\psi(s)ds.
\]
and hence 
\begin{eqnarray*}
\norm{\psi'(t)-\beta}_{2} & \leq & \lambda\int_{0}^{t}\left(\norm{\alpha}_{2}\cosh(\sqrt{\lambda}s)+\frac{\norm{\beta}_{2}}{\sqrt{\lambda}}\sinh(\sqrt{\lambda}s)\right)ds\\
 & \leq & \sqrt{\lambda}\norm{\alpha}_{2}\sinh(\sqrt{\lambda}t)+\norm{\beta}_{2}\left(\cosh(\sqrt{\lambda}t)-1\right)\\
 & \leq & 2\lambda t\norm{\alpha}_{2}+\frac{3}{5}\lambda t^{2}\norm{\beta}_{2}
\end{eqnarray*}
\end{proof}
The following is a matrix version of the above result.
\begin{lem}
\label{lem:jacobi_approx_sol_mat}Let $\Psi$ be the solution of 
\begin{eqnarray*}
\frac{d^{2}}{dt^{2}}\Psi(t)+R(t)\Psi(t) & = & 0,\\
\frac{d}{dt}\Psi(0) & = & B,\\
\Psi(0) & = & A
\end{eqnarray*}
where $\Psi,A$ and $B$ are matrices with a compatible size. Suppose
that $\norm{R(t)}_{2}\leq\lambda$. For any $0\leq t\leq\frac{1}{\sqrt{\lambda}}$,
we have that
\[
\norm{\Psi(t)-A-Bt}_{F}\leq\max_{0\leq s\leq t}\norm{R(s)}_{F}\left(t^{2}\norm A_{2}+\frac{t^{3}}{5}\norm B_{2}\right).
\]
\end{lem}
\begin{proof}
Note that $\Psi(t)x$ is the solution of (\ref{eq:jacobi_general})
with $\beta=Bx$ and $\alpha=Ax$. Therefore, Lemma \ref{lem:stable_Jacobi}
shows that
\begin{eqnarray*}
\norm{\Psi(t)x}_{2} & \leq & \norm{Ax}_{2}\cosh(\sqrt{\lambda}t)+\frac{\norm{Bx}_{2}}{\sqrt{\lambda}}\sinh(\sqrt{\lambda}t)\\
 & \leq & \left(\norm A_{2}\cosh(\sqrt{\lambda}t)+\frac{\norm B_{2}}{\sqrt{\lambda}}\sinh(\sqrt{\lambda}t)\right)\norm x_{2}
\end{eqnarray*}
for all $x$. Therefore, we have that
\[
\Psi(t)^{T}\Psi(t)\preceq\left(\norm A_{2}\cosh(\sqrt{\lambda}t)+\frac{\norm B_{2}}{\sqrt{\lambda}}\sinh(\sqrt{\lambda}t)\right)^{2}I.
\]
Hence, we have that
\[
\Psi(t)\Psi(t)^{T}\preceq\left(\norm A_{2}\cosh(\sqrt{\lambda}t)+\frac{\norm B_{2}}{\sqrt{\lambda}}\sinh(\sqrt{\lambda}t)\right)^{2}I.
\]

Using this, we have that
\begin{eqnarray*}
\norm{\Psi(t)-A-Bt}_{F} & = & \norm{\int_{0}^{t}(t-s)R(s)\Psi(s)ds}_{F}\\
 & \leq & \int_{0}^{t}(t-s)\sqrt{\tr\Psi^{T}(s)R^{T}(s)R(s)\Psi(s)}ds\\
 & = & \int_{0}^{t}(t-s)\sqrt{\tr R(s)\Psi(s)\Psi^{T}(s)R^{T}(s)}ds\\
 & \leq & \int_{0}^{t}(t-s)\left(\norm A_{2}\cosh(\sqrt{\lambda}s)+\frac{\norm B_{2}}{\sqrt{\lambda}}\sinh(\sqrt{\lambda}s)\right)\norm{R(s)}_{F}ds.\\
 & \leq & \max_{0\leq s\leq t}\norm{R(s)}_{F}\int_{0}^{t}(t-s)\left(\norm A_{2}\cosh(\sqrt{\lambda}s)+\frac{\norm B_{2}}{\sqrt{\lambda}}\sinh(\sqrt{\lambda}s)\right)ds\\
 & = & \max_{0\leq s\leq t}\norm{R(s)}_{F}\left(\norm A_{2}\left(\cosh(\sqrt{\lambda}t)-1\right)+\frac{\norm B_{2}}{\sqrt{\lambda}}\left(\sinh(\sqrt{\lambda}t)-\sqrt{\lambda}t\right)\right).
\end{eqnarray*}
Since $0\leq t\leq\frac{1}{\sqrt{\lambda}}$, we have that $\left|\cosh(\sqrt{\lambda}t)-1\right|\leq\lambda t^{2}$
and $\left|\sinh(\sqrt{\lambda}t)-\sqrt{\lambda}t\right|\leq\frac{\lambda^{3/2}t^{3}}{5}$.
This gives the result.
\end{proof}

\subsection{Almost one-to-one correspondence of geodesics}

\todo{changed a lot}We do not know if every pair $x,y\in M$ has
a unique geodesic connecting $x$ and $y$. Due to this, the probability
density $p_{x}$ at $y$ on $\mathcal{M}$ can be the sum over all
possible geodesics connect $x$ and $y$. The goal of this section
is to show there is a 1-1 map between geodesics paths connecting $x$
to $y$ as we move $x.$
\begin{lem}
\label{lem:smoothness_dx}Given a geodesic $\gamma(t)=\exp_{x}(\frac{t}{\ell}v_{x})$.
Let the end points $x=\gamma(0)$ and $y=\gamma(\ell)$. Suppose that
$\ell^{2}\leq\frac{1}{R}$ where $R=\norm{R(t)}_{F}$ and $R(t)$
defined in Definition \ref{def:Rt_definition}, then there is an unique
smooth invertible function $v:U\subset M\rightarrow V\subset T_{x}M$
such that
\[
y=\exp_{z}(v(z))
\]
for any $z\in U$ where $U$ is a neighborhood of $x$ and $V$ is
a neighborhood of $v_{x}=v(x)$. Furthermore, for any $\eta=\alpha v_{x}+\eta_{2}$
with $\eta_{2}\perp v_{x}$ and scale $\alpha$, we have that 
\[
\nabla_{\eta}v(x)=-\eta+\zeta
\]
where $\norm{\zeta}_{x}\leq\frac{3}{2}\ell^{2}R\norm{\eta_{2}}_{x}\leq\frac{3}{2}\norm{\eta_{2}}_{x}$
and $\zeta\perp v_{x}$. In particular, we have that $\norm{\nabla_{\eta}v(x)}_{x}\leq\frac{5}{2}\norm{\eta}_{x}$.
\end{lem}
\begin{proof}
Consider the smooth function $f(z,w)=\exp_{z}(w)$. From Lemma \ref{lem:Jac-approx},
the differential of $w$ at $(x,v_{x})$ on the $w$ variables, i.e.
$d\exp_{x}(v_{x})$, is invertible. Hence, the implicit function theorem
shows that there is a open neighborhood $U$ of $x$ and a unique
function $v$ on $U$ such that $f(z,v(z))=f(x,v_{x})$, i.e. $y=\exp_{x}(v_{x})=\exp_{z}(v(z))$.

To compute $\nabla_{\eta}v$, let $c(s)$ be a geodesic starting from
$x$ with $c'(0)=\eta$ and $c(t,s)=\exp_{c(s)}(\frac{t}{\ell}v(c(s)))$
be a family of geodesics with the end points $c(\ell,s)=\exp_{c(s)}(v(c(s)))=y$.
Note that $\psi(t)\defeq\frac{\partial c(t,s)}{\partial s}\left|_{s=0}\right.$
satisfies the Jacobi field equation
\[
D_{t}^{2}\psi(t)+R(\psi,\frac{\partial c}{\partial t})\frac{\partial c}{\partial t}=0.
\]
Let $\chi=\nabla_{\eta}v(x)$, we know that
\begin{eqnarray*}
\psi(0) & = & \frac{\partial c(0,s)}{\partial s}\left|_{s=0}\right.=\gamma'(0)=\eta,\\
D_{t}\psi(0) & = & D_{t}\frac{\partial c(t,s)}{\partial s}\left|_{t,s=0}\right.=D_{s}\frac{\partial c(t,s)}{\partial t}\left|_{t,s=0}\right.=\frac{1}{\ell}\nabla_{\eta}v(x)=\frac{\chi}{\ell},\\
\psi(\ell) & = & \frac{\partial c(\ell,s)}{\partial s}\left|_{s=0}\right.=0
\end{eqnarray*}
where we used $D_{t}\frac{\partial}{\partial s}=D_{s}\frac{\partial}{\partial t}$
(torsion free-ness, Fact \ref{fact:basic_RG}). 

From Fact \ref{fact:Jacobi_field_split}, we know that $\psi$ can
be split into the tangential part $\psi_{1}$ and the normal part
$\psi_{2}$. For the tangential part $\psi_{1}$, we know that
\begin{eqnarray*}
\psi_{1}(t) & = & \left(\left\langle \psi(0),\gamma'(0)\right\rangle _{\gamma(0)}+\left\langle D_{t}\psi(0),\gamma'(0)\right\rangle _{\gamma(0)}t\right)\gamma'(t)\\
 & = & \left(\left\langle \eta,v_{x}\right\rangle _{x}+\left\langle \frac{\chi}{\ell},v_{x}\right\rangle _{x}t\right)\gamma'(t)\\
 & = & \left(\alpha\norm{v_{x}}_{x}^{2}+\left\langle \frac{\chi}{\ell},v_{x}\right\rangle _{x}t\right)\gamma'(t).
\end{eqnarray*}
Since $\psi_{1}(\ell)=0$, we have that
\begin{equation}
\left\langle \chi,v_{x}\right\rangle =-\alpha\norm{v_{x}}_{x}^{2}.\label{eq:tangential_part_general_smooth}
\end{equation}

For the normal part $\psi_{2}$, it is easier to calculate using orthogonal
frames. Similar to Lemma \ref{lem:formula_dexp}, we pick an arbitrary
orthogonal frame $X_{i}(t)$ parallel transported along the curve
$\gamma(t)$ and let $\bar{\psi}(t)$ be $\psi_{2}(t)$ represented
in that orthogonal frame. Hence, we have that
\begin{eqnarray*}
\frac{d^{2}}{dt^{2}}\bar{\psi}(t)+R(t)\bar{\psi}(t) & = & 0,\\
\bar{\psi}(0) & = & \overline{\eta}_{2},\\
\bar{\psi}'(0) & = & \frac{\overline{\chi}_{2}}{\ell},\\
\bar{\psi}(\ell) & = & 0
\end{eqnarray*}
where $\overline{\chi}_{2}=\overline{\chi}-\left\langle \chi,\frac{v_{x}}{\norm{v_{x}}}\right\rangle _{x}\frac{\overline{v}_{x}}{\norm{v_{x}}}=\overline{\chi}+\alpha\overline{v}_{x}$
by (\ref{eq:tangential_part_general_smooth}) and $R(t)_{ij}=\left\langle R(X_{i}(t),\gamma'(t))\gamma'(t),X_{j}(t)\right\rangle $.
Since $\norm{R(t)}_{2}\leq\norm{R(t)}_{F}\leq R$ and $\ell^{2}\leq\frac{1}{R}$,
Lemma \ref{lem:jacobi_approx_sol} shows that
\begin{eqnarray*}
\norm{\overline{\eta}_{2}+\overline{\chi}_{2}}_{2}=\norm{\bar{\psi}(\ell)-\overline{\eta}_{2}-\overline{\chi}_{2}}_{2} & \leq & R\ell^{2}\norm{\overline{\eta}_{2}}_{2}+\frac{R\ell^{2}}{5}\norm{\overline{\chi}_{2}}_{2}\\
 & \leq & \frac{6}{5}R\ell^{2}\norm{\overline{\eta}_{2}}_{2}+\frac{R\ell^{2}}{5}\norm{\overline{\eta}_{2}+\overline{\chi}_{2}}_{2}.
\end{eqnarray*}
Hence, we have that $\norm{\overline{\eta}_{2}+\overline{\chi}_{2}}_{2}\leq\frac{3}{2}R\ell^{2}\norm{\overline{\eta}_{2}}_{2}$.
Therefore, we have that
\begin{eqnarray*}
\nabla_{\eta}v(x) & = & \chi\\
 & = & -\alpha v_{x}+\chi_{2}\\
 & = & -\eta+(\eta_{2}+\chi_{2})
\end{eqnarray*}
where $\norm{\chi_{2}+\eta_{2}}_{x}\le\frac{3}{2}R\ell^{2}\norm{\eta_{2}}_{x}$.
Furthermore, we have that both $\eta_{2}$ and $\chi_{2}$ are orthogonal
to $\eta$.
\end{proof}
\begin{rem*}
If the above lemma holds without the assumption $\ell^{2}\leq\frac{1}{R}$,
this would imply uniqueness of geodesics.
\end{rem*}
The following lemma shows there is a 1-1 map between geodesics paths
connecting $x$ to $y$ as we move $x.$ When we move $x$, the geodesic
$\gamma$ from $x$ to $y$ changes and hence we need to bound $V(\gamma)$. 
\begin{lem}
\label{lem:one_one_cor}Given a geodesic $\gamma(t)=\exp_{x}(\frac{t}{\ell}v_{x})$
with step size $h$ satisfying $h\leq\min(H,\frac{1}{2nR_{1}})$,
let $c(s)$ be any geodesic starting at $\gamma(0)$. Let $x=c(0)=\gamma(0)$
and $y=\gamma(1)$. Suppose that the length of $c(s)$ is less than
$\frac{V_{0}}{4V_{1}}$ and $V(\gamma)\leq\frac{V_{0}}{2}$. Then,
there is a unique vector field $v$ on $c$ such that 
\[
y=\exp_{c(s)}(v(s)).
\]
Furthermore, this vector field is uniquely determined by the geodesic
$c(s)$ and any $v(s)$ on this vector field. Also, we have that $V(\exp_{c(s)}(v(s)))\leq V_{0}$
for all $s$.
\end{lem}
\begin{proof}
Let $s_{\max}$ be the supremum of $s$ such that $v(s)$ can be defined
continuously such that $y=\exp_{c(s)}(v(s))$ and $V(\gamma_{s})\leq V_{0}$
where $\gamma_{s}(t)=\exp_{c(s)}(\frac{t}{\ell}v(s)).$ Lemma \ref{lem:smoothness_dx}
shows that there is a neighborhood $N$ at $x$ and a vector field
$u$ on $N$ such that for any $z\in N$, we have that 
\[
y=\exp_{z}(u(z)).
\]
Also, this lemma shows that $u(s)$ is smooth and hence the parameter
$V_{1}$ shows that $V(\gamma_{s})$ is Lipschitz in $s$. Therefore,
$V(\gamma_{s})\leq V_{0}$ for a small neighborhood of $0$. Hence
$s_{\max}>0$.

Now, we show $s_{\max}>1$ by contraction. By the definition of $s_{\max}$,
we have that $V(\gamma_{s})\leq V_{0}$ for any $0\leq s<s_{\max}$.
Hence, we can apply Lemma \ref{lem:smoothness_dx} with $R=R_{1}$.
In particular, this shows that $\norm{D_{s}v(s)}=\norm{\nabla_{\frac{d}{ds}c}u(x)}\leq\frac{5}{2}\norm{\frac{d}{ds}c}=\frac{5}{2}L$
where $L$ is the length of $c$. Therefore, the function $v$ is
Lipschitz and hence $v(s_{\max})$ is well-defined and $V(\gamma_{s_{\max}})\leq V_{0}$
by continuity. Hence, we can apply Lemma \ref{lem:smoothness_dx}
at $v(s_{\max})$ and extend the domain of $v(s)$ beyond $s_{\max}$.

To bound $V(\gamma_{s})$, we note that $\norm{D_{s}\gamma_{s}'}=\frac{1}{\ell}\norm{D_{s}v(s)}\leq\frac{5}{2\ell}L$
and $\norm{\frac{d}{ds}c}=L$. Hence, $\left|\frac{d}{ds}V(\gamma_{s})\right|\leq(L+\frac{5}{2}L)V_{1}$
by the definition of $V_{1}$. Therefore, if $L\leq\frac{V_{0}}{4V_{1}}$,
we have that $V(\gamma_{s})\leq V(\gamma)+\frac{V_{0}}{2}\leq V_{0}$
for all $s\leq1.1$ whenever $v(s)$ is defined. Therefore, this draws
a contradiction that $s_{\max}$ is not the supremum. Hence, $s_{\max}\geq1$. 

The uniqueness follows from Lemma \ref{lem:smoothness_dx}.
\end{proof}

\subsection{Implementation\label{subsec:implementation_general}}

Here, we explain in high level how to implement the geodesic walk
in general via an efficient algorithm for approximately solving ODEs.
Note that to implement the step, we need to compute the geodesic and
compute the probabilities $p(x\overset{w}{\rightarrow}y)$ and $p(y\overset{w'}{\rightarrow}x)$.
From the formula (\ref{eq:formula_ratio}), we see that they involve
the term $d\exp_{x}(w)$. Lemma \ref{lem:formula_dexp} shows that
$d\exp_{x}(w)$ can be computed by a Jacobi field and the latter can
be computed by an ODE if it is written in the orthonormal frame systems.
Therefore, to implement the geodesic walk, we need to compute geodesic,
parallel transport and Jacobi field (See Algorithm \ref{algo:geodesic_detailed}).
All of these are ODEs and can be solved using the collocation method.
In the later sections, we will see how to use the collocation method
to solve these ODEs in nearly matrix multiplication time for the log
barrier.

\begin{algorithm2e}
\caption{Geodesic Walk (Detailed)}

\SetAlgoLined\label{algo:geodesic_detailed}

Pick a Gaussian random vector $w\sim N_{x}(0,I)$$.$

$\ $

\tcc{Compute $y=\exp_{x}(\sqrt{h}w+\frac{h}{2}\mu(x))$ where $\mu(x)$
is given by (\ref{eq:drift_term_general}).}

\tcc{Compute a corresponding $w'$ s.t. $x=exp_{y}(\sqrt{h}w'+\frac{h}{2}\mu(y)).$}

Generate a random direction $d=\sqrt{h}w+\frac{h}{2}\mu(x)$ where
$\mu(x)$ is given by (\ref{eq:drift_term_general}).

Solve the geodesic equation $D_{\gamma'}\gamma'=0$ with $\gamma(0)=x$
and $\gamma'(0)=d$ using collocation method (Sec \ref{subsec:CollocationMethod}).

Set $y=\gamma(1)$ and $w'=-\gamma'(1)$.

$\ $

\tcc{Compute the probability $p(x\overset{w}{\rightarrow}y)$ of
going from $x$ to $y$ using the step $w$.}

Pick an orthonormal frame $\{X_{i}\}_{i=1}^{n}$ at $x$.

Compute the parallel transport of $\{X_{i}\}_{i=1}^{n}$ along $\gamma(t)$
using collocation method.

Compute $d\exp_{x}(w)$ via solve the Jacobi field equation (\ref{eq:jacobian})
(in the coordinate systems $\{X_{i}(t)\}$).

Set $p(x\overset{w}{\rightarrow}y)=\det(d\exp_{x}(w))^{-1}\sqrt{\frac{\det\left(g(y)\right)}{\left(2\pi h\right)^{n}}}\exp\left(-\frac{1}{2}\norm{\frac{w-\frac{h}{2}\mu(x)}{\sqrt{h}}}_{x}^{2}\right)$.

$\ $

Compute $p(y\overset{w'}{\rightarrow}x)$ similarly.

$\ $

With probability $\min\left(1,\frac{p(y\overset{w'}{\rightarrow}x)}{p(x\overset{w}{\rightarrow}y)}\right)$,
go to $y$; otherwise, stay at $x$.

\end{algorithm2e}

\pagebreak{}
\section{\label{sec:Logarithmic-Barrier}Logarithmic barrier}

For any polytope $M=\{Ax>b\}$, the logarithmic barrier function $\phi(x)$
is defined as
\[
\phi(x)=-\sum_{i=1}^{m}\log(a_{i}^{T}x-b_{i}).
\]
We denote the Hessian manifold induced by the logarithmic barrier
on $M$ by $M_{L}$. The goal of this section is to analyze the geodisic
walk on $M_{L}.$

In section \ref{subsec:RG_L}, we give explicit formulas for various
Riemannian geometry concepts on $M_{L}$. In Section \ref{subsec:geodesic_walk_log},
we describe the geodesic walk specialized to $M_{L}$. In Sections
\ref{sec:walk_is_random_log} to \ref{subsec:Stability-of-norm},
we bound the parameters required by Theorem \ref{thm:gen-convergence},
resulting in the following theorem.
\begin{thm}
\label{thm:gen-convergence-log} The geodesic walk on $M_{L}$ with
step size $h=\frac{c}{n^{3/4}}$ has mixing time $O(mn^{3/4})$ for
some universal constant $c$.
\end{thm}
In later sections, we show how to implement geodesic walk and calculate
the rejection probability. To implement these, we apply the techniques
developed in Section \ref{subsec:CollocationMethod} to solve the
corresponding ODEs, after showing that the geodesic, parallel transport
and Jacobi field are complex analytic (Section \ref{subsec:Complex_analytic_geodesic_log}),
for a large radius of convergence (Section \ref{subsec:Complex_analytic_geodesic_log}). 
\begin{thm}
\label{thm:implementation_log} There exists a universal constant
$c>0$ s.t. for the standard logarithmic barrier, one step of the
geodesic walk with step size $h\leq\frac{c}{\sqrt{n}}$ can be implemented
in time $O(mn^{\omega-1}\log^{2}(n))$.
\end{thm}

\subsection{Riemannian geometry on $M_{L}$ ($G_{2}$)}

\label{subsec:RG_L}We use the following definitions throughout this
section:
\begin{itemize}
\item $A_{x}=S_{x}^{-1}A$.
\item $s_{x}=Ax-b$, $S_{x}=\Diag(s_{x})$, $s_{x,v}=A_{x}v$, $S_{x,v}=\Diag(A_{x}v)$.
\item $P_{x}=A_{x}(A_{x}^{T}A_{x})^{-1}A_{x}^{T}$, $\sigma_{x}=\Diag(P_{x})$,
$\left(P_{x}^{(2)}\right)_{ij}=\left(P_{x}\right)_{ij}^{2}$.
\item Gradient of $\phi$: $\phi_{i}=-\sum_{\ell}\left(e_{\ell}^{T}A_{x}e_{i}\right)$.
\item Hessian of $\phi$ and its inverse: $g_{ij}=\phi_{ij}=\sum_{\ell}\left(e_{\ell}^{T}A_{x}e_{i}\right)\left(e_{\ell}^{T}A_{x}e_{j}\right)$,
$g^{ij}=e_{i}^{T}\left(A_{x}^{T}A_{x}\right)^{-1}e_{j}$.
\item Third derivatives of $\phi$: $\phi_{ijk}=-2\sum_{\ell}\left(e_{\ell}^{T}A_{x}e_{i}\right)\left(e_{\ell}^{T}A_{x}e_{j}\right)\left(e_{\ell}^{T}A_{x}e_{k}\right)$.
\item For brevity (overloading notation), we define $s_{\gamma'}=s_{\gamma,\gamma'}$,
$s_{\gamma''}=s_{\gamma,\gamma''}$ , $S_{\gamma'}=S_{\gamma,\gamma'}$
and $S_{\gamma''}=S_{\gamma,\gamma''}$ for a curve $\gamma(t)$.
\item Formula for the transition probability:
\[
p(x\overset{v_{x}}{\rightarrow}y)=\sum_{v_{x}:\exp_{x}(v_{x})=y}\det(d\exp_{x}(v_{x}))^{-1}\sqrt{\frac{\det\left(A_{y}^{T}A_{y}\right)}{\left(2\pi h\right)^{n}}}\exp\left(-\frac{1}{2h}\norm{A_{x}\left(v_{x}-\frac{h}{2}\mu(x)\right)}^{2}\right)
\]
where the formula of $\mu(x)$ is given in Lemma \ref{lem:geodesic_walk_log_def}. 
\end{itemize}
In this and subsequent sections, we will frequently use the following
elementary calculus facts (using only the chain/product rules and
formula for derivative of inverse of a matrix):

\begin{align*}
\frac{dA_{\gamma}}{dt} & =-S_{\gamma'}A_{\gamma},\\
\frac{dP_{\gamma}}{dt} & =-S_{\gamma'}P_{\gamma}-P_{\gamma}S_{\gamma'}+2P_{\gamma}S_{\gamma'}P_{\gamma},\\
\frac{dS_{\gamma'}}{dt} & =\Diag(-S_{\gamma'}A_{\gamma}\gamma'+A_{\gamma}\gamma'')=-S_{\gamma'}^{2}+S_{\gamma''},\\
\frac{d\sigma_{\gamma}}{dt} & =\Diag(\frac{dP_{\gamma}}{dt}).
\end{align*}
We also use these matrix inequalities: $\tr(AB)=\tr(BA)$, $\tr(PAP)\le\tr(A)$
for any psd matrix $A$; $\tr(ABA^{T})\le\tr(AZA^{T})$ for any $B\preceq Z$;
the Cauchy-Schwartz, namely, $\tr(AB)\le\tr(AA^{T})^{\frac{1}{2}}\tr(BB^{T})^{\frac{1}{2}}.$
We also use $P^{2}=P$ since $P$ is a projection matrix.

Since the Hessian manifold $M_{L}$ is naturally embedded in $\Rn$,
we identify $T_{x}M_{L}$ by Euclidean coordinates unless otherwise
stated. Therefore, we have that
\begin{eqnarray*}
\left\langle u,v\right\rangle _{x} & = & u^{T}\nabla^{2}\phi(x)v\\
 & = & u^{T}A_{x}^{T}A_{x}v.
\end{eqnarray*}

\begin{lem}
\label{lem:geo_equ_log}Let $u(t)$ be a vector field defined on a
curve $\gamma(t)$ in $M_{L}$. Then, we have that
\[
\nabla_{\gamma'}u=\frac{du}{dt}-\left(A_{\gamma}^{T}A_{\gamma}\right)^{-1}A_{\gamma}^{T}S_{\gamma'}s_{\gamma,u}.
\]
In particular, the geodesic equation on $M_{L}$ is given by
\begin{equation}
\gamma''=\left(A_{\gamma}^{T}A_{\gamma}\right)^{-1}A_{\gamma}^{T}s_{\gamma'}^{2}\label{eq:geodesic_log}
\end{equation}
and the equation for parallel transport on a curve $\gamma(t)$ is
given by
\begin{equation}
\frac{d}{dt}v(t)=\left(A_{\gamma}^{T}A_{\gamma}\right)^{-1}A_{\gamma}^{T}S_{\gamma'}A_{\gamma}v.\label{eq:parallel_log}
\end{equation}
\end{lem}
\begin{proof}
By Lemma \ref{lem:Hessian_formula}, the Christoffel symbols respect
to the Euclidean coordinates is given by 
\begin{eqnarray*}
\Gamma_{ij}^{k} & = & \frac{1}{2}\sum_{l}g^{kl}\phi_{ijl}\\
 & = & -\sum_{z}\sum_{l}e_{k}^{T}\left(A_{x}^{T}A_{x}\right)^{-1}e_{l}\left(A_{x}e_{i}\right)_{z}\left(A_{x}e_{j}\right)_{z}\left(A_{x}e_{l}\right)_{z}\\
 & = & -e_{k}^{T}\left(A_{x}^{T}A_{x}\right)^{-1}A_{x}^{T}\left(\left(A_{x}e_{i}\right)\left(A_{x}e_{j}\right)\right).
\end{eqnarray*}
Recall that the Levi-Civita connection is given by 
\[
\nabla_{v}u=\sum_{ik}v_{i}\frac{\partial u_{k}}{\partial x^{i}}e_{k}+\sum_{ijk}v_{i}u_{j}\Gamma_{ij}^{k}e_{k}.
\]
Therefore, we have
\[
\nabla_{v}u=\sum_{ik}v^{i}\frac{\partial u^{k}}{\partial x^{i}}e_{k}-\left(A_{x}^{T}A_{x}\right)^{-1}A_{x}^{T}S_{x,v}s_{x,u}.
\]
Since $U$ is a vector field defined on a curve $\gamma(t)$, we have
that $\sum_{ik}\gamma'_{i}\frac{\partial u_{k}}{\partial x^{i}}e_{k}=\frac{du}{dt}$.

The geodesic equation follows from $\nabla_{\gamma'}\gamma'=0$ and
the parallel transport equation follows from $\nabla_{\gamma'}v=0$.
\end{proof}
\begin{lem}
\label{lem:Riemann_tensor_log}Given $u,v,w,z\in T_{x}M_{L}$, the
Riemann Curvature Tensor is given by
\begin{eqnarray*}
\left\langle R(u,v)w,z\right\rangle  & = & \left(s_{x,u}s_{x,w}\right)^{T}P_{x}\left(s_{x,v}s_{x,z}\right)-\left(s_{x,u}s_{x,z}\right)^{T}P_{x}\left(s_{x,v}s_{x,w}\right),\\
R(u,v)w & = & \left(A_{x}^{T}A_{x}\right)^{-1}A_{x}^{T}\left(S_{x,v}P_{x}S_{x,w}-\diag(P_{x}s_{x,v}s_{x,w})\right)A_{x}u
\end{eqnarray*}
and the Ricci curvature is given by
\begin{eqnarray*}
Ric(u) & = & s_{x,u}^{T}P_{x}^{(2)}s_{x,u}-\sigma_{x}^{T}P_{x}s_{x,u}^{2}.
\end{eqnarray*}
\end{lem}
\begin{proof}
By Lemma \ref{lem:Hessian_formula}, we have that
\begin{eqnarray*}
\left\langle R(u,v)w,z\right\rangle  & = & \frac{1}{4}\sum_{pqijlk}g^{pq}\left(\phi_{jkp}\phi_{ilq}-\phi_{ikp}\phi_{jlq}\right)u_{i}v_{j}w_{l}z_{k}\\
 & = & \sum_{pq}g^{pq}\left(e_{p}^{T}A_{x}^{T}\left(s_{x,v}s_{x,z}\right)e_{q}^{T}A_{x}^{T}\left(s_{x,u}s_{x,w}\right)-e_{p}^{T}A_{x}^{T}\left(s_{x,u}s_{x,z}\right)e_{q}^{T}A_{x}^{T}\left(s_{x,v}s_{x,w}\right)\right)\\
 & = & \left(s_{x,u}s_{x,w}\right)^{T}A_{x}\left(A_{x}^{T}A_{x}\right)^{-1}A_{x}^{T}\left(s_{x,v}s_{x,z}\right)-\left(s_{x,u}s_{x,z}\right)^{T}A_{x}\left(A_{x}^{T}A_{x}\right)^{-1}A_{x}^{T}\left(s_{x,v}s_{x,w}\right).
\end{eqnarray*}
Rewriting it, we have that
\begin{align*}
\left\langle R(u,v)w,z\right\rangle  & =s_{x,u}S_{x,w}P_{x}S_{x,v}s_{x,z}-\left(s_{x,u}s_{x,z}\right)^{T}P_{x}\left(s_{x,v}s_{x,w}\right)\\
 & =z^{T}A_{x}^{T}S_{x,w}P_{x}S_{x,v}A_{x}u-z^{T}A_{x}^{T}\diag(P_{x}s_{x,v}s_{x,w})A_{x}u\\
 & =z^{T}A_{x}^{T}\left(S_{x,w}P_{x}S_{x,v}-\diag(P_{x}s_{x,v}s_{x,w})\right)A_{x}u
\end{align*}
Since $\left\langle \alpha,\beta\right\rangle =\alpha^{T}A_{x}^{T}A_{x}\beta$,
we have that
\[
R(u,v)w=\left(A_{x}^{T}A_{x}\right)^{-1}A_{x}^{T}\left(A_{x}^{T}S_{x,w}P_{x}S_{x,v}-\diag(P_{x}s_{x,v}s_{x,w})\right)A_{x}u.
\]

For the Ricci curvature, we have that
\begin{eqnarray*}
Ric(u) & = & \sum_{jl}g^{jl}\left\langle R(u,\frac{\partial}{\partial x^{j}}),\frac{\partial}{\partial x^{l}},u\right\rangle \\
 & = & \sum_{jl}g^{jl}\left(\left(s_{x,u}s_{x,e_{l}}\right)^{T}P_{x}\left(s_{x,e_{j}}s_{x,u}\right)-\left(s_{x,u}^{2}\right)^{T}P_{x}\left(s_{x,e_{j}}s_{x,e_{l}}\right)\right)\\
 & = & \sum_{jl}g^{jl}e_{l}^{T}A_{x}^{T}\left(S_{x,u}P_{x}S_{x,u}-\diag(P_{x}s_{x,u}^{2})\right)A_{x}e_{j}\\
 & = & \tr\left((A_{x}^{T}A_{x})^{-1}A_{x}^{T}\left(S_{x,u}P_{x}S_{x,u}-\diag(P_{x}s_{x,u}^{2})\right)A_{x}\right)\\
 & = & \tr P_{x}S_{x,u}P_{x}S_{x,u}-\tr P_{x}\diag(P_{x}s_{x,u}^{2})\\
 & = & s_{x,u}^{T}P_{x}^{(2)}s_{x,u}-\sigma_{x}^{T}P_{x}s_{x,u}^{2}.
\end{eqnarray*}
\end{proof}
\begin{lem}
\label{lem:Jacobi_field_log}Given a geodesic $\gamma(t)$ on $M_{L}$
and an orthogonal frame $\{x_{i}\}_{i=1}^{n}$ on $\gamma(t)$. The
Jacobi field equation (in the orthogonal frame coordinates) is given
by
\[
\frac{d^{2}u}{dt^{2}}+X^{-1}(A_{\gamma}^{T}A_{\gamma})^{-1}\left(A_{\gamma}^{T}S_{\gamma'}P_{\gamma}S_{\gamma'}A_{\gamma}-A_{\gamma}^{T}\Diag(P_{\gamma}s_{\gamma'}^{2})A_{\gamma}\right)Xu=0
\]
where $X(t)=[x_{1}(t),x_{2}(t),\cdots,x_{n}(t)]$.
\end{lem}
\begin{proof}
The equation for Jacobi field along is
\[
D_{t}^{2}v+R(v,\gamma')\gamma'=0.
\]
By Lemma \ref{lem:Riemann_tensor_log}, under Euclidean coordinates,
we have
\[
D_{t}^{2}v+(A_{\gamma}^{T}A_{\gamma})^{-1}\left(A_{\gamma}^{T}S_{\gamma'}P_{\gamma}S_{\gamma'}A_{\gamma}-A_{\gamma}^{T}diag(P_{\gamma}s_{\gamma'}^{2})A_{\gamma}\right)v=0
\]
We write $v$ in terms of the orthogonal frame, namely, $v(t)=X(t)u(t)$
where $u(t)\in\R^{n}$. Then, we have that $D_{t}^{2}v=X(t)\frac{d^{2}u(t)}{dt^{2}}.$
Hence, under the orthogonal frame coordinate, we have that 
\[
\frac{d^{2}u}{dt^{2}}+X^{-1}(A_{\gamma}^{T}A_{\gamma})^{-1}\left(A_{\gamma}^{T}S_{\gamma'}P_{\gamma}S_{\gamma'}A_{\gamma}-A_{\gamma}^{T}\Diag(P_{\gamma}s_{\gamma'}^{2})A_{\gamma}\right)Xu=0.
\]
\end{proof}
Now, we show the relation between the distance induced by the metric
and the Hilbert metric.
\begin{lem}
\label{lem:distances_log}For any $x,y\in M_{L}$, we have that 
\[
\frac{d(x,y)}{d_{H}(x,y)}\leq\sqrt{m}.
\]
Hence, we have that $G_{2}\leq\sqrt{m}$.
\end{lem}
\begin{proof}
First, we note that it suffices to prove that $\frac{d(x,y)}{d_{H}(x,y)}\leq(1+O(\varepsilon))\sqrt{m}$
for any $x,y\in M_{L}$ with $d(x,y)\leq\varepsilon$. Then, one can
run a limiting argument as follows. Let $x_{t}=t\cdot x+(1-t)\cdot y$,
then we have that
\[
d_{H}(x,y)=\lim_{n\rightarrow\infty}\sum_{i=0}^{n-1}d_{H}(x_{k/n},x_{(k+1)/n}).
\]
Since $d_{H}(x_{k/n},x_{(k+1)/n})=O_{x,y}(\frac{1}{n})$ , we have
that 
\begin{eqnarray*}
d_{H}(x,y) & = & \lim_{n\rightarrow\infty}\sum_{i=0}^{n-1}d_{H}(x_{k/n},x_{(k+1)/n})\\
 & \geq & \lim_{n\rightarrow\infty}\frac{(1-O_{x,y}(\frac{1}{n}))}{\sqrt{m}}\sum_{i=0}^{n-1}d(x_{k/n},x_{(k+1)/n})\\
 & = & \frac{1}{\sqrt{m}}\lim_{n\rightarrow\infty}\sum_{i=0}^{n-1}d(x_{k/n},x_{(k+1)/n})\\
 & \geq & \frac{1}{\sqrt{m}}d(x,y)=\frac{d(x,y)}{\sqrt{m}}.
\end{eqnarray*}

Now, we can assume $d(x,y)\leq\varepsilon$. Let $p$ and $q$ are
on the boundary of $M_{L}$ such that $p,x,y,q$ are on the straight
line $\overline{xy}$ and are in order. Without loss of generality,
we assume $p$ is closer to $x$. Then, we have that $p\in M_{L}\cap(x-M_{L})$,
equivalently, we have that $\left|a_{i}^{T}p-a_{i}^{T}x\right|\leq a_{i}^{T}x-b_{i}$
for all $i$ and hence $\norm{A_{x}(p-x)}_{\infty}\leq1$. Therefore,
we have that
\[
\norm{A_{x}(p-x)}_{2}\leq\sqrt{m}\norm{A_{x}(p-x)}_{\infty}\leq\sqrt{m}.
\]
Since $p,x,y$ are on the same line, we have that
\begin{eqnarray*}
\frac{\norm{x-y}_{2}\norm{p-q}_{2}}{\norm{p-x}_{2}\norm{y-q}_{2}} & \geq & \frac{\norm{x-y}_{2}}{\norm{p-x}_{2}}=\frac{\norm{A_{x}(x-y)}_{2}}{\norm{A_{x}(p-x)}_{2}}\geq\frac{\norm{A_{x}(x-y)}_{2}}{\sqrt{m}}.
\end{eqnarray*}

Since $d(x,y)<\varepsilon$, Lemma 3.1 in \cite{nesterov2002riemannian}
shows that
\[
d(x,y)\leq-\log\left(1-\norm{A_{x}(x-y)}_{2}\right).
\]
Hence, we have that $\norm{A_{x}(x-y)}_{2}\geq(1-O(\varepsilon))d(x,y)$
and hence
\[
d_{H}(x,y)=\frac{\norm{x-y}_{2}\norm{p-q}_{2}}{\norm{p-x}_{2}\norm{y-q}_{2}}\geq(1-O(\varepsilon))\frac{d(x,y)}{\sqrt{m}}.
\]
\end{proof}

\subsection{Geodesic walk on $M_{L}$}

\label{subsec:geodesic_walk_log}Recall that the geodesic walk is
given by

\[
x^{(\text{new})}=\exp_{x}(\sqrt{h}w+\frac{h}{2}\mu(x))
\]
where $w\sim N_{x}(0,I)$. In many proofs in this section, we consider
the geodesic $\gamma$ from $x$ with the initial velocity
\[
\gamma'(0)=\frac{w}{\sqrt{n}}+\frac{1}{2}\sqrt{\frac{h}{n}}\mu(x).
\]
The scaling is to make the speed of geodesic $\gamma$ close to one.
Since $w$ is a Gaussian vector, we have that $0.9\le\norm{\gamma'(0)}_{\gamma(0)}\le1.1$
with high probability. Due to this rescaling, the geodesic is defined
from $0$ to $\ell=\sqrt{nh}$.

We often work in Euclidean coordinates. In this case, the geodesic
walk is given by the following formula.
\begin{lem}
\label{lem:geodesic_walk_log_def}Given $x\in M_{L}$ and step size
$h>0$, one step of the geodesic walk starting at $x$ in Euclidean
coordinates is given by the solution $\gamma(\ell)$ of the following
geodesic equation
\begin{eqnarray*}
\gamma''(t) & = & \left(A_{\gamma}^{T}A_{\gamma}\right)^{-1}A_{\gamma}^{T}s_{\gamma'}^{2}\text{ for }0\leq t\leq\ell\\
\gamma'(0) & = & \frac{w}{\sqrt{n}}+\frac{1}{2}\sqrt{\frac{h}{n}}\mu(x)\\
\gamma(0) & = & x
\end{eqnarray*}
where $\ell=\sqrt{nh}$, $w\sim N(0,(A_{x}^{T}A_{x})^{-1})$ and $\mu(x)=\left(A_{x}^{T}A_{x}\right)^{-1}A_{x}^{T}\sigma_{x}$.
\end{lem}
\begin{proof}
The geodesic equation is given by Lemma \ref{lem:geo_equ_log}. From
(\ref{eq:drift_term_general}), the drift term is given by
\begin{align*}
\mu_{i}(x) & =\frac{1}{2}\sum_{j=1}^{n}\frac{\partial}{\partial x_{j}}\left(\left(\nabla^{2}\phi(X_{t})\right)^{-1}\right)_{ij}=\frac{1}{2}\sum_{j=1}^{n}\frac{\partial}{\partial x_{j}}\left(\left(A_{x}^{T}A_{x}\right)^{-1}\right)_{ij}\\
 & =\sum_{j}e_{i}^{T}(A_{x}^{T}A_{x})^{-1}A_{x}^{T}S_{x,e_{j}}A_{x}(A_{x}^{T}A_{x})^{-1}e_{j}\\
 & =\sum_{j}\sum_{k}V_{ik}e_{k}^{T}A_{x}e_{j}V_{jk}\quad\mbox{\mbox{where }}V=(A_{x}^{T}A_{x})^{-1}A_{x}^{T}\\
 & =\sum_{k}V_{ik}e_{k}^{T}A_{x}\sum_{j}e_{j}e_{j}^{T}(A_{x}^{T}A_{x})^{-1}A_{x}^{T}e_{k}\\
 & =\sum_{k}V_{ik}e_{k}^{T}A_{x}(A_{x}^{T}A_{x})^{-1}A_{x}^{T}e_{k}\\
 & =V_{i}\sigma_{x}=e_{i}^{T}(A_{x}^{T}A_{x})^{-1}A_{x}^{T}\sigma_{x}.
\end{align*}
\end{proof}

\subsection{Randomness and stability of the geodesic ($V_{0}$) \label{sec:walk_is_random_log}}

Many parameters of a Hessian manifold relates to how fast a geodesic
approaches the boundary of the polytope. Since the initial velocity
of the geodesic consists mainly the Gaussian part (plus a small drift
term), one can imagine that $\norm{s_{\gamma'(0)}}_{\infty}=O\left(\frac{1}{\sqrt{n}}\right)\norm{s_{\gamma'(0)}}_{2}$,
namely, the geodesic initial approaches/leaves every facet of the
polytope in roughly same slow pace. If this holds on the whole geodesic,
this would allowed us to give very tight bounds on various parameters.
Although we are not able to prove that $\norm{s_{\gamma'(t)}}_{\infty}$
is stable throughout $0\leq t\leq\ell$, we can show that $\norm{s_{\gamma'(t)}}_{4}$
is stable and that allows us to give a good bound on $\norm{s_{\gamma'(t)}}_{\infty}$. 

Throughout this section, we only use the randomness of geodesic to
prove that both $\norm{s_{\gamma'(t)}}_{4}$ and $\norm{s_{\gamma'(t)}}_{\infty}$
is small with high probability. Since $\norm{s_{\gamma'(t)}}_{4}=O(n^{-1/4})$
and $\norm{s_{\gamma'(t)}}_{\infty}=O(\sqrt{\frac{\log n}{n}}+\sqrt{h})$
with high probability (Lemma \ref{lem:V0_bound}), we define 
\[
V(\gamma)\defeq\max_{0\leq t\leq\ell}\left(\frac{\norm{s_{\gamma'(t)}}_{4}}{n^{-1/4}}+\frac{\norm{s_{\gamma'(t)}}_{\infty}}{\sqrt{\frac{\log n}{n}}+\sqrt{h}}\right)
\]
 to capture this randomness involves in generating the geodesic walk.
This allows us to perturb the geodesic (Lemma \ref{lem:one_one_cor})
without worrying about the dependence on randomness. 

Here, we first prove the the geodesic is stable in $L_{4}$ norm and
hence $V(\gamma)$ can be simply approximated by $\norm{s_{\gamma'(0)}}_{4}$
and $\norm{s_{\gamma'(0)}}_{\infty}$.
\begin{lem}
\label{lem:geodesic_4}Let $\gamma$ be a geodesic in $M_{L}$ starting
at $x$. Let $v_{4}=\norm{s_{\gamma'(0)}}_{4}$. Then, for $0\leq t\leq\frac{1}{10v_{4}}$,
we have that

\begin{enumerate}
\item $\norm{s_{\gamma'(t)}}_{4}\leq1.25v_{4}$.
\item $\norm{\gamma''(t)}_{\gamma(t)}^{2}\leq3v_{4}^{4}$.
\end{enumerate}
\end{lem}
\begin{proof}
Let $u(t)=\norm{s_{\gamma'(t)}}_{4}$. Then, we have
\begin{eqnarray}
\frac{du}{dt} & \leq & \norm{\frac{d}{dt}\left(A_{\gamma}\gamma'\right)}_{4}\nonumber \\
 & = & \norm{A_{\gamma}\gamma''-\left(A_{\gamma}\gamma'\right)^{2}}_{4}\nonumber \\
 & \leq & \norm{A_{\gamma}\gamma''}_{4}+u^{2}(t).\label{eq:geodesic_error_d}
\end{eqnarray}

Under the Euclidean coordinates, the geodesic equation is given by
\[
\gamma''=\left(A_{\gamma}^{T}A_{\gamma}\right)^{-1}A_{\gamma}^{T}s_{\gamma'}^{2}.
\]
Hence, we have that
\begin{eqnarray}
\norm{\gamma''}_{\gamma}^{2} & = & \left(s_{\gamma'}^{2}\right)^{T}A_{\gamma}\left(A_{\gamma}^{T}A_{\gamma}\right)^{-1}\left(A_{\gamma}^{T}A_{\gamma}\right)\left(A_{\gamma}^{T}A_{\gamma}\right)^{-1}A_{\gamma}^{T}s_{\gamma'}^{2}\nonumber \\
 & \leq & \sum_{i}(s_{\gamma'}^{4})_{i}=u^{4}(t).\label{eq:geodesic_changes}
\end{eqnarray}
Therefore, we have
\[
\norm{A_{\gamma}\gamma''}_{4}\leq\norm{A_{\gamma}\gamma''}_{2}\leq u^{2}(t).
\]

Plugging it into (\ref{eq:geodesic_error_d}), we have that
\[
\frac{du}{dt}\leq u^{2}(t)+u^{2}(t)\leq2u^{2}(t).
\]
Since $u(0)=v_{4}$, for all $t\geq0$, we have
\begin{equation}
u(t)\leq\frac{v_{4}}{1-2v_{4}t}.\label{eq:geodesic_inf_est}
\end{equation}
For $0\leq t\leq\frac{1}{10v_{4}}$, we have $u(t)\leq1.25v_{4}$
and this gives the first inequality. Using (\ref{eq:geodesic_changes}),
we get the second inequality.
\end{proof}
Using this, we prove that $V(\gamma)$ is small with high probability.
\begin{lem}
\label{lem:V0_bound}Assume that $h\leq\frac{1}{1000\sqrt{n}}$. We
have that
\[
P\left(V(\gamma)\leq24\right)\geq1-\frac{3}{n}
\]
where $\gamma$ is generated from geodesic walk.
\end{lem}
\begin{proof}
From the definition of the geodesic walk (Lemma \ref{lem:geodesic_walk_log_def}),
we have that
\[
\gamma'(0)=u+v
\]
where $u\sim N\left(0,\frac{1}{n}\left(A_{\gamma}^{T}A_{\gamma}\right)^{-1}\right)$
and $v=\frac{1}{2}\sqrt{\frac{h}{n}}\left(A_{\gamma}^{T}A_{\gamma}\right)^{-1}A_{\gamma}^{T}\sigma_{\gamma}$.
Therefore, we have
\begin{eqnarray*}
\norm{A_{\gamma}\gamma'(0)}_{4} & \leq & \norm{A_{\gamma}u}_{4}+\norm{\frac{1}{2}\sqrt{\frac{h}{n}}A_{\gamma}\left(A_{\gamma}^{T}A_{\gamma}\right)^{-1}A_{\gamma}^{T}\sigma_{\gamma}}_{2}\\
 & = & \norm{A_{\gamma}u}_{4}+\frac{1}{2}\sqrt{\frac{h}{n}}\sqrt{\sigma_{\gamma}^{T}A_{\gamma}\left(A_{\gamma}^{T}A_{\gamma}\right)^{-1}A_{\gamma}^{T}\sigma_{\gamma}}.
\end{eqnarray*}

Note that $A_{\gamma}u=Bx$ where $B=\frac{1}{\sqrt{n}}A_{\gamma}\left(A_{\gamma}^{T}A_{\gamma}\right)^{-1/2}$
and $x\sim N(0,I)$. Since $\sum_{i}\norm{e_{i}^{T}B}_{2}^{4}=\frac{1}{n^{2}}\sum_{i}(\sigma_{\gamma})_{i}^{2}\leq\frac{1}{n}$
and $\norm B_{2\rightarrow4}\leq\norm B_{2\rightarrow2}=\frac{1}{\sqrt{n}}$.
Hence, Lemma \ref{lem:norm_random_Ax} shows that
\[
P\left(\norm{A_{\gamma}u}_{4}^{4}\leq\left(\left(\frac{3}{n}\right)^{1/4}+\frac{t}{\sqrt{n}}\right)^{4}\right)\leq1-\exp\left(-\frac{t^{2}}{2}\right).
\]
Using our assumption on $h$, we have that
\begin{align*}
v_{4}\defeq\norm{A_{\gamma}\gamma'(0)}_{4} & \leq\frac{3^{1/4}+\sqrt{2}}{n^{1/4}}+\frac{1}{2}\sqrt{h}\leq\frac{3}{n^{1/4}}
\end{align*}
with probability at least $1-\exp(-\sqrt{n})$. Now, we apply Lemma
\ref{lem:geodesic_4} to get that
\[
\norm{s_{\gamma'(t)}}_{4}\leq1.25v_{4}\leq\frac{4}{n^{1/4}}
\]
for all $0\leq t\leq\ell\leq\frac{1}{10v_{4}}$. 

Next, we estimate $\norm{s_{\gamma'(t)}}_{\infty}$. Since $e_{i}^{T}A_{\gamma}u=e_{i}^{T}Bx\sim N(0,\frac{\sigma_{i}}{n})$,
we have
\[
P_{u}\left(\left|e_{i}^{T}A_{\gamma}u\right|\geq\sqrt{\frac{\sigma_{i}}{n}}t\right)\leq2\exp\left(-\frac{t^{2}}{2}\right).
\]
Hence, we have that
\[
P_{u}\left(\norm{A_{\gamma}u}_{\infty}\geq2\sqrt{\frac{\log n}{n}}\right)\leq2\sum_{i}\exp\left(-\frac{2\log n}{\sigma_{i}}\right)
\]
Since $\sum_{i}\exp\left(-\frac{2\log n}{\sigma_{i}}\right)$ is concave
, the maximum of $\sum_{i}\exp\left(-\frac{\log n}{\sigma_{i}}\right)$
on the feasible set $\{0\leq\sigma\leq1,\sum\sigma_{i}=n\}$ occurs
on its vertices. Hence, we have that
\[
P_{u}\left(\norm{A_{\gamma}u}_{\infty}\geq2\sqrt{\frac{\log n}{n}}\right)\leq2n\exp\left(-2\log n\right)=\frac{2}{n}.
\]
Therefore, with probability $1-\frac{2}{n}$, we have that
\begin{align*}
\norm{A_{\gamma}\gamma'(0)}_{\infty} & \leq\norm{A_{\gamma}u}_{\infty}+\norm{\frac{1}{2}\sqrt{\frac{h}{n}}A_{\gamma}\left(A_{\gamma}^{T}A_{\gamma}\right)^{-1}A_{\gamma}^{T}\sigma_{\gamma}}_{\infty}\\
 & \leq2\sqrt{\frac{\log n}{n}}+\frac{1}{2}\sqrt{\frac{h}{n}}\sqrt{\sigma_{\gamma}^{T}A_{\gamma}\left(A_{\gamma}^{T}A_{\gamma}\right)^{-1}A_{\gamma}^{T}\sigma_{\gamma}}\\
 & \leq2\sqrt{\frac{\log n}{n}}+\frac{1}{2}\sqrt{h}.
\end{align*}
Lemma \ref{lem:geodesic_4} shows that $\norm{A_{\gamma}\gamma''}_{\infty}\leq\norm{\gamma''}_{\gamma(t)}\leq\sqrt{3}v_{4}^{2}\leq16n^{-1/2}$.
Hence, for any $0\leq t\leq\ell$, we have that

\begin{eqnarray*}
\norm{s_{\gamma'(t)}}_{\infty} & \leq & \norm{A_{\gamma}\gamma'(0)}_{\infty}+\int_{0}^{\ell}\norm{A_{\gamma(t)}\gamma''(t)}_{\infty}dt\\
 & \leq & 2\sqrt{\frac{\log n}{n}}+\frac{1}{2}\sqrt{h}+16\ell n^{-1/2}\\
 & = & 20\left(\sqrt{\frac{\log n}{n}}+\sqrt{h}\right).
\end{eqnarray*}

Combining with our estimate on $\norm{s_{\gamma'(t)}}_{4}$, we have
that
\[
P\left(\frac{\norm{s_{\gamma'(t)}}_{4}}{n^{-1/4}}+\frac{\norm{s_{\gamma'(t)}}_{\infty}}{\sqrt{\frac{\log n}{n}}+\sqrt{h}}\geq24\right)\leq e^{-\sqrt{n}}+\frac{2}{n}\leq\frac{3}{n}.
\]
\end{proof}
Due to this lemma, we set
\begin{equation}
V_{0}=48.\label{eq:log_V0}
\end{equation}
Here, we collect some simple consequences of small $V(\gamma)$ that
we will use through this section.
\begin{lem}
\label{lem:gamma_est}Given a geodesic $\gamma$ on $M_{L}$ with
$V(\gamma)\leq V_{0}$. For any $0\leq t\leq\ell$,

\begin{enumerate}
\item $\norm{A_{\gamma}\gamma'(t)}_{4}\leq48n^{-1/4}$.
\item $\norm{A_{\gamma}\gamma'(t)}_{\infty}\leq48\left(\sqrt{\frac{\log n}{n}}+\sqrt{h}\right)$.
\item $\norm{A_{\gamma}\gamma'(t)}_{2}\leq48$.
\item $\norm{\gamma''(t)}_{\gamma}^{2}\leq10^{8}n^{-1}$.
\end{enumerate}
\end{lem}
\begin{proof}
The first two inequality simply follows from the definition of $V(\gamma)$.
The third inequality comes from the calculation
\[
\norm{A_{\gamma}\gamma'}_{2}\leq n^{1/4}\norm{A_{\gamma}\gamma'}_{4}\leq48.
\]
Since $\norm{A_{\gamma}\gamma'(0)}_{4}\leq48n^{-1/4}$, Lemma \ref{lem:geodesic_4}
shows the last inequality.
\end{proof}

\subsection{Stability of Drift ($D_{0}$, $D_{1}$ and $D_{2}$)\label{sec:drift_stable_log}}
\begin{lem}
\label{lem:log_D0}For any $x\in M_{L}$, we have that $\norm{\mu(x)}_{x}^{2}\le n.$
Hence, $D_{0}\leq\sqrt{n}$.
\end{lem}
\begin{proof}
We have
\begin{align*}
\norm{\mu(x)}_{x}^{2} & =\sigma_{x}^{T}A_{x}\left(A_{x}^{T}A_{x}\right)^{-1}A_{x}\sigma_{x}\leq\sum_{i}(\sigma_{x})_{i}^{2}\le n.
\end{align*}
\end{proof}
\begin{lem}
\label{lem:log_D1}Let $\gamma(t)$ be a geodesic on $M_{L}$ with
$V(\gamma)\leq V_{0}$. Then, we have that
\[
D_{1}=\sup_{0\leq t\leq\ell}\left|\frac{d}{dt}\norm{\mu(\gamma(t))}^{2}\right|=O(n\sqrt{h}+\sqrt{n\log n}).
\]
\end{lem}
\begin{proof}
Note that
\begin{eqnarray*}
\norm{\mu(\gamma(t))}_{\gamma}^{2} & = & \sigma_{\gamma}^{T}A_{\gamma}\left(A_{\gamma}^{T}A_{\gamma}\right)^{-1}A_{\gamma}\sigma_{\gamma}\\
 & = & 1^{T}\Diag(P_{\gamma})P_{\gamma}\Diag(P_{\gamma})1.
\end{eqnarray*}
Using $\frac{d}{dt}P_{\gamma}=-S_{\gamma'}P_{\gamma}-P_{\gamma}S_{\gamma'}+2P_{\gamma}S_{\gamma'}P_{\gamma}$,
we have
\begin{eqnarray*}
\frac{d}{dt}\norm{\mu(\gamma(t))}_{\gamma}^{2} & = & -2\cdot1^{T}\Diag(S_{\gamma'}P_{\gamma})P_{\gamma}\Diag(P_{\gamma})1\\
 &  & -2\cdot1^{T}\Diag(P_{\gamma}S_{\gamma'})P_{\gamma}\Diag(P_{\gamma})1\\
 &  & +4\cdot1^{T}\Diag(P_{\gamma}S_{\gamma'}P_{\gamma})P_{\gamma}\Diag(P_{\gamma})1\\
 &  & -1^{T}\Diag(P_{\gamma})S_{\gamma'}P_{\gamma}\Diag(P_{\gamma})1\\
 &  & -1^{T}\Diag(P_{\gamma})P_{\gamma}S_{\gamma'}\Diag(P_{\gamma})1\\
 &  & +2\cdot1^{T}\Diag(P_{\gamma})P_{\gamma}S_{\gamma'}P_{\gamma}\Diag(P_{\gamma})1\\
 & = & -2\cdot1^{T}\Diag(S_{\gamma'}P_{\gamma})P_{\gamma}\Diag(P_{\gamma})1\\
 &  & -4\cdot1^{T}\Diag(P_{\gamma})S_{\gamma'}P_{\gamma}\Diag(P_{\gamma})1\\
 &  & +4\cdot1^{T}\Diag(P_{\gamma}S_{\gamma'}P_{\gamma})P_{\gamma}\Diag(P_{\gamma})1\\
 &  & +2\cdot1^{T}\Diag(P_{\gamma})P_{\gamma}S_{\gamma'}P_{\gamma}\Diag(P_{\gamma})1.
\end{eqnarray*}
Now, we bound these 4 terms separately. Note that
\begin{align*}
\left|1^{T}\Diag(S_{\gamma'}P_{\gamma})P_{\gamma}\Diag(P_{\gamma})1\right| & \leq\sqrt{1^{T}\Diag(S_{\gamma'}P_{\gamma})P_{\gamma}\Diag(S_{\gamma'}P_{\gamma})1}\sqrt{1^{T}\Diag(P_{\gamma})P_{\gamma}\Diag(P_{\gamma})1}\\
 & \leq\sqrt{\sum_{i}(\sigma_{\gamma}s_{\gamma'})_{i}^{2}}\sqrt{\sum_{i}(\sigma_{\gamma})_{i}^{2}}\leq\norm{\gamma'}_{\gamma}\sqrt{n},
\end{align*}
\begin{align*}
\left|1^{T}\Diag(P_{\gamma})S_{\gamma'}P_{\gamma}\Diag(P_{\gamma})1\right| & \leq\sqrt{1^{T}\Diag(P_{\gamma})S_{\gamma'}P_{\gamma}S_{\gamma'}\Diag(P_{\gamma})1}\sqrt{1^{T}\Diag(P_{\gamma})P_{\gamma}\Diag(P_{\gamma})1}\\
 & \leq\sqrt{\sum_{i}(\sigma_{\gamma}s_{\gamma'})_{i}^{2}}\sqrt{\sum_{i}(\sigma_{\gamma})_{i}^{2}}\leq\norm{\gamma'}_{\gamma}\sqrt{n},
\end{align*}
\begin{align*}
\left|1^{T}\Diag(P_{\gamma}S_{\gamma'}P_{\gamma})P_{\gamma}\Diag(P_{\gamma})1\right| & \leq\sqrt{1^{T}\Diag(P_{\gamma}S_{\gamma'}P_{\gamma})P_{\gamma}\Diag(P_{\gamma}S_{\gamma'}P_{\gamma})1}\sqrt{1^{T}\Diag(P_{\gamma})P_{\gamma}\Diag(P_{\gamma})1}\\
 & \leq\sqrt{\sum_{i}(P_{\gamma}S_{\gamma'}P_{\gamma})_{ii}^{2}}\sqrt{\sum_{i}(\sigma_{\gamma})_{i}^{2}}\\
 & \leq\norm{S_{\gamma'}}_{\infty}\sum_{i}(\sigma_{\gamma})_{i}^{2}\leq\norm{S_{\gamma'}}_{\infty}n,
\end{align*}
\begin{align*}
1^{T}\Diag(P_{\gamma})P_{\gamma}S_{\gamma'}P_{\gamma}\Diag(P_{\gamma})1 & \leq\norm{S_{\gamma'}}_{\infty}1^{T}\Diag(P_{\gamma})P_{\gamma}P_{\gamma}\Diag(P_{\gamma})1\\
 & \leq\norm{S_{\gamma'}}_{\infty}\sum_{i}(\sigma_{\gamma})_{i}^{2}\leq\norm{S_{\gamma'}}_{\infty}n.
\end{align*}
Since $\norm{\gamma'}_{\gamma}\leq48$ and $\norm{A_{\gamma}\gamma'}_{\infty}\leq48\left(\sqrt{\frac{\log n}{n}}+\sqrt{h}\right)$
(Lemma \ref{lem:gamma_est}), we have
\[
D_{1}=\sup\left|\frac{d}{dt}\norm{\mu(\gamma(t))}^{2}\right|\leq288\sqrt{n}+288n\left(\sqrt{\frac{\log n}{n}}+\sqrt{h}\right)\leq576\sqrt{n\log n}+288n\sqrt{h}
\]
\end{proof}
\begin{lem}
\label{lem:log_D2}Given $x\in M_{L}$ and any curve $c(t)$ starting
at $x$ with unit speed. We have that
\[
\norm{D_{t}\mu(c(t))}=O(\sqrt{n}).
\]
Hence, $D_{2}=O(\sqrt{n})$.
\end{lem}
\begin{proof}
By Lemma \ref{lem:geo_equ_log}, we have that
\[
D_{t}\mu(c(t))=\frac{d\mu(c)}{dt}-(A_{c}^{T}A_{c})^{-1}A_{c}^{T}S_{c'}s_{c,\mu}.
\]
We bound the terms separately. 

For the second term, since $c$ is an unit speed curve, we have that
$\norm{s_{c'}}_{2}=1$ and 
\begin{eqnarray*}
\norm{(A_{c}^{T}A_{c})^{-1}A_{c}^{T}S_{c'}s_{c,\mu}}_{c} & \leq & \norm{S_{c'}s_{c,\mu}}_{2}\leq\norm{S_{c,\mu}}_{\infty}\\
 & = & \norm{A_{c}(A_{c}^{T}A_{c})^{-1}A_{c}^{T}\sigma_{c}}_{\infty}\\
 & \leq & \norm{A_{c}(A_{c}^{T}A_{c})^{-1}A_{c}^{T}\sigma_{c}}_{2}\\
 & \leq & \norm{\sigma_{c}}_{2}\leq\sqrt{n}.
\end{eqnarray*}

For the first term,
\begin{eqnarray*}
\frac{d}{dt}\left(A_{c}\mu(c)\right) & = & \frac{d}{dt}\left(P_{c}\diag(P_{c})\right)\\
 & = & -S_{c'}P_{c}\diag(P_{c})-P_{c}S_{c'}\diag(P_{c})+2P_{c}S_{c'}P_{c}\diag(P_{c})\\
 &  & -P_{c}\diag(S_{c'}P_{c})-P_{c}\diag(P_{c}S_{c'})+2P_{c}\diag(P_{c}S_{c'}P_{c}).
\end{eqnarray*}
Using $\norm{s_{c'}}_{2}=1$, we have
\begin{eqnarray*}
\norm{\frac{d}{dt}\mu(c)}_{c} & \leq & \norm{S_{c'}A_{c}\mu(c)}_{2}\\
 &  & +\norm{S_{c'}P_{c}\diag(P_{c})}_{2}+2\norm{P_{c}S_{c'}\diag(P_{c})}_{2}+2\norm{P_{c}S_{c'}P_{c}\diag(P_{c})}_{2}\\
 &  & +\norm{P_{c}\diag(P_{c}S_{c'})}_{2}+\norm{P_{c}\diag(P_{c}S_{c'}P_{c})}_{2}\\
 & \leq & 8\norm{S_{c'}}_{\infty}\sqrt{n}\leq8\sqrt{n}.
\end{eqnarray*}

Combining both terms, we have the result.
\end{proof}

\subsection{Smoothness of the metric ($G_{1}$)}

\label{subsec:smoothness_metric_log}
\begin{lem}
\label{lem:log_G1}Let $\gamma(t)$ be a geodesic on $M_{L}$ with
$V(\gamma)\leq V_{0}$. Let $f(t)=\log\det(A_{\gamma(t)}^{T}A_{\gamma(t)})$.
Then, we have 
\[
\sup_{0\leq t\leq\ell}\left|f'''(t)\right|=O\left(\sqrt{h}+\sqrt{\frac{\log n}{n}}\right).
\]
Hence, $G_{1}=O\left(\sqrt{\frac{\log n}{n}}+\sqrt{h}\right)$.
\end{lem}
\begin{proof}
Note that
\[
f'(t)=-2\tr(A_{\gamma}^{T}A_{\gamma})^{-1}A_{\gamma}^{T}S_{\gamma'}A_{\gamma}=-2\tr P_{\gamma}S_{\gamma'}.
\]
Since $\frac{d}{dt}P_{\gamma}=-S_{\gamma'}P_{\gamma}-P_{\gamma}S_{\gamma'}+2P_{\gamma}S_{\gamma'}P_{\gamma}$
and $\frac{d}{ds}S_{\gamma'}=-S_{\gamma'}^{2}+S_{\gamma''}$, we have
\begin{align*}
f''(t)= & 2\tr S_{\gamma'}P_{\gamma}S_{\gamma'}+2\tr P_{\gamma}S_{\gamma'}^{2}-4\tr P_{\gamma}S_{\gamma'}P_{\gamma}S_{\gamma'}\\
 & +2\tr P_{\gamma}S_{\gamma'}^{2}-2\tr P_{\gamma}S_{\gamma''}\\
= & -4\tr P_{\gamma}S_{\gamma'}P_{\gamma}S_{\gamma'}+6\tr P_{\gamma}S_{\gamma'}^{2}-2\tr P_{\gamma}S_{\gamma''}.
\end{align*}
So, we have
\begin{align*}
f'''(t)= & +8\tr S_{\gamma'}P_{\gamma}S_{\gamma'}P_{\gamma}S_{\gamma'}+8\tr P_{\gamma}S_{\gamma'}^{2}P_{\gamma}S_{\gamma'}-16\tr P_{\gamma}S_{\gamma'}P_{\gamma}S_{\gamma'}P_{\gamma}S_{\gamma'}\\
 & +8\tr P_{\gamma}S_{\gamma'}^{2}P_{\gamma}S_{\gamma'}-8\tr P_{\gamma}S_{\gamma''}P_{\gamma}S_{\gamma'}\\
 & -6\tr S_{\gamma'}P_{\gamma}S_{\gamma'}^{2}-6\tr P_{\gamma}S_{\gamma'}S_{\gamma'}^{2}+12\tr P_{\gamma}S_{\gamma'}P_{\gamma}S_{\gamma'}^{2}\\
 & -12\tr P_{\gamma}S_{\gamma'}^{3}+12\tr P_{\gamma}S_{\gamma'}S_{\gamma''}\\
 & +2\tr S_{\gamma'}P_{\gamma}S_{\gamma''}+2\tr P_{\gamma}S_{\gamma'}S_{\gamma''}-4\tr P_{\gamma}S_{\gamma'}P_{\gamma}S_{\gamma''}\\
 & +2\tr P_{\gamma}S_{\gamma'}S_{\gamma''}-2\tr P_{\gamma}S_{\gamma'''}\\
= & -16\tr P_{\gamma}S_{\gamma'}P_{\gamma}S_{\gamma'}P_{\gamma}S_{\gamma'}+36\tr P_{\gamma}S_{\gamma'}^{2}P_{\gamma}S_{\gamma'}\\
 & -12\tr P_{\gamma}S_{\gamma''}P_{\gamma}S_{\gamma'}-24\tr P_{\gamma}S_{\gamma'}^{3}\\
 & +16\tr P_{\gamma}S_{\gamma'}S_{\gamma''}-2\tr P_{\gamma}\frac{d}{dt}S_{\gamma''}.
\end{align*}
Hence, we have
\[
\left|f'''(t)\right|\leq(16+36+24)\sqrt{\tr S_{\gamma'}^{4}\tr S_{\gamma'}^{2}}+(12+16)\sqrt{\tr S_{\gamma''}^{2}\tr S_{\gamma'}^{2}}+2\left|\tr P_{\gamma}\frac{d}{dt}S_{\gamma''}\right|.
\]
Since $\tr S_{\gamma'}^{2}\leq48$, $\tr S_{\gamma'}^{4}\leq48^{4}n^{-1}$,
$\tr S_{\gamma''}^{2}\leq10^{8}n^{-1}$ (Lemma \ref{lem:gamma_est}),
we have

\begin{eqnarray*}
\left|f'''(t)\right| & = & O(n^{-1/2})+2\left|\tr P_{\gamma}\frac{d}{dt}S_{\gamma''}\right|.
\end{eqnarray*}

To bound the last term, we start with the geodesic equation:
\[
s_{\gamma''}=P_{\gamma}s_{\gamma'}^{2}.
\]
Since $\frac{d}{dt}P_{\gamma}=-S_{\gamma'}P_{\gamma}-P_{\gamma}S_{\gamma'}+2P_{\gamma}S_{\gamma'}P_{\gamma}$,
we have that
\begin{eqnarray*}
\frac{d}{dt}s_{\gamma''} & = & -S_{\gamma'}P_{\gamma}s_{\gamma'}^{2}-P_{\gamma}S_{\gamma'}s_{\gamma'}^{2}+2P_{\gamma}S_{\gamma'}P_{\gamma}s_{\gamma'}^{2}\\
 &  & +2P_{\gamma}s_{\gamma'}s_{\gamma''}-2P_{\gamma}s_{\gamma'}^{3}\\
 & = & -S_{\gamma'}P_{\gamma}s_{\gamma'}^{2}-3P_{\gamma}s_{\gamma'}^{3}+2P_{\gamma}S_{\gamma'}P_{\gamma}s_{\gamma'}^{2}+2P_{\gamma}s_{\gamma'}s_{\gamma''}\\
 & = & -S_{\gamma'}P_{\gamma}s_{\gamma'}^{2}-3P_{\gamma}s_{\gamma'}^{3}+4P_{\gamma}S_{\gamma'}P_{\gamma}s_{\gamma'}^{2}.
\end{eqnarray*}
Hence, we have that
\begin{eqnarray*}
\left|\tr P_{\gamma}\frac{d}{dt}S_{\gamma''}\right| & \leq & \left|\sigma_{\gamma}^{T}S_{\gamma'}P_{\gamma}s_{\gamma'}^{2}\right|+3\left|\sigma_{\gamma}^{T}P_{\gamma}s_{\gamma'}^{3}\right|+4\left|\sigma_{\gamma}^{T}P_{\gamma}S_{\gamma'}P_{\gamma}s_{\gamma'}^{2}\right|\\
 & = & \left|(S_{\gamma'}\sigma_{\gamma})^{T}P_{\gamma}s_{\gamma'}^{2}\right|+3\left|(S_{\gamma'}P_{\gamma}\sigma_{\gamma})^{T}s_{\gamma'}^{2}\right|+4(S_{\gamma'}P_{\gamma}\sigma_{\gamma})^{T}(P_{\gamma}s_{\gamma'}^{2})\\
 & \leq & \sqrt{\sum s_{\gamma'}^{2}\sum s_{\gamma'}^{4}}+\sqrt{\sum s_{\gamma'}^{4}}\sqrt{\sigma_{\gamma}^{T}P_{\gamma}S_{\gamma'}^{2}P_{\gamma}\sigma_{\gamma}}\\
 & = & O\left(n^{-1/2}\right)+O\left(n^{-1}\right)O\left(\sqrt{\frac{\log n}{n}}+\sqrt{h}\right)O\left(n\right)\\
 & \leq & O\left(\sqrt{\frac{\log n}{n}}+\sqrt{h}\right).
\end{eqnarray*}

Hence, we have that
\begin{eqnarray*}
\left|f'''(t)\right| & \leq & O\left(\sqrt{\frac{\log n}{n}}+\sqrt{h}\right).
\end{eqnarray*}
\end{proof}

\subsection{Stability of Curvatures ($R_{1}$ and $R_{2}$)\label{subsec:Stability-of-Jacobian-log-barrier}}
\begin{lem}
\label{lem:total_ricci2}Let $\gamma(t)$ be a geodesic on $M_{L}$
with $V(\gamma)\leq V_{0}$. Then, we have that
\[
\sup_{0\leq t\leq\ell}\norm{R(t)}_{F}=O(n^{-1/2})
\]
where $R(t)$ defined in Definition \ref{def:Rt_definition}. Hence,
$R_{1}=O(n^{-1/2})$.
\end{lem}
\begin{proof}
Let $\overline{R(t)}$ such that $\left\langle R(u,\gamma'(t))\gamma'(t),v\right\rangle =u^{T}\overline{R(t)}v$.
Lemma \ref{lem:Riemann_tensor_log} shows that 
\begin{eqnarray*}
u^{T}\overline{R(t)}v & = & s_{u}^{T}S_{\gamma'}P_{\gamma}S_{\gamma'}s_{v}-(s_{v}s_{u})^{T}P_{\gamma}s_{\gamma'}^{2}.
\end{eqnarray*}
Hence, we have
\begin{eqnarray*}
\overline{R(t)} & = & A_{\gamma}^{T}S_{\gamma'}P_{\gamma}S_{\gamma'}A_{\gamma}-A_{\gamma}^{T}\Diag(P_{\gamma}s_{\gamma'}^{2})A_{\gamma}.
\end{eqnarray*}
Pick $X_{i}=(A^{T}A_{\gamma})^{-1/2}e_{i}$. Then, we can write $R(t)$
in the coordinate systems $\{X_{i}\}_{i}$ and get
\begin{eqnarray*}
R(t) & = & (A_{\gamma}^{T}A_{\gamma})^{-1/2}\left(A_{\gamma}^{T}S_{\gamma'}P_{\gamma}S_{\gamma'}A_{\gamma}-A_{\gamma}^{T}\Diag(P_{\gamma}s_{\gamma'}^{2})A_{\gamma}\right)(A_{\gamma}^{T}A_{\gamma})^{-1/2}.
\end{eqnarray*}
Therefore, we have that
\begin{eqnarray*}
 &  & \norm{R(t)}_{F}^{2}\\
 & \leq & 2\norm{(A_{\gamma}^{T}A_{\gamma})^{-1/2}A_{\gamma}^{T}S_{\gamma'}P_{\gamma}S_{\gamma'}A_{\gamma}(A_{\gamma}^{T}A_{\gamma})^{-1/2}}_{F}^{2}+2\norm{(A_{\gamma}^{T}A_{\gamma})^{-1/2}A_{\gamma}^{T}\Diag(P_{\gamma}s_{\gamma'}^{2})A_{\gamma}(A_{\gamma}^{T}A_{\gamma})^{-1/2}}_{F}^{2}\\
 & = & 2\tr P_{\gamma}S_{\gamma'}P_{\gamma}S_{\gamma'}P_{\gamma}S_{\gamma'}P_{\gamma}S_{\gamma'}+2\tr P_{\gamma}\Diag(P_{\gamma}s_{\gamma'}^{2})P_{\gamma}\Diag(P_{\gamma}s_{\gamma'}^{2})\\
 & \leq & 4\norm{s_{\gamma'}}_{4}^{4}.
\end{eqnarray*}
The claim follows from Lemma \ref{lem:gamma_est}.
\end{proof}
\begin{lem}
\label{lem:log_R2}Given a geodesic $\gamma(t)$ on $M_{L}$ with
$V(\gamma)\leq V_{0}$. Assume that $h\leq\frac{1}{\sqrt{n}}$. For
any $t$ such that $0\leq t\leq\ell$, any curve $c(s)$ starting
from $\gamma(t)$ and any vector field $v(s)$ on $c(s)$ with $v(0)=\gamma'(t)$,
we have that
\[
\left|\frac{d}{ds}\left.Ric(v(s))\right|_{s=0}\right|\leq O\left(1+\sqrt{\frac{\log n}{nh}}\right)\left(\norm{\frac{dc}{ds}}+\ell\norm{D_{s}v}\right).
\]
Therefore, $R_{2}=O(1+\sqrt{\frac{\log n}{nh}})$. \todo{Check this. Changed a lot.}
\end{lem}
\begin{proof}
By Lemma \ref{lem:Riemann_tensor_log}, we know that
\begin{eqnarray*}
Ric(v(s)) & = & s_{c(s),v(s)}^{T}P_{c(s)}^{(2)}s_{c(s),v(s)}-\sigma_{c(s)}^{T}P_{c(s)}s_{c(s),v(s)}^{2}\\
 & = & \tr(S_{c(s),v(s)}P_{c(s)}S_{c(s),v(s)}P_{c(s)})-\tr(\Diag(P_{c(s)}s_{c(s),v(s)}^{2})P_{c(s)}).
\end{eqnarray*}
Note that the ``s'' in $c(s)$ is the parameter of the curve $s$
while the first letter ``s'' in both $s_{c(s)}$ and $S_{c(s)}$
denotes the slack. For simplicity, we suppress the parameter $s$
and hence, we have
\[
Ric(v)=\tr(S_{c,v}P_{c}S_{c,v}P_{c})-\tr(\Diag(P_{c}S_{c,v}^{2})P_{c}).
\]
We write $\frac{d}{ds}c=c_{s}$ and $\frac{d}{ds}v=v_{s}$ (in Euclidean
coordinate). Since $\frac{d}{ds}P_{c}=-S_{c_{s}}P_{c}-P_{c}S_{c_{s}}+2P_{c}S_{c_{s}}P_{c}$
and $\frac{d}{ds}S_{c,v}=-S_{c_{s}}S_{c,v}+S_{c,v_{s}}$, we have
that
\begin{eqnarray*}
 &  & \frac{d}{ds}Ric(v)\\
 & = & -2\tr(S_{c,v}S_{c_{s}}P_{c}S_{c,v}P_{c})-2\tr(S_{c,v}P_{c}S_{c_{s}}S_{c,v}P_{c})+4\tr(S_{c,v}P_{c}S_{c_{s}}P_{c}S_{c,v}P_{c})\\
 &  & -2\tr(S_{c_{s}}S_{c,v}P_{c}S_{c,v}P_{c})+2\tr(S_{c,v_{s}}P_{c}S_{c,v}P_{c})\\
 &  & +\tr(\Diag(P_{c}s_{c,v}^{2})S_{c_{s}}P_{c})+\tr(\Diag(P_{c}s_{c,v}^{2})P_{c}S_{c_{s}})-2\tr(\Diag(P_{c}s_{c,v}^{2})P_{c}S_{c_{s}}P_{c})\\
 &  & +\tr(\Diag(P_{c}S_{c_{s}}s_{c,v}^{2})P_{c})+\tr(\Diag(S_{c_{s}}P_{c}s_{c,v}^{2})P_{c})-2\tr(\Diag(P_{c}S_{c_{s}}P_{c}s_{c,v}^{2})P_{c})\\
 &  & +2\tr(\Diag(P_{c}S_{c,v}S_{c_{s}}s_{c,v})P_{c})-2\tr(\Diag(P_{c}S_{c,v}s_{c,v_{s}})P_{c})\\
 & = & -6\tr(S_{c,v}S_{c_{s}}P_{c}S_{c,v}P_{c})+4\tr(S_{c,v}P_{c}S_{c_{s}}P_{c}S_{c,v}P_{c})+2\tr(S_{c,v_{s}}P_{c}S_{c,v}P_{c})\\
 &  & +3\tr(\Diag(P_{c}s_{c,v}^{2})S_{c_{s}}P_{c})-2\tr(\Diag(P_{c}s_{c,v}^{2})P_{c}S_{c_{s}}P_{c})\\
 &  & +3\tr(\Diag(P_{c}S_{c_{s}}s_{c,v}^{2})P_{c})-2\tr(\Diag(P_{c}S_{c_{s}}P_{c}s_{c,v}^{2})P_{c})\\
 &  & -2\tr(\Diag(P_{c}S_{c,v}s_{c,v_{s}})P_{c}).
\end{eqnarray*}
Let $\frac{d}{ds}Ric(v)=(1)+(2)$ where $(1)$ is the sum of all terms
not involving $v_{s}$ and $(2)$ is the sum of other terms. 

For the first term $(1)$, we have that
\begin{eqnarray*}
\left|(1)\right| & \leq & 6\left|\tr(S_{c,v}S_{c_{s}}P_{c}S_{c,v}P_{c})\right|+4\left|\tr(S_{c,v}P_{c}S_{c_{s}}P_{c}S_{c,v}P_{c})\right|\\
 &  & +3\left|\tr(\Diag(P_{c}s_{c,v}^{2})S_{c_{s}}P_{c})\right|+2\left|\tr(\Diag(P_{c}s_{c,v}^{2})P_{c}S_{c_{s}}P_{c})\right|\\
 &  & +3\left|\tr(\Diag(P_{c}S_{c_{s}}s_{c,v}^{2})P_{c})\right|+2\left|\tr(\Diag(P_{c}S_{c_{s}}P_{c}s_{c,v}^{2})P_{c})\right|\\
 & \leq & 6\norm{S_{c_{s}}}_{\infty}\sqrt{\sum_{i}(s_{c,v})_{i}^{2}}\sqrt{\sum_{i}(s_{c,v})_{i}^{2}}+4\norm{S_{c_{s}}}_{\infty}\left|\tr(P_{c}S_{c,v}P_{c}S_{c,v}P_{c})\right|\\
 &  & +3\sqrt{\sum_{i}(s_{c,v})_{i}^{4}}\norm{S_{c_{s}}}_{2}+2\sqrt{\sum_{i}(s_{c,v})_{i}^{4}}\sqrt{\sum_{i}(P_{c}S_{c_{s}}P_{c})_{ii}^{2}}\\
 &  & +3\norm{S_{c_{s}}}_{\infty}\sqrt{\sum_{i}(s_{c,v})_{i}^{4}}\sqrt{\sum_{i}(P_{c})_{ii}^{2}}+2\norm{S_{c_{s}}}_{\infty}\sqrt{\sum_{i}(s_{c,v})_{i}^{4}}\sqrt{\sum_{i}(P_{c})_{ii}^{2}}\\
 & \leq & 10\norm{S_{c_{s}}}_{\infty}\norm{S_{c,v}}_{2}^{2}+3\norm{S_{c,v}}_{4}^{2}\norm{S_{c_{s}}}_{2}+7\norm{S_{c_{s}}}_{\infty}\norm{S_{c,v}}_{4}^{2}\sqrt{n}\\
 & \leq & 20\norm{S_{c_{s}}}_{2}\norm{S_{c,v}}_{4}^{2}\sqrt{n}.
\end{eqnarray*}
Since $s_{c,v}=s_{\gamma'}$ at $s=0$, we have that $\norm{S_{c,v}}_{4}^{2}\leq48^{2}n^{-1/2}$
and hence
\[
\left|(1)\right|=O\left(1\right)\norm{s_{c_{s}}}_{2}.
\]

For the second term $(2)$, we have that
\begin{eqnarray*}
\left|(2)\right| & \leq & 2\left|\tr(S_{c,v_{s}}P_{c}S_{c,v}P_{c})\right|+2\left|\tr(\Diag(P_{c}S_{c,v}s_{c,v_{s}})P_{c})\right|\\
 & \leq & 2\norm{s_{c,v_{s}}}_{2}\norm{s_{c,v}}_{2}+2\sqrt{n}\sqrt{\sum_{i}(s_{c,v_{s}}s_{c,v})_{i}^{2}}\\
 & \leq & O\left(1\right)\norm{s_{c,v_{s}}}_{2}+O\left(\sqrt{\log n}+\sqrt{nh}\right)\norm{s_{c,v_{s}}}_{2}\\
 & = & O\left(\sqrt{\log n}+\sqrt{nh}\right)\norm{s_{c,v_{s}}}_{2}
\end{eqnarray*}
where we used $\norm{s_{c,v}}_{\infty}=\norm{s_{\gamma'}}_{\infty}=O\left(\sqrt{\frac{\log n}{n}}+\sqrt{h}\right)$
and $\norm{s_{c,v}}_{2}=\norm{s_{\gamma'}}_{2}=O(1)$ at $s=0$ in
the second last line.

Note that at $s=0$, we have
\begin{align*}
D_{s}v & =\frac{dv}{ds}-\left(A_{c}^{T}A_{c}\right)^{-1}A_{c}^{T}S_{c_{s}}s_{c,v}.
\end{align*}
Therefore, we have
\[
s_{c,v_{s}}=A_{c}\left(D_{s}v\right)-A_{c}\left(A_{c}^{T}A_{c}\right)^{-1}A_{c}^{T}S_{c_{s}}s_{c,v}
\]
and hence
\begin{eqnarray*}
\norm{s_{c,v_{s}}}_{2} & \leq & \norm{D_{s}v}+\norm{A_{c}\left(A_{c}^{T}A_{c}\right)^{-1}A_{c}^{T}S_{c_{s}}s_{c,v}}_{2}\\
 & \leq & \norm{D_{s}v}+\norm{s_{\gamma'}}_{\infty}\norm{s_{c_{s}}}_{2}
\end{eqnarray*}
Therefore, we have
\begin{align*}
\left|(2)\right| & =O\left(\sqrt{\log n}+\sqrt{nh}\right)\left(\norm{D_{s}v}+\left(\sqrt{\frac{\log n}{n}}+\sqrt{h}\right)\norm{s_{c_{s}}}_{2}\right).
\end{align*}

Therefore, we have
\begin{align*}
\left|\frac{d}{ds}\left.Ric(v(s))\right|_{s=0}\right|= & O\left(1\right)\norm{s_{c_{s}}}_{2}+O\left(\sqrt{\log n}+\sqrt{nh}\right)\norm{D_{s}v}\\
 & +O\left(\sqrt{\log n}+\sqrt{nh}\right)\left(\sqrt{\frac{\log n}{n}}+\sqrt{h}\right)\norm{s_{c_{s}}}_{2}\\
= & O\left(1+\sqrt{\frac{\log n}{nh}}\right)\left(\norm{\frac{dc}{ds}}+\ell\norm{D_{s}v}\right).
\end{align*}
where we used $h\leq\frac{1}{\sqrt{n}}$ at the last line.
\end{proof}

\subsection{\label{subsec:Stability-of-norm}Stability of $L_{4}+L_{\infty}$
norm ($V_{1}$) }
\begin{lem}
\label{lem:log_V1}Given a family of geodesic $\gamma_{s}(t)$ on
$M_{L}$ with $V(\gamma_{0})\leq V_{0}$. Suppose that $h\leq\frac{1}{\sqrt{n}}$,
we have that
\[
\left|\frac{d}{ds}V(\gamma_{s})\right|\leq O\left(\frac{1}{\sqrt{n}h+\sqrt{h\log n}}\right)\left(\norm{\frac{d}{ds}\gamma_{s}(0)}_{\gamma_{s}(0)}+\ell\norm{D_{s}\gamma_{s}'}_{\gamma_{s}(0)}\right).
\]
Hence, we have that $V_{1}=O\left(\frac{1}{\sqrt{n}h+\sqrt{h\log n}}\right)$.
\end{lem}
\begin{proof}
Since $\frac{d}{ds}s_{\gamma'}=-s_{\gamma'}A_{\gamma}\frac{d}{ds}\gamma+A_{\gamma}\frac{d}{ds}\gamma'$,
we have 
\begin{align*}
\norm{\frac{d}{ds}s_{\gamma'}}_{2} & \leq\norm{s_{\gamma'}}_{\infty}\norm{A_{\gamma}\frac{d}{ds}\gamma}_{2}+\norm{A_{\gamma}\frac{d}{ds}\gamma'}_{2}.
\end{align*}
For the last term, we note that
\begin{align*}
D_{s}\gamma' & =\frac{d\gamma'}{ds}-\left(A_{\gamma}^{T}A_{\gamma}\right)^{-1}A_{\gamma}^{T}S_{\frac{d\gamma}{ds}}s_{\gamma'}.
\end{align*}
Hence, we have
\[
\norm{A_{\gamma}\frac{d}{ds}\gamma'}_{2}\leq\norm{D_{s}\gamma'}+\norm{S_{\frac{d\gamma}{ds}}s_{\gamma'}}_{2}\leq\norm{D_{s}\gamma'}+\norm{s_{\gamma'}}_{\infty}\norm{S_{\frac{d\gamma}{ds}}}_{2}.
\]
Therefore, we have that
\begin{align}
\norm{\frac{d}{ds}s_{\gamma'}}_{2} & \leq2\norm{s_{\gamma'}}_{\infty}\norm{A_{\gamma}\frac{d}{ds}\gamma}_{2}+\norm{D_{s}\gamma'}\nonumber \\
 & =O\left(\sqrt{\frac{\log n}{n}}+\sqrt{h}\right)\norm{\frac{d}{ds}\gamma}+\norm{D_{s}\gamma'}.\label{eq:log_d_s_s_gamma}
\end{align}

By Lemma \ref{lem:total_ricci2}, we have that $\norm{R(t)}_{F}=O(n^{-1/2})$.
Since $\gamma_{s}$ is a family of geodesic, $\frac{d}{ds}\gamma(t)$
is a Jacobi field and Lemma \ref{lem:jacobi_approx_sol} shows that
\[
\norm{\frac{d}{ds}\gamma(t)-\alpha-\beta t}\leq\lambda t^{2}\norm{\frac{d}{ds}\gamma_{s}(0)}+\frac{\lambda t^{3}}{5}\norm{D_{s}\gamma_{s}'(0)}
\]
and
\[
\norm{D_{t}\frac{d}{ds}\gamma(t)-\beta}\leq2\lambda t\norm{\frac{d}{ds}\gamma_{s}(0)}+\frac{3\lambda t^{2}}{5}\norm{D_{s}\gamma_{s}'(0)}
\]
with $\lambda=O(n^{-1/2})$. Using $h\leq\frac{1}{\sqrt{n}}$, for
any $0\leq t\leq\ell$, we have that
\[
\norm{\frac{d}{ds}\gamma(t)}\leq O\left(1\right)\left(\norm{\frac{d}{ds}\gamma_{s}(0)}+\ell\norm{D_{s}\gamma_{s}'(0)}\right)
\]
and
\[
\ell\norm{D_{t}\frac{d}{ds}\gamma(t)}\leq O\left(1\right)\left(\norm{\frac{d}{ds}\gamma_{s}(0)}+\ell\norm{D_{s}\gamma_{s}'(0)}\right).
\]
Putting these into (\ref{eq:log_d_s_s_gamma}) and using $h\leq\frac{1}{\sqrt{n}}$,
we have
\begin{align*}
\norm{\frac{d}{ds}s_{\gamma'}}_{2} & =O\left(\sqrt{\frac{\log n}{n}}+\sqrt{h}+\frac{1}{\sqrt{nh}}\right)\left(\norm{\frac{d}{ds}\gamma_{s}(0)}+\ell\norm{D_{s}\gamma_{s}'(0)}\right)\\
 & =O\left(\frac{1}{\sqrt{nh}}\right)\left(\norm{\frac{d}{ds}\gamma_{s}(0)}+\ell\norm{D_{s}\gamma_{s}'(0)}\right).
\end{align*}

Hence, we have that
\begin{align*}
\frac{d}{ds}\left(\frac{\norm{s_{\gamma'(t)}}_{4}}{n^{-1/4}}+\frac{\norm{s_{\gamma'(t)}}_{\infty}}{\sqrt{\frac{\log n}{n}}+\sqrt{h}}\right)\leq & \frac{\norm{\frac{d}{ds}s_{\gamma'(t)}}_{4}}{n^{-1/4}}+\frac{\norm{\frac{d}{ds}s_{\gamma'(t)}}_{\infty}}{\sqrt{\frac{\log n}{n}}+\sqrt{h}}\\
= & O\left(\frac{1}{\sqrt{\frac{\log n}{n}}+\sqrt{h}}\right)\norm{\frac{d}{ds}s_{\gamma'(t)}}_{2}\\
= & O\left(\frac{1}{\sqrt{n}h+\sqrt{h\log n}}\right)\left(\norm{\frac{d}{ds}\gamma_{s}(0)}+\ell\norm{D_{s}\gamma_{s}'(0)}\right).
\end{align*}
Hence, we have that the result.
\end{proof}

\subsection{Mixing Time}
\begin{proof}[Proof of Theorem \ref{thm:gen-convergence-log}]
In the last previous sections, we proved that if $h\leq\frac{1}{1000\sqrt{n}}$

\begin{enumerate}
\item $V_{0}=48$ (\ref{eq:log_V0})
\item $V_{1}=O\left(\frac{1}{\sqrt{n}h+\sqrt{h\log n}}\right)$ (Lemma \ref{lem:log_V1})
\item $D_{0}=O(\sqrt{n})$ (Lemma \ref{lem:log_D0})
\item $D_{1}=O(n\sqrt{h}+\sqrt{n\log n})$ (Lemma \ref{lem:log_D1})
\item $D_{2}=O(\sqrt{n})$ (Lemma \ref{lem:log_D2})
\item $G_{1}=O(\sqrt{h}+\sqrt{\frac{\log n}{n}})$ (Lemma \ref{lem:log_G1})
\item $G_{2}=O(\sqrt{m})$ (Lemma \ref{lem:distances_log})
\item $R_{1}=O(1/\sqrt{n})$ (Lemma \ref{lem:total_ricci2})
\item $R_{2}=O(1+\sqrt{\frac{\log n}{nh}})$ (Lemma \ref{lem:log_R2})
\end{enumerate}
Lemma \ref{lem:V0_bound} proved that 
\[
P\left(V(\gamma)\leq24\right)\geq1-\frac{3}{n}.
\]
Therefore, if we set $h=\Omega(n^{-1})$, we have that
\[
P\left(V(\gamma)\leq24\right)\geq1-\frac{V_{0}}{100V_{1}}.
\]
Hence, these are valid parameters for $M_{L}$ with $H=\frac{1}{1000\sqrt{n}}$.

Theorem \ref{thm:gen-convergence} implies that the walk has mixing
time $O(G_{2}^{2}/h)$ as long as
\[
h\le\Theta(1)\min\left\{ \frac{1}{n^{1/3}D_{1}^{2/3}},\frac{1}{D_{2}},\frac{1}{nR_{1}},\frac{1}{(nD_{0}R_{1})^{2/3}},\frac{1}{\left(nR_{2}\right)^{2/3}},\frac{1}{nG_{1}^{2/3}},H\right\} \leq\frac{1}{\Theta(n^{3/4})}
\]
if $n$ is large enough.
\end{proof}
\begin{rem*}
The bottleneck of the proof occurs in the parameters $D_{1}$ and
$G_{1}$.
\end{rem*}

\pagebreak{}
\section{Collocation Method for ODE\label{subsec:CollocationMethod}}

In this section, we study a collocation method for solving ordinary
differential equation (ODE) and show how to solve a nice enough ODE
in nearly constant number of iterations without explicitly computing
higher derivatives.

Collocation method is a general framework for solving differential
equations. This framework finds a solution to the differential equation
by finding a low degree polynomial (or other finite dimensional space
that approximate the function space) that satisfies the differential
equation at a set of predesignated points (called collocation points).
By choosing the finite dimensional space and the collocation points
carefully, one can make sure there is an unique solution and that
the solution can be found using simple iterative methods (See \cite[Sec 3.4]{iserles2009first}
for an introduction). One key departure of our analysis to the standard
analysis is that we use $L_{4}$ norm (instead of the standard $L_{2}$
norm). This is essential to take advantage of the stability of the
$L_{4}$ norm of the geodesic.

\subsection{First Order ODE}

We first consider the following first order ODE
\begin{eqnarray}
\frac{d}{dt}u(t) & = & F(u(t),t),\text{ for }0\leq t\leq\ell\label{eq:ODE}\\
u(0) & = & v\nonumber 
\end{eqnarray}
where $F:\R^{n+1}\rightarrow\Rn$ and $u(t)\in\Rn$. The idea of collocation
methods is to find a degree $d$ polynomial $p$ such that 

\begin{eqnarray}
\frac{d}{dt}p(t) & = & F(p(t),t),\text{ for }t=c_{1},c_{2},\cdots,c_{d}\label{eq:ODE_discrete}\\
p(0) & = & v\nonumber 
\end{eqnarray}
where $c_{1},c_{2},\cdots,c_{d}$ are carefully chosen distinct points
on $[0,\ell]$. Here, we call $p:\R\rightarrow\Rn$ is a degree $d$
polynomial if $p(t)=[p_{1}(t);p_{2}(t);\cdots;p_{n}(t)]$ and each
$p_{i}(t)$ is an univariate polynomial with degree at most $d$.
The first part of the proof shows the existence of a solution for
the systems (\ref{eq:ODE_discrete}). To describe the algorithm, it
is easier to consider an equivalent integral equation.
\begin{lem}
\label{lem:two_ODE_version}Given distinct points $c_{1},c_{2},\cdots,c_{d}\in\R$
and $F:\R^{n+1}\rightarrow\Rn$, consider the nonlinear map $T:\R^{d\times n}\rightarrow\R^{d\times n}$
defined by
\[
T(\zeta)_{(i,k)}=\int_{0}^{c_{i}}\sum_{j=1}^{d}F(\zeta_{j},c_{j})_{k}\phi_{j}(s)ds\text{ for }i\in[d],k\in[n]
\]
where $\phi_{i}(s)=\prod_{j\neq i}\frac{s-c_{j}}{c_{i}-c_{j}}$ are
the Lagrange basis polynomials. Given any $\zeta\in\R^{d\times n}$
such that
\begin{equation}
\zeta_{i}=v+T(\zeta)_{i},\text{ for }i\in[d]\label{eq:ODE_fix_point}
\end{equation}
the polynomial
\[
p(t)=v+\int_{0}^{t}\sum_{j=1}^{d}F(\zeta_{j},c_{j})\phi_{j}(s)ds
\]
is a solution of the system (\ref{eq:ODE_discrete}). 
\end{lem}
\begin{proof}
Define the polynomials $\phi_{i}(s)=\prod_{j\neq i}\frac{s-c_{j}}{c_{i}-c_{j}}.$
Note that $\phi_{i}(c_{j})=\delta_{ij}$. Therefore, we have
\[
\sum_{j=1}^{d}\alpha_{j}\phi_{j}(c_{i})=\alpha_{i}.
\]
Therefore, $p(0)=v$ and
\[
\frac{d}{dt}p(c_{i})=\sum_{j=1}^{d}F(\zeta_{j},c_{j})\phi_{j}(c_{i})=F(\zeta_{j},c_{j}).
\]
Since $\zeta_{i}=v+T(\zeta)_{i}$, we have that
\[
\zeta_{i}=v+\int_{0}^{c_{i}}\sum_{j=1}^{d}F(\zeta_{j},c_{j})\phi_{j}(s)ds=p(c_{i}).
\]
Hence, we have $\frac{d}{dt}p(c_{i})=F(p(c_{i}),c_{i}).$ Therefore,
$p$ is a solution to the system (\ref{eq:ODE_discrete}).
\end{proof}
From Lemma \ref{lem:two_ODE_version}, we see that it suffices to
solve the system (\ref{eq:ODE_fix_point}). We solve it by a simple
fix point iteration shown in Algorithm \ref{alg:ode}.

\begin{algorithm2e}
\caption{$\texttt{CollocationMethod}$}

\label{alg:ode}

\SetAlgoLined

\textbf{Input:} An ordinary differential equation $\frac{d}{dt}u(t)=F(u(t),t)\text{ for }0\leq t\leq\ell$
with initial condition $u(0)=v$. 

Define $T(\zeta)_{(i,k)}=\int_{0}^{c_{i}}\sum_{j=1}^{d}F(\zeta_{j},c_{j})_{k}\phi_{j}(s)ds\text{ for }i\in[d],k\in[n]$
as defined in Lemma \ref{lem:two_ODE_version}.

$c_{i}=\frac{\ell}{2}+\frac{\ell}{2}\cos(\frac{2i-1}{2d}\pi)$ for
all $i\in[d]$, $K=4000\ell\max_{t\in[0,\ell]}\norm{F(v,t)}_{p}$,
$Z=\log_{2}(K/\varepsilon)$

$\zeta_{i}^{(0)}=v+T(\overline{v})_{i}$ for all $i\in[d]$ where
$\overline{v}=(v,v,\cdots,v)\in\R^{d\times n}$.

\For{$z=0,\cdots,Z$}{

$\zeta_{i}^{(z+1)}=v+T(\zeta^{(z)})_{i}$ for $i\in[d]$.

}

\textbf{Output:} $p(t)=v+\int_{0}^{t}\sum_{j=1}^{d}F(\zeta_{j}^{(Z+1)},c_{j})\phi_{j}(s)ds$.

\end{algorithm2e}

In the following lemma, we show that $T$ is a contraction mapping
if $F$ is smooth enough and hence this algorithm converges linearly.
We will use the following norm: $\norm x_{\infty;p}=\max_{i\in[d]}\left(\sum_{k\in[n]}\left|x_{(i,k)}\right|^{p}\right)^{1/p}$.
\begin{lem}
\label{lem:smothness_T}Given $x,y\in\R^{d\times n}$, suppose that
$\norm{F(x_{i},t)-F(y_{i},t)}_{p}\leq L\norm{x_{i}-y_{i}}_{p}$ for
all $0\leq t\leq\ell$ and all $i\in[d]$. For $c_{k}=\frac{\ell}{2}+\frac{\ell}{2}\cos(\frac{2k-1}{2d}\pi)$
and $T$ defined Lemma \ref{lem:two_ODE_version}, we have that
\[
\norm{T(x)-T(y)}_{\infty;p}\leq1000\ell L\norm{x-y}_{\infty;p}.
\]
Also, for any $\alpha\in\R^{d\times n}$, we have that
\[
\int_{0}^{t}\sum_{j=1}^{d}\alpha_{j}\phi_{j}(s)ds\leq1000\ell\max_{i}\norm{\alpha_{i}}_{p}
\]
\end{lem}
\begin{proof}
Using the definition of $T$ and $\norm{\cdot}_{\infty;p}$, we have
\begin{eqnarray*}
\norm{T(x)-T(y)}_{\infty;p} & = & \max_{i}\norm{\int_{0}^{c_{i}}\sum_{j=1}^{d}F(x_{j},c_{j})\phi_{j}(s)ds-\int_{0}^{c_{i}}\sum_{j=1}^{d}F(y_{j},c_{j})\phi_{j}(s)ds}_{p}\\
 & \leq & \left(\max_{i}\sum_{j=1}^{d}\left|\int_{0}^{c_{i}}\phi_{j}(s)ds\right|\right)\left(\max_{i}\norm{F(x_{i},c_{i})-F(y_{i},c_{i})}_{p}\right)\\
 & \leq & L\norm{x-y}_{\infty;p}\left(\max_{i}\sum_{j=1}^{d}\left|\int_{0}^{c_{i}}\phi_{j}(s)ds\right|\right).
\end{eqnarray*}

To bound $\max_{i}\sum_{j=1}^{d}\left|\int_{0}^{c_{i}}\phi_{j}(s)ds\right|$,
it is easier to work with the function $\psi_{j}(x)\defeq\phi_{j}(\frac{2}{\ell}x-1)$.
Note that $\psi_{k}$ is the Lagrange basis polynomials on the nodes
$\{\cos(\frac{2k-1}{2d}\pi)\}_{k=1}^{d}$ and hence we have 
\[
\psi_{k}(x)=\frac{(-1)^{k-1}\sqrt{1-x_{k}^{2}}\cos\left(d\cos^{-1}x\right)}{d(x-x_{k})}
\]
where $x_{k}=\cos(\frac{2k-1}{2d}\pi)$. Lemma \ref{lem:integrate_interpolation}
at the end of this section shows that $\left|\int_{-1}^{t}\psi_{k}(x)dx\right|\leq\frac{2000}{d}$
for all $t$. Therefore, we have that $\max_{i}\sum_{j=1}^{d}\left|\int_{0}^{c_{i}}\phi_{j}(s)ds\right|\leq1000\ell$.
This gives the first inequality. The second inequality is similar.
\end{proof}
In each iteration of the collation method, we need to compute $\int_{0}^{c_{i}}\sum\alpha_{j}\phi_{j}(s)ds$
for some $\alpha_{j}$. The following theorem shows that this can
be done in $O(d\log(d/\varepsilon))$ time using multipole method.
\begin{thm}[{\cite[Sec 5]{dutt1996fast}}]
\label{thm:integrate_basis}Let $\phi_{i}(s)$ be the Lagrange basis
polynomials on the Chebyshev nodes on $[0,\ell]$, namely, $\phi_{i}(s)=\prod_{j\neq i}\frac{s-c_{j}}{c_{i}-c_{j}}$
with $c_{i}=\frac{\ell}{2}+\frac{\ell}{2}\cos(\frac{2i-1}{2d}\pi)$.
Given a polynomial $p(s)=\sum_{j}\alpha_{j}\phi_{j}(s)$ and a point
set $\{x_{1},x_{2},\cdots,x_{d}\}$, one can compute $t_{i}$ such
that
\[
\left|t_{i}-\int_{0}^{x_{i}}p(s)ds\right|\leq\varepsilon\ell\sqrt{\sum_{j\neq k}(\alpha_{j}-\alpha_{k})^{2}}\text{ for }i\in[d]
\]
in time $O(d\log(d/\varepsilon))$.
\end{thm}
Now we have everything to state our main result in this subsection.
\begin{thm}
\label{thm:Solving_ODE}Let $u(t)\in\Rn$ be the solution of the ODE
(\ref{eq:ODE}). Suppose we are given $\varepsilon,\ell>0$ and $1\leq p\leq\infty$
such that

\begin{enumerate}
\item There is a degree $d$ polynomial $q$ from $\R$ to $\Rn$ such that
$q(0)=v$ and $\norm{\frac{d}{dt}u(t)-\frac{d}{dt}q(t)}_{p}\leq\frac{\varepsilon}{\ell}$
for all $0\leq t\leq\ell$.
\item We have that $\norm{F(x,t)-F(y,t)}_{p}\leq\frac{1}{2000\ell}\norm{x-y}_{p}$
for all $\norm{x-v}_{p}\leq K$ and $\norm{y-v}_{p}\leq K$ with $K=4000\ell\max_{t\in[0,\ell]}\norm{F(v,t)}_{p}.$ 
\end{enumerate}
Then, Algorithm $\ensuremath{\texttt{CollocationMethod}}$ outputs
a degree $d$ polynomial $p(t)$ such that $\max_{0\leq t\leq\ell}\norm{u(t)-p(t)}_{p}\leq4003\varepsilon$
in $O(nd\log^{2}(dK/\varepsilon))$ time and $O(d\log(K/\varepsilon))$
evaluations of $F$.
\end{thm}
\begin{proof}
First, we estimate the initial error. Let $\zeta^{(\infty)}$ be the
solution to (\ref{eq:ODE_fix_point}), $\zeta_{i}^{(0)}=v+T(\overline{v})_{i}$
be the initial vector and $\overline{v}=(v,v,\cdots,v)\in\R^{d\times n}$.
Then, we have that
\begin{eqnarray*}
\norm{\zeta^{(\infty)}-\zeta^{(0)}}_{\infty;p} & = & \max_{i}\norm{(v+T(\zeta^{(\infty)})_{i})-(v+T(\overline{v})_{i})}_{p}\\
 & = & \norm{T(\zeta^{(\infty)})-T(\overline{v})}_{\infty;p}\\
 & \leq & \frac{1}{2}\norm{\zeta^{(\infty)}-\overline{v}}_{\infty;p}\\
 & \leq & \frac{1}{2}\left(\norm{\zeta^{(\infty)}-\zeta^{(0)}}_{\infty;p}+\norm{\zeta^{(0)}-\overline{v}}_{\infty;p}\right).
\end{eqnarray*}
Therefore, we have that
\begin{eqnarray*}
\norm{\zeta^{(\infty)}-\zeta^{(0)}}_{\infty;p} & \leq & \norm{\zeta^{(0)}-\overline{v}}_{\infty;p}\\
 & = & \norm{T((v,v,\cdots,v))}_{\infty;p}\\
 & \leq & 1000\ell\max_{i}\norm{F(v,c_{i})}_{p}\\
 & \leq & \frac{K}{4}.
\end{eqnarray*}
Hence, we have that $\norm{\zeta^{(\infty)}-\overline{v}}_{\infty;p}\leq\frac{K}{2}$. 

Using the assumption on $F$, Lemma \ref{lem:smothness_T} shows that
$\norm{T(x)-T(y)}_{\infty;p}\leq\frac{1}{2}\norm{x-y}_{\infty;p}$
and hence
\begin{eqnarray*}
\norm{\zeta^{(\infty)}-\zeta^{(k)}}_{\infty;p} & \leq & \frac{K}{2^{k+1}}.
\end{eqnarray*}
Thus, it takes $\log_{2}(K/\varepsilon)$ iteration to get a point
$\zeta^{(k)}$ with 
\begin{equation}
\norm{\zeta^{(\infty)}-\zeta^{(k)}}_{\infty;p}\leq\varepsilon.\label{eq:accuracy_zeta}
\end{equation}
Also, this shows that $\norm{\zeta^{(\infty)}-\overline{v}}_{\infty;p}\leq K$.
Hence, we only requires the assumption (2) to be satisfied in this
region.

Now, we show that $\zeta^{(k)}$ is close to the solution by using
the existence of $q$. By the assumption on $q$, we have that $\norm{\frac{d}{dt}u(t)-\frac{d}{dt}q(t)}_{p}\leq\frac{\varepsilon}{\ell}$
and hence $\norm{u(t)-q(t)}_{p}\leq\varepsilon$. Using the smoothness
of $F$, we have that $\norm{F(u(t),t)-F(q(t),t)}_{p}\leq\frac{\varepsilon}{\ell}$
for all $0\leq t\leq\ell$. Therefore, we have that
\begin{eqnarray*}
\norm{\frac{d}{dt}q(t)-F(q(t),t)}_{p} & \leq & \norm{\frac{d}{dt}q(t)-\frac{d}{dt}u(t)}_{p}+\norm{F(q(t),t)-F(u(t),t)}_{p}\\
 & \leq & 2\frac{\varepsilon}{\ell}
\end{eqnarray*}
for all $0\leq t\leq1$. Therefore, we have
\begin{eqnarray*}
q(t) & = & v+\int_{0}^{t}\sum_{j=1}^{d}\frac{d}{dt}q(c_{j})\phi_{j}(s)ds\\
 & = & v+\int_{0}^{t}\sum_{j=1}^{d}\left(F(q(c_{j}),c_{j})+\delta_{j}\right)\phi_{j}(s)ds
\end{eqnarray*}
where $\norm{\delta_{i}}_{p}\leq2\frac{\varepsilon}{\ell}$ for all
$i\in[d]$. By Lemma \ref{lem:smothness_T}, we have that
\[
\norm{q(t)-v-\int_{0}^{t}\sum_{j=1}^{d}F(q(c_{j}),c_{j})\phi_{j}(s)ds}_{p}\leq2000\varepsilon.
\]
Now, we compare $\overline{q}_{j}\defeq q(c_{j})$ with the approximate
solution constructed by the fix point algorithm
\[
q^{(k)}(t)=v+\int_{0}^{t}\sum_{j=1}^{d}F(\zeta_{j}^{(k)},c_{j})\phi_{j}(s)ds.
\]
For $0\leq t\leq\ell$, we have
\begin{eqnarray}
\norm{q(t)-q^{(k)}(t)}_{p} & \leq & \norm{T(\overline{q})-T(\zeta^{(k)})}_{\infty;p}+2000\varepsilon\nonumber \\
 & \leq & \frac{1}{2}\norm{\overline{q}-\zeta^{(k)}}_{\infty;p}+2000\varepsilon.\label{eq:P_difference}
\end{eqnarray}
Setting $t=c_{i}$, we have
\begin{eqnarray*}
\norm{\overline{q}-\zeta^{(k+1)}}_{\infty;p} & = & \max_{i}\norm{\overline{q}(c_{i})-q^{(k)}(c_{i})}_{p}\\
 & \leq & \frac{1}{2}\norm{q-\zeta^{(k)}}_{\infty;p}+2000\varepsilon.
\end{eqnarray*}
Since $\norm{\zeta^{(k+1)}-\zeta^{(k)}}_{\infty;p}\leq\norm{\zeta^{(k+1)}-\zeta^{(\infty)}}_{\infty;p}+\norm{\zeta^{(k)}-\zeta^{(\infty)}}_{\infty;p}\leq2\varepsilon$,
we have
\begin{eqnarray*}
\norm{\overline{q}-\zeta^{(k)}}_{\infty;p} & \leq & \frac{1}{2}\norm{\overline{q}-\zeta^{(k)}}_{\infty;p}+2\varepsilon+2000\varepsilon
\end{eqnarray*}
Therefore, we have
\[
\norm{\overline{q}-\zeta^{(k)}}_{\infty;p}\leq4004\varepsilon.
\]
Putting it into (\ref{eq:P_difference}), we have
\[
\norm{q(t)-q^{(k)}(t)}_{p}\leq4002\varepsilon
\]
for all $0\leq t\leq\ell$. Using that $\norm{u(t)-q(t)}_{p}\leq\varepsilon$,
we have
\[
\norm{u(t)-q^{(k)}(t)}_{p}\leq4003\varepsilon.
\]
Hence, it proves the guarantee.

Each iteration involves computing $v_{k}+\int_{0}^{c_{i}}\sum_{j=1}^{d}F(\zeta_{j}^{(z)},c_{j})_{k}\phi_{j}(s)ds$
for all $i\in[d],k\in[n]$. Note that $\sum_{j=1}^{d}F(\zeta_{j}^{(z)},c_{j})_{k}\phi_{j}(s)$
is a polynomial expressed by Lagrange polynomials. Theorem \ref{thm:integrate_basis}
shows they can be computed in $O(d\log(dK/\varepsilon)))$ with $\frac{\varepsilon}{Kd^{O(1)}}$
accuracy. Since there are $n$ coordinates, it takes $O(nd\log(dK/\varepsilon)))$
time plus $d$ evaluation per iteration.
\end{proof}
The theorem above essentially says that if the solution is well approximated
by a polynomial and if the $F$ has small enough Lipschitz constant,
then we can reconstruct the solution efficiently. Note that this method
is not useful for stochastic differential equation because Taylor
expansion of the solution involves the high moments of probability
distributions which is very expensive to store.

The assumption on the Lipschitz constant of $F$ holds for our application.
For the rest of this subsection, we show that this assumption is not
necessary by taking multiple steps. This is mainly to make the result
easier to use for other applications and not needed for this paper.
First, we prove that the collocation method is stable under small
perturbation of the initial solution. 
\begin{lem}
\label{lem:collocation_is_stable}Let $p(t)$ and $\tilde{p}(t)$
be the outputs of $\ensuremath{\texttt{CollocationMethod}}$ for the
initial value $v$ and the initial value $\tilde{v}$. We make the
same assumption as Theorem \ref{thm:Solving_ODE} for the initial
condition $v$ (albeit a small change in the constants). Suppose that
$\norm{\tilde{v}-v}_{p}\leq c\ell\max_{t\in[0,\ell]}\norm{F(v,t)}_{p}$
with small enough constant $c$, we have that 
\[
\max_{0\leq t\leq\ell}\norm{p(t)-\tilde{p}(t)}_{p}=O(\norm{\tilde{v}-v}_{p}).
\]
\end{lem}
\begin{proof}
With a proof similar to Theorem \ref{thm:Solving_ODE}, we know that
$T$ is $\frac{1}{2}$ Lipschitz. Here, we use that $F$ is Lipschitz
in a neighbour of $v$ and a neighbour of $\tilde{v}$. Since $\tilde{v}$
is close to $v$, we only need to make the assumption on a neighbour
of $v$.

Let $\zeta^{(k)}$ , $\tilde{\zeta}^{(k)}$ be the corresponding intermediate
variables in the algorithm. Since $\zeta_{i}^{(0)}=v+T(v,v,\cdots)_{i}$
and $\tilde{\zeta}_{i}^{(0)}=\tilde{v}+T(\tilde{v},\tilde{v},\cdots)_{i}$,
we have that
\begin{align*}
\norm{\tilde{\zeta}^{(0)}-\zeta^{(0)}}_{\infty;p} & \leq\norm{\tilde{v}-v}_{p}+\norm{T(\tilde{v},\tilde{v},\cdots)-T(v,v,\cdots)}_{\infty;p}\\
 & \leq\norm{\tilde{v}-v}_{p}+\frac{1}{2}\norm{\tilde{v}-v}_{p}=\frac{3}{2}\norm{\tilde{v}-v}_{p}\leq2\norm{\tilde{v}-v}_{p}.
\end{align*}
Since $\zeta_{i}^{(k+1)}=v+T(\zeta^{(k)})_{i}$ and $\tilde{\zeta}_{i}^{(k+1)}=\tilde{v}+T(\tilde{\zeta}^{(k)})_{i}$,
by induction, we have that
\begin{align*}
\norm{\tilde{\zeta}^{(k+1)}-\zeta^{(k+1)}}_{\infty;p} & \leq\norm{\tilde{v}-v}_{p}+\frac{1}{2}\norm{\tilde{\zeta}^{(k)}-\zeta^{(k)}}_{\infty;p}\\
 & \leq\norm{\tilde{v}-v}_{p}+\frac{1}{2}\cdot2\norm{\tilde{v}-v}_{p}\leq2\norm{\tilde{v}-v}_{p}.
\end{align*}

Now, we note that $p(t)=v+\int_{0}^{t}\sum_{j=1}^{d}F(\zeta_{j}^{(Z+1)},c_{j})\phi_{j}(s)ds$.
By our assumption on $F$, we have that 
\begin{align*}
\norm{F(\zeta_{j}^{(Z+1)},c_{j})-F(\tilde{\zeta}_{j}^{(Z+1)},c_{j})}_{p} & =O(\frac{1}{\ell})\norm{\zeta_{j}^{(Z+1)}-\tilde{\zeta}_{j}^{(Z+1)}}_{p}\\
 & =O(\frac{1}{\ell})\norm{\tilde{v}-v}_{p}.
\end{align*}
Apply second part of Lemma \ref{lem:smothness_T}, we have that
\[
\norm{p(t)-\tilde{p}(t)}_{p}=O(\norm{\tilde{v}-v}_{p})
\]
for any $0\leq t\leq\ell$.
\end{proof}
Now, we give the main theorem of this subsection.
\begin{thm}
\label{thm:Solving_ODE_2}Let $u(t)\in\Rn$ be the solution of the
ODE (\ref{eq:ODE}). Suppose we are given $\varepsilon>0$ and $1\leq p\leq\infty$
such that

\begin{enumerate}
\item There is a degree $d$ polynomial $q$ from $\R$ to $\Rn$ such that
$q(0)=v$ and $\norm{\frac{d}{dt}u(t)-\frac{d}{dt}q(t)}_{p}\leq\varepsilon$
for all $0\leq t\leq1$.
\item For some $L\geq1$, we have that $\norm{F(x,t)-F(y,t)}_{p}\leq L\norm{x-y}_{p}$
for all $x,y$ and $0\leq t\leq1$.
\end{enumerate}
\end{thm}
Then, we can compute $u$ such that $\norm{u-u(1)}_{p}=O(\varepsilon)$
in $O(ndL^{3}\log^{2}(dK/\varepsilon))$ time and $O(dL^{2}\log(K/\varepsilon))$
evaluations of $F$ where $K=\max_{x,0\leq t\leq1}\norm{F(x,t)}_{p}$.
\begin{proof}
Let $\ell=\frac{1}{2000L}$. By Theorem \ref{thm:Solving_ODE}, we
can compute $u(\ell)$ in $O(nd\log^{2}(dK/\varepsilon))$ time and
$O(d\log(K/\varepsilon))$ evaluations of $F$. Now, we can apply
\ref{thm:Solving_ODE} again with our approximate of $u(\ell)$ to
get $u(2\ell)$. Repeating this process, we can compute $u(1)$.

Lemma \ref{lem:collocation_is_stable} shows that if the initial value
of $u(0)$ has $\delta$ error, the collocation method would output
$u(\ell)$ with $O(\delta)+\varepsilon$ error. Therefore, to compute
$u(1)$ with $\varepsilon$ error, we need to iteratively compute
$y(i\ell)$ with error $\varepsilon/\Theta(1)^{1/\ell-i}$. Hence,
the total running time is
\[
\sum_{i=1}^{1/\ell}O(nd\log^{2}(\Theta(1)^{1/\ell-i}dK/\varepsilon))=O\left(nd\left(\frac{1}{\ell}\right)^{3}\log^{2}(dK/\varepsilon)\right)
\]
with 
\[
\sum_{i=1}^{1/\ell}O(d\log(\Theta(1)^{1/\ell-i}dK/\varepsilon))=O\left(d\left(\frac{1}{\ell}\right)^{2}\log^{2}(dK/\varepsilon)\right)
\]
 evaluations of $F$.
\end{proof}
\begin{rem}
The main bottleneck of this algorithm is that the error blows up exponentially
when we iteratively apply the collocation method. We believe this
can be alleviated by using the collocation method in one shot but
with a more smart initialization. Since we do not face this problem
for our main application, we left this as a future investigation.
\end{rem}

\subsection{Second Order ODE}

Now, we consider the following second order ODE
\begin{eqnarray}
\frac{d^{2}}{dt^{2}}u(t) & = & F(\frac{d}{dt}u(t),u(t),t),\text{ for }0\leq t\leq\ell\label{eq:ODE2}\\
\frac{d}{dt}u(0) & = & w,\nonumber \\
u(0) & = & v\nonumber 
\end{eqnarray}
where $F:\R^{2n+1}\rightarrow\Rn$ and $u(t)\in\Rn$. Using a standard
reduction from second order ODE to first order ODE, we show how to
apply our first order ODE method to second order ODE. 
\begin{thm}
\label{thm:Solving_ODE_2nd}Let $x(t)\in\Rn$ be the solution of the
ODE (\ref{eq:ODE2}). Given some $\varepsilon,\ell>0$ and $1\leq p\leq\infty$,
let $\alpha=4000\ell$ and suppose that 

\begin{enumerate}
\item There is a degree $d$ polynomial $q$ from $\R$ to $\Rn$ such that
$q(0)=v$, $q'(0)=w$, $\norm{\frac{d}{dt^{2}}x(t)-\frac{d}{dt^{2}}q(t)}_{p}\leq\frac{\varepsilon}{\ell^{2}}$
for all $0\leq t\leq\ell$.
\item We have that $\norm{F(x,\gamma,t)-F(y,\eta,t)}_{p}\leq\frac{1}{\alpha}\norm{x-y}_{p}+\frac{1}{\alpha^{2}}\norm{\gamma-\eta}_{p}$
for all $\norm{x-w}_{p}\leq K$, $\norm{y-w}_{p}\leq K$, $\norm{\gamma-v}_{p}\leq\alpha K$,
$\norm{\eta-v}_{p}\leq\alpha K$ where $K=\alpha\max_{t\in[0,\ell]}\norm{F(w,v,t)}_{p}+\norm w_{p}$.
\end{enumerate}
Then, in $O(nd\log^{2}(dK/\varepsilon))$ time plus $O(d\log(K/\varepsilon))$
evaluations of $F$, we can find $p(t)$ such that 
\[
\max_{0\leq t\leq\ell}\norm{u(t)-p(t)}_{p}=O(\varepsilon)\quad\text{and}\quad\max_{0\leq t\leq\ell}\norm{u'(t)-p'(t)}_{p}=O(\varepsilon/\ell).
\]
\end{thm}
\begin{proof}
Let $\alpha=4000\ell$. Let $x(t)=(\alpha u'(\alpha t),u(\alpha t))\in\R^{2n}$.
Note that $x(t)$ satisfies the following ODE
\begin{eqnarray}
\frac{d}{dt}x(t) & = & \overline{F}(x(t),t)\text{ for }0\leq t\leq\ell\label{eq:new_2nd_ODE}\\
x(0) & = & (\alpha w,v).\nonumber 
\end{eqnarray}
where $\overline{F}(x(t),t)=(\alpha^{2}F(\alpha^{-1}x_{(1)}(t),x_{(2)}(t),\alpha t),x_{(1)}(t))$,
$x_{(1)}(t)$ is the first $n$ variables of $x(t)$ and $x_{(2)}(t)$
is the last $n$ variables of $x(t)$. Next, we verify the conditions
of Theorem \ref{thm:Solving_ODE} for this ODE. Let $\overline{\ell}=\frac{1}{6000}$
and $\overline{K}=4000\overline{\ell}\max_{t\in[0,\ell]}\norm{\overline{F}(x(0),t)}_{p}$.
Note that
\[
\overline{K}\leq\alpha^{2}\max_{t\in[0,\ell]}\norm{F(w,v,t)}_{p}+\alpha\norm w_{p}.
\]
For any $y,z$ such that $\norm{y-x(0)}_{p}\leq\overline{K}$ and
$\norm{z-x(0)}_{p}\leq\overline{K}$, we apply the assumption on $F$
and get that 
\begin{eqnarray*}
 &  & \norm{\overline{F}(y,t)-\overline{F}(z,t)}_{p}\\
 & \leq & \alpha^{2}\norm{F(\alpha^{-1}y_{(1)},y_{(2)},\alpha t)-F(\alpha^{-1}z_{(1)},z_{(2)},\alpha t)}_{p}+\norm{y_{(1)}-z_{(1)}}_{p}\\
 & \leq & \alpha^{2}\left(\frac{1}{\alpha}\norm{\alpha^{-1}y_{(1)}-\alpha^{-1}z_{(1)}}_{p}+\frac{1}{\alpha^{2}}\norm{y_{(2)}-z_{(2)}}_{p}\right)+\norm{y_{(1)}-z_{(1)}}_{p}\\
 & \leq & 3\norm{y-z}_{p}.
\end{eqnarray*}

Also, by our assumption on $x$, we have a polynomial $\overline{q}=\left(\alpha q'(\alpha t),q(\alpha t)\right)$
such that
\[
\norm{\frac{d}{dt}x(t)-\frac{d}{dt}\overline{q}(t)}_{p}=O\left(\frac{\varepsilon}{\overline{\ell}}\right)
\]
where $x$ is the solution of the ODE (\ref{eq:new_2nd_ODE}). Therefore,
Theorem \ref{thm:Solving_ODE} shows that we can compute $x$ with
$O(\varepsilon)$ error in $O(nd\log^{2}(dK/\varepsilon))$ time plus
$O(d\log(K/\varepsilon))$ evaluations of $F$. Since $x=(\alpha u',u)$,
this gives us an approximate of $u$ with error $O(\varepsilon)$
and $u'$ with error $O(\varepsilon/\ell)$.
\end{proof}

\subsection{Example: Discretization of Physarum dynamics\label{subsec:physarum}}

In this section, we use the Physarum dynamics as an example to showcase
the usefulness of the collocation method. We consider a linear program
of the form
\begin{equation}
\min_{Ax=b,x\geq0}c^{T}x\label{eq:LP}
\end{equation}
where $A\in\mathbb{Z}^{m\times n}$, $c\in\mathbb{Z}_{>0}^{n}$ and
$b\in\mathbb{Z}^{m}$. Inspired by Physarum polycephalum (a slime
mold), Straszak and Vishnoi \cite{straszak2015natural} introduce
the following dynamics for solving the linear program:
\begin{equation}
\frac{dx}{dt}=WA^{T}(AWA)^{-1}b-x\label{eq:LP_ODE}
\end{equation}
where $W$ is a diagonal matrix with the diagonal $W_{ii}=x_{i}/c_{i}$.
In this subsection, we follow their notations/assumptions:
\begin{enumerate}
\item Assume that $A$ is full rank and this linear program has a feasible
solution. 
\item Let $\text{OPT}$ be the optimal value of the linear program (\ref{eq:LP}).
\item Let $D$ be the maximum sub-determinant of $A$, i.e. $D=\max_{A'\text{ is a square submatrix of }A}\left|\det(A')\right|$.
For graph problems, $A$ is usually an unimodular matrix and hence
$D=1$. 
\item Assume that we have an initial point $x(0)$ such that $Ax(0)=b$
and a parameter $M$ such that
\[
M^{-1}\leq x(0)\leq M\text{ and }c^{T}s\leq M\cdot\text{OPT}.
\]
\end{enumerate}
They showed that the continuous dynamics converges linearly to the
solution:
\begin{lem}[{Convergence of Physarum dynamics, \cite[Thm 6.3]{straszak2015natural}}]
Consider $x(t)$ be the solution of \ref{eq:LP_ODE} with initial
point $Ax(0)=b$. Then, we have that
\[
\text{OPT}\leq c^{T}x(t)\leq\text{OPT}+(n+M)^{2}e^{8D^{2}\norm c_{1}\norm b_{1}-D^{-3}t}.
\]
\end{lem}
Furthermore, they analyzed the Euler method for this dynamics and
obtained the following result:
\begin{lem}[{Euler method for Physarum dynamics, \cite[Thm 7.1]{straszak2015natural}}]
Consider the discretization of the Physarum dynamics
\begin{align*}
x^{(k+1)} & =(1-h)x^{(k)}+hW^{(k)}A^{T}(AW^{(k)}A)^{-1}b,\\
x^{(0)} & =x(0)
\end{align*}
with the diagonal matrix $W_{ii}^{(k)}=x_{i}^{(k)}/c_{i}$. Then,
for any $\varepsilon>0$ and $h=\frac{1}{6}\varepsilon\norm c_{1}^{-2}D^{-2}$,
we have that
\[
\text{OPT}\leq c^{T}x^{(k)}\leq(1+\varepsilon)\text{OPT}
\]
where $k=O\left(\frac{\ln M}{\varepsilon^{2}h^{2}}\right)=O\left(\frac{\norm c_{1}^{4}D^{4}\ln M}{\epsilon^{4}}\right)$.
\end{lem}
Note that the continuous process converges linearly while the discrete
process converges sublinearly. In this section, we show how to use
collocation method to get a discrete process that converges ``linearly''.
\begin{lem}[Collocation method for Physarum dynamics]
For any $T\geq0$ and $1>\varepsilon>0$, we can compute $y$ such
that 
\[
(1-\varepsilon)x_{i}(T)\leq y_{i}\leq(1+\varepsilon)x_{i}(T)
\]
for all $i$ in time $O\left(n^{\omega+3}D^{6}\norm c_{1}^{3}T^{3}\log^{2}(nD\norm c_{1}/\varepsilon)\right).$
\end{lem}
\begin{proof}
Let $y=\ln x$. The Physarum dynamics can be written as
\begin{align*}
\frac{dy}{dt} & =\diag\left(\frac{1}{c}\right)A^{T}\left(A\diag\left(\frac{e^{y}}{c}\right)A^{T}\right)^{-1}b-1\\
 & \defeq F(y).
\end{align*}
To use Theorem \ref{thm:Solving_ODE}, we need to bound the Lipschitz
constant of $F$ and show that $y(t)$ can be approximated by polynomial.

The Lipschitz constant of $F$: Given any path $\tilde{y}(t)$ in
$\Rn$. Let $W(t)=\diag\left(\frac{e^{\tilde{y}(t)}}{c}\right)$,
$C=\diag\left(c\right)$ and $P(t)=A^{T}\left(AW(t)A^{T}\right)^{-1}AW(t)$.
Then, we have that 
\begin{align*}
\frac{d}{dt}F(\tilde{y}(t)) & =\frac{d}{dt}C^{-1}A^{T}\left(AW(t)A^{T}\right)^{-1}b\\
 & =-C^{-1}A^{T}\left(AW(t)A^{T}\right)^{-1}AW(t)\diag\left(\frac{d\tilde{y}}{dt}\right)A^{T}\left(AW(t)A^{T}\right)^{-1}b\\
 & =-C^{-1}P(t)\diag\left(\frac{d\tilde{y}}{dt}\right)P(t)c.
\end{align*}
where we used $Ax=b$ at the last line. By \cite[Lem 5.2]{straszak2015natural},
we have that
\begin{equation}
\norm{P(t)}_{1\rightarrow\infty}\leq D.\label{eq:slime_P}
\end{equation}
Hence, we have that
\begin{align*}
\norm{P(t)\diag\left(\frac{d\tilde{y}}{dt}\right)P(t)c}_{\infty}\leq & D\norm{\diag\left(\frac{d\tilde{y}}{dt}\right)P(t)c}_{1}\\
\leq & nD\norm{\frac{d\tilde{y}}{dt}}_{\infty}\norm{P(t)c}_{\infty}\leq nD^{2}\norm{\frac{d\tilde{y}}{dt}}_{\infty}\norm c_{1}.
\end{align*}
Since $c$ is integral, we have that
\begin{equation}
\norm{\frac{d}{dt}F(\tilde{y}(t))}_{\infty}\leq nD^{2}\norm{\frac{d\tilde{y}}{dt}}_{\infty}\norm c_{1}.\label{eq:slim_Lip}
\end{equation}
Hence, $F$ has Lipschitz constant $nD^{2}\norm c_{1}$ in $L^{\infty}$
norm. 

Analyticity of $y(t)$: Note that
\begin{equation}
\frac{dy}{dt}=C^{-1}P(t)c-1.\label{eq:slime_D}
\end{equation}
By (\ref{eq:slime_P}), we have that $\norm{\frac{dy}{dt}}_{\infty}\leq1+D\norm c_{1}\leq2D\norm c_{1}$.
Note that
\[
\frac{d}{dt}P(t)=-P(t)\diag(\frac{dy}{dt})P(t)+P(t)\diag(\frac{dy}{dt}).
\]
Hence, we have that
\[
\frac{d^{2}y}{dt^{2}}=-C^{-1}P(t)\diag(\frac{dy}{dt})P(t)c+C^{-1}P(t)\diag(\frac{dy}{dt})c
\]
and 
\begin{align*}
\norm{\frac{d^{2}y}{dt^{2}}}_{\infty} & \leq nD\norm{\frac{dy}{dt}}_{\infty}D\norm c_{1}+D\norm{\frac{dy}{dt}}_{\infty}\norm c\\
 & \leq2nD^{2}\norm c_{1}\norm{\frac{dy}{dt}}_{\infty}\leq4nD^{3}\norm c_{1}^{2}.
\end{align*}
By induction, one can show that
\[
\norm{\frac{d^{k}y}{dt^{k}}}_{\infty}=O(1)^{k}k!D^{2k-1}n^{k-1}\norm c_{1}^{k}.
\]
For $d=\log(1/\varepsilon)$ and $\ell=\frac{1}{\Omega(nD^{2}\norm c_{1})}$,
we have that
\[
\norm{y(t)-\sum_{k=0}^{d-1}\frac{1}{k!}y^{(k)}(0)t^{k}}_{\infty}\leq O(1)^{d}D^{2d-1}n^{d-1}\norm c_{1}^{d}\ell^{d}\leq\varepsilon
\]
for any $0\leq t\leq\ell$. Similarly, we have that
\[
\norm{\frac{d}{dt}y(t)-\frac{d}{dt}\sum_{k=0}^{d-1}\frac{1}{k!}y^{(k)}(0)t^{k}}_{\infty}\leq\varepsilon
\]
for any $0\leq t\leq\ell$. 

Since $F$ has Lipschitz constant $nD^{2}\norm c_{1}$, we can apply
Theorem \ref{thm:Solving_ODE} with $\ell=O\left(\frac{1}{nD^{2}\norm c_{1}\log(1/\varepsilon)}\right)$,
$d=\log(1/\varepsilon)$, $p=\infty$. Hence, we can compute $y(\ell)$
with $\varepsilon$ error in $O(n\log^{2}(nD\norm c_{1}/\varepsilon))$
time with $O(\log(nD\norm c_{1}/\varepsilon))$ evaluations of $F$.
Since $F$ can be computed in matrix multiplication time, we can compute
$y(\ell)$ with $\varepsilon$ error in $O(n^{\omega}\log^{2}(nD\norm c_{1}/\varepsilon))$
time.

Next, we note that if the initial value of $y(0)$ has $\delta$ error,
the collocation method would output $y(\ell)$ with $O(\delta)+\varepsilon$
error. Therefore, to compute $y(T)$ with $\varepsilon$ error, we
need to iteratively compute $y(i\ell)$ with error $\varepsilon/\Theta(1)^{T/\ell-i}$.
Hence, the total running time is
\[
\sum_{i=1}^{T/\ell}O(n^{\omega}\log^{2}(\Theta(1)^{T/\ell-i}nD\norm c_{1}/\varepsilon))=O\left(n^{\omega}\left(\frac{T}{\ell}\right)^{3}\log^{2}(nD\norm c_{1}/\varepsilon)\right).
\]
\end{proof}
\begin{rem}
We use this example merely to showcase that the collocation method
is useful for getting a polynomial time algorithm for solving ordinary
differential equations. The $k^{th}$ order derivatives of the path
$x(t)$ can be computed efficiently in $(kn)^{O(1)}$ time and hence
one can simply use Taylor series to approximate Physarum dynamics.
Alternatively, since the path $x(t)$ is the solution of certain convex
optimization problem \cite[Thm 4.2]{straszak2015natural}, we can
be computed it directly in $\tilde{O}(n^{3})$ time \cite{lee2015faster}.
\end{rem}

\subsection{Derivative Estimations\label{sec:est_var}}

For any smooth one dimension function $f$, we know by Taylor's theorem
that 
\[
f(x)=\sum_{k=0}^{N}\frac{f^{(k)}(a)}{k!}(x-a)^{k}+\frac{1}{N!}\int_{a}^{x}(x-t)^{N}f^{(N+1)}(t)dt.
\]
This formula provides a polynomial estimate of $f$ around $a$. To
analyze the accuracy of this estimate, we need to bound $\left|f^{(N+1)}(t)\right|$.
In one dimension, we could simply give explicit formulas for the derivatives
of $f$ and use it to estimate the remainder term. However, for functions
in high dimension, it is usually too tedious. Here we describe some
techniques for bounding the derivatives of higher dimensional functions. 

The derivatives of one-variable functions can be bounded via Cauchy's
estimate (Theorem \ref{thm:cauchy_estimate}). In Section \ref{subsec:est_n_var},
we give calculus rules that reduces the problem of estimating derivatives
of high-dimensional functions to derivatives of one-dimensional functions.
In Section \ref{subsec:est_ode}, we show how to reduce bounding the
derivative for an arbitrary ODE to an ODE in one dimension. 

\subsection{Explicit Function\label{subsec:est_n_var}}

In this subsection, we show how to bound the derivatives of a complicated
explicit function using the following object, generating upper bound.
We reduce estimates of the derivatives of functions in high dimension
to one variable rational functions. Since rational functions are holomorphic,
one can apply Cauchy's estimates (Theorem (\ref{thm:cauchy_estimate}))
to bound their derivatives. 
\begin{defn}
Given a function $F$. We call that $F\leq_{x}f$ for some one variable
function $f:\R\rightarrow\R$ if
\begin{equation}
\norm{D^{(k)}F(x)[\Delta_{1},\Delta_{2},\cdots,\Delta_{k}]}\leq f^{(k)}(0)\prod_{i=1}^{k}\norm{\Delta_{i}}^{k}\label{eq:less_sim_est}
\end{equation}
for any $k\geq0$ and any $\Delta_{i}$.
\end{defn}
\begin{rem*}
In general, $F\leq_{x}f$ and $f(t)\leq g(t)$ point-wise does NOT
imply $F\leq_{x}g$. However, $F\leq_{x}f$ and $f\leq_{0}g$ does
imply $F\leq_{x}g$.
\end{rem*}
\begin{rem*}
The bounds we give in this subsection only assume that $\norm{\cdot}$
is a norm and it satisfies $\norm{ab}\leq\norm a\norm b$. In the
later sections, we always use $\norm{\cdot}_{2}$ for both matrix
and scalar.
\end{rem*}
This concept is useful for us to reduce bounding derivatives of a
high dimension function to bounding derivatives of 1 dimension function.
First of all, we note that upper bounds are composable.
\begin{lem}
\label{lem:generating_composite}Given $F\leq_{x}f$ and $G\leq_{F(x)}g$,
we have that
\[
G\circ F\leq_{x}g\circ\overline{f}
\]
where $\overline{f}(s)=f(s)-f(0)$.
\end{lem}
\begin{proof}
Fix any $\Delta_{1},\Delta_{2},\cdots$ be unit vectors in the domain
of $F$. Let $H(x)=G\circ F(x)$. By chain rule, we have that
\begin{eqnarray*}
DH(x)[\Delta_{1}] & = & DG(F(x))[DF(x)[\Delta_{1}]],\\
DH(x)[\Delta_{1},\Delta_{2}] & = & DG(F(x))[D^{2}F(x)[\Delta_{1},\Delta_{2}]]\\
 &  & +D^{2}G(F(x))[DF(x)[\Delta_{1}],DF(x)[\Delta_{2}]],\\
DH(x)[\Delta_{1},\Delta_{2},\Delta_{3}] & = & DG(F(x))[D^{2}F(x)[\Delta_{1},\Delta_{2},\Delta_{3}]]\\
 &  & +D^{2}G(F(x))[D^{2}F(x)[\Delta_{1},\Delta_{2}],DF(x)[\Delta_{3}]]\\
 &  & +D^{2}G(F(x))[D^{2}F(x)[\Delta_{1},\Delta_{3}],DF(x)[\Delta_{2}]]\\
 &  & +D^{2}G(F(x))[D^{2}F(x)[\Delta_{2},\Delta_{3}],DF(x)[\Delta_{1}]]\\
 &  & +D^{3}G(F(x))[DF(x)[\Delta_{1}],DF(x)[\Delta_{2}],DF(x)[\Delta_{3}]],\\
 & \vdots
\end{eqnarray*}
Since $G\leq_{F(x)}g$, equation (\ref{eq:less_sim_est}) shows that
\begin{eqnarray*}
\norm{DH(x)[\Delta_{1}]} & \leq & g^{(1)}(0)\norm{DF(x)[\Delta_{1}]},\\
\norm{D^{2}H(x)[\Delta_{1},\Delta_{2}]} & \leq & g^{(1)}(0)\norm{D^{2}F(x)[\Delta_{1},\Delta_{2}]}\\
 &  & +g^{(2)}(0)\norm{DF(x)[\Delta_{1}]}_{2}\norm{DF(x)[\Delta_{2}]},\\
\norm{D^{3}H(x)[\Delta_{1},\Delta_{2},\Delta_{3}]} & \leq & g^{(1)}(0)\norm{D^{2}F(x)[\Delta_{1},\Delta_{2},\Delta_{3}]}\\
 &  & +g^{(2)}(0)\norm{D^{2}F(x)[\Delta_{1},\Delta_{2}]}\norm{DF(x)[\Delta_{3}]}\\
 &  & +g^{(2)}(0)\norm{D^{2}F(x)[\Delta_{1},\Delta_{3}]}\norm{DF(x)[\Delta_{2}]}\\
 &  & +g^{(2)}(0)\norm{D^{2}F(x)[\Delta_{2},\Delta_{3}]}\norm{DF(x)[\Delta_{1}]}\\
 & \vdots & +g^{(3)}(0)\norm{DF(x)[\Delta_{1}]}\norm{DF(x)[\Delta_{2}]}\norm{DF(x)[\Delta_{3}]}.
\end{eqnarray*}
Now, we use $F\leq_{x}f$ to get
\begin{eqnarray*}
\norm{DH(x)[\Delta_{1}]} & \leq & g^{(1)}(0)f^{(1)}(0)=\left(g\circ\overline{f}\right)^{(1)}(0)\\
\norm{D^{2}H(x)[\Delta_{1},\Delta_{2}]} & \leq & g^{(1)}(0)f^{(2)}(0)+g^{(2)}(0)f^{(1)}(0)^{2}=\left(g\circ\overline{f}\right)^{(2)}(0),\\
\norm{D^{3}H(x)[\Delta_{1},\Delta_{2},\Delta_{3}]} & \leq & g^{(1)}(0)f^{(3)}(0)+2g^{(2)}(0)f^{(2)}(0)f^{(1)}(0)\\
 & \vdots & +g^{(3)}(0)f^{(1)}(0)^{3}=\left(g\circ\overline{f}\right)^{(3)}(0).
\end{eqnarray*}
Therefore, we have that $\norm{D^{k}H(x)[\Delta_{i}]}\leq\left(g\circ\overline{f}\right)^{(k)}(0)$
for all $k\geq1$. For $k=0$, we have that $\norm{H(x)}=\norm{G(F(x))}\leq g(0)=g(\overline{f}(0))$.
\end{proof}
Next, we give some extra calculus rule for generating upper bounds. 
\begin{lem}
\label{lem:calculus_rule}Given that $H_{i}\leq_{x}h_{i}$ for all
$i=1,\cdots,k$. Then, we have that 
\[
\sum_{i=1}^{k}H_{i}\leq_{x}\sum_{i=1}^{k}h_{i}\text{ and }\prod_{i=1}^{k}H_{i}\leq_{x}\prod_{i=1}^{k}h_{i}.
\]
Given that $H\leq_{x}h$ and $\norm{H^{-1}(x)}\leq C$, we have that
\[
H^{-1}\leq_{x}\frac{1}{C^{-1}-(h(s)-h(0))}.
\]
\end{lem}
\begin{proof}
Fix $\Delta_{1},\Delta_{2},\cdots$ be unit vectors. For the first
claim, let $H=\sum H_{i}$, we note that
\[
D^{j}H[\Delta_{1},\cdots,\Delta_{j}]=\sum_{i=1}^{k}D^{j}H_{i}[\Delta_{1},\cdots,\Delta_{j}].
\]
Therefore, we have that $\norm{D^{j}H[\Delta_{1},\cdots,\Delta_{j}]}\leq\sum_{i=1}^{k}\norm{D^{j}H_{i}[\Delta_{1},\cdots,\Delta_{j}]}$.
Since $H\leq_{x}h_{i}$, we have that $\norm{D^{j}H[\Delta_{1},\cdots,\Delta_{j}]}\leq\sum_{i=1}^{k}h^{(j)}(0)$.
Hence, we have $H\leq_{x}\sum_{i=1}^{k}h_{i}$.

For the second claim, we let $G=\prod_{i=1}^{k}H_{i}$ and note that
\[
D^{j}G=\sum_{i_{1}+i_{2}+\cdots+i_{k}=j}\prod_{l=1}^{k}D^{i_{l}}H_{l}.
\]
Let $g=\prod_{i=1}^{k}h_{i}$. Then, we have that
\[
\norm{D^{j}G}\leq\sum_{i_{1}+i_{2}+\cdots+i_{k}=j}\prod_{l=1}^{k}\norm{D^{i_{l}}H_{l}}\leq\sum_{i_{1}+i_{2}+\cdots+i_{k}=j}\prod_{l=1}^{k}h_{i}^{(i_{l})}(0)=D^{j}g(0).
\]
Hence, $G\leq_{x}g$.

For the last claim, we first consider the function $\Phi(M)=M^{-1}$.
Note that $D\Phi[\Delta_{1}]=-M^{-1}\Delta_{1}M^{-1}$, $D^{2}\Phi[\Delta]=M^{-1}\Delta_{1}M^{-1}\Delta_{2}M^{-1}+M^{-1}\Delta_{2}M^{-1}\Delta_{1}M^{-1}$
and hence
\[
\norm{D^{j}\Phi}\leq j!\norm{M^{-1}}^{j+1}=j!C^{j+1}
\]
Hence, we have $\Phi\leq_{M}\sum_{j=0}^{\infty}\frac{j!C^{j+1}}{j!}s^{j}=\frac{1}{C^{-1}-s}$.
By Lemma \ref{lem:generating_composite}, we see that $H^{-1}\leq_{x}\frac{1}{C^{-1}-(h(s)-h(0))}$.
\end{proof}

\subsection{Explicit ODE\label{subsec:est_ode}}

\label{subsec:upperbound_ODE}In this section, we study the Taylor
expansion of the solution of ODE (\ref{eq:ODE}).
\begin{lem}
\label{lem:taylor_ODE}Let $u(t)$ be the solution of the ODE $u'(t)=F(u(t))$.
Suppose that $F\leq_{u(0)}f$ and let $\psi(t)$ be the solution of
the ODE $\psi'(t)=f(\psi(t))$ and $\psi(0)=0$. Then, we have
\[
\norm{u^{(k)}(0)}\leq\psi^{(k)}(0)
\]
for all $k\geq1$.
\end{lem}
\begin{proof}
Since $u'(t)=F(u(t))$, we have that
\begin{eqnarray*}
u^{(2)}(t) & = & DF(u(t))[u^{(1)}(t)],\\
u^{(3)}(t) & = & DF(u(t))[u^{(2)}(t)]+D^{2}F(u(t))[u^{(1)}(t),u^{(1)}(t)],\\
u^{(4)}(t) & = & DF(u(t))[u^{(3)}(t)]+2D^{2}F(u(t))[u^{(2)}(t),u^{(1)}(t)]\\
 &  & +D^{3}F(u(t))[u^{(1)}(t),u^{(1)}(t),u^{(1)}(t)]\\
 & \vdots
\end{eqnarray*}
Therefore, we have
\begin{eqnarray*}
\norm{u^{(2)}(0)} & \leq & f^{(1)}(0)\norm{u^{(1)}(0)},\\
\norm{u^{(3)}(0)} & \leq & f^{(1)}(0)\norm{u^{(2)}(0)}+f^{(2)}(0)\norm{u^{(1)}(0)}^{2},\\
\norm{u^{(4)}(0)} & \leq & f^{(1)}(0)\norm{u^{(3)}(0)}+2f^{(2)}(0)\norm{u^{(2)}(0)}\norm{u^{(1)}(0)}+f^{(3)}(0)\norm{u^{(1)}(0)}^{3},\\
 & \vdots
\end{eqnarray*}

By expanding $\psi'(t)=f(\psi(t))$ at $t=0$, we see that 
\begin{eqnarray*}
\psi^{(2)}(0) & = & f^{(1)}(0)\psi^{(1)}(0),\\
\psi^{(3)}(0) & = & f^{(1)}(0)\psi^{(2)}(0)+f^{(2)}(0)\left(\psi^{(1)}(0)\right)^{2},\\
\psi^{(4)}(0) & = & f^{(1)}(0)\psi^{(3)}(0)+2f^{(2)}(0)\psi^{(2)}(0)\psi^{(1)}(0)+f^{(2)}(0)\left(\psi^{(1)}(0)\right)^{3},\\
 & \vdots
\end{eqnarray*}
Since $\norm{u^{(1)}(0)}=\norm{F(u(0))}\leq f(0)=\psi^{(1)}(0)$,
we have that $\norm{u^{(k)}(0)}\leq\psi^{(k)}(0)$ for all $k\geq1$. 
\end{proof}
Now, we apply Lemma \ref{lem:taylor_ODE} to second order ODEs.
\begin{lem}
\label{lem:taylor_ODE_2nd}Let $u(t)$ be the solution of the ODE
$u''(t)=F(u'(t),u(t),t)$. Given some $\alpha>0$ and define $\overline{F}(y,x,t)=\alpha^{2}F(\alpha^{-1}y(t)+u'(0),x(t)+\alpha tu'(0),\alpha t)$.
Suppose that $\overline{F}\leq_{(0,u(0),0)}f$ and let $\psi(t)$
be the solution of the ODE $\psi'(t)=1+\psi(t)+f(\psi(t))$ and $\psi(0)=0$.
Then, we have
\[
\norm{u^{(k)}(0)}\leq\frac{\psi^{(k)}(0)}{\alpha^{k}}
\]
for all $k\geq2$.
\end{lem}
\begin{proof}
Let $x(t)=u(\alpha t)-\alpha tu'(0)$ and $y(t)=\alpha u'(\alpha t)-\alpha u'(0)$.
Then, we can write the problem into first order ODE
\begin{eqnarray*}
y'(t) & = & \alpha^{2}F(\alpha^{-1}y(t)+u'(0),x(t)+\alpha tu'(0),\alpha t)=\overline{F}(y(t),x(t),t),\\
x'(t) & = & y(t),\\
t' & = & 1.
\end{eqnarray*}
Let $\mathbf{F}(y,x,t)=(\overline{F}(y,x,t),y,1)$. Then, we have
that $\norm{D^{k}\mathbf{F}}\leq\norm{D^{k}\overline{F}}+\norm{D^{k}y}+\norm{D^{k}1}$.
Using $\overline{F}\leq_{(0,u(0),0)}f$, we have that  
\[
\mathbf{F}\leq_{(0,x(0),0)}f+t+1.
\]
By Lemma \ref{lem:taylor_ODE}, we know that 
\[
\norm{u^{(k)}(0)}=\frac{\norm{x^{(k)}(0)}}{\alpha^{k}}\leq\frac{\psi^{(k)}(0)}{\alpha^{k}}.
\]
\end{proof}

\pagebreak{}
\section{Implementation of Geodesic walk for log barrier \label{sec:Implementation-log-barrier}}

\subsection{Complex Analyticity of geodesics, parallel transport and Jacobi fields}

\label{subsec:Complex_analytic_geodesic_log}

The Cauchy\textendash Kowalevski theorem (Theorem \ref{thm:cauchy_kowalevski})
shows that if a differential equation is complex analytic, then the
equation has a unique complex analytic solution. As we see in section
\ref{subsec:RG_L}, the equations for geodesic, parallel transport
and Jacobi field involve only rational functions. Since rational functions
are complex analytic, the Cauchy-Kowalevski theorem shows that geodesics,
orthogonal frames and Jacobi fields are also complex analytic.
\begin{lem}
\label{cor:geo_jacobi_analytic_log}Geodesics, parallel transport
and Jacobi fields are complex analytic for the Hessian manifold induced
by the logarithmic barrier.
\end{lem}
We bound the higher-order derivatives of geodesic, parallel transport
and Jacobi field, using techniques developed in Section \ref{sec:est_var}.
Since these are complex analytic, this gives us a bound for their
radius of convergence. 

The purpose of these derivative bounds is to show that the solutions
of the corresponding ODEs are well-approximated by low-degree polynomials,
where the degree of the polynomial grows as $\log\left(\frac{1}{\epsilon}\right)$
for desired accuracy $\epsilon.$ The bound on the degree also implies
that the Collocation method for solving ODEs is efficient (roughly
matrix multiplication time). 

\subsubsection{Geodesic}

Motivated from the geodesic equation under Euclidean coordinate (\ref{eq:geodesic_log}),
we define the following auxiliary function 
\begin{equation}
\overline{F}_{\eta}(y,x,t)=\left(A_{x+t\eta}^{T}A_{x+t\eta}\right)^{-1}A_{x+t\eta}^{T}s_{x+t\eta,y+\eta}^{2}.\label{eq:geo_F}
\end{equation}
The derivative bounds on geodesic rely on the smoothness of this auxiliary
function.
\begin{lem}
\label{lem:geo_F_est}Under the normalization $A^{T}A=I$, $S_{x}=I$,
we have that 
\[
\overline{F}_{\eta}(y,x,t)\leq_{(0,x,0)}\frac{3\left(\norm{A\eta}_{4}+1\right)^{2}}{1-\max\left(8+8\norm{A\eta}_{\infty},1\right)t}
\]
where $\overline{F}_{\eta}$ is defined in (\ref{eq:geo_F}).
\end{lem}
\begin{proof}
By the assumption that $S_{x}=I$, we have that $\norm{S_{x}}_{2}=1$.
Using $A^{T}A=I$, we have that
\begin{eqnarray*}
\norm{DS_{x+t\eta}(x,t)[(d_{x},d_{t})]}_{2} & = & \norm{\Diag(Ad_{x}+d_{t}A\eta)}_{2}\\
 & \leq & \norm{Ad_{x}}_{\infty}+\norm{d_{t}A\eta}_{\infty}\\
 & \leq & \norm{d_{x}}_{2}+\left|d_{t}\right|\norm{A\eta}_{\infty}\\
 & \leq & (1+\norm{A\eta}_{\infty})\norm{(d_{x},d_{t})}_{2}.
\end{eqnarray*}
Let $\beta=1+\norm{A\eta}_{\infty}$. Then, we have that $S_{x+t\eta}\leq_{(x,t=0)}1+\beta t$. 

By using the inverse formula (Lemma \ref{lem:calculus_rule}), 
\[
S_{x+t\eta}^{-1}\leq_{(x,0)}\frac{1}{1-\beta t}.
\]
Using $A^{T}A=I$, we have that $A\leq_{(x,0)}1$ and hence product
formula (Lemma \ref{lem:calculus_rule}) shows that 
\[
A_{x+t\eta}=S_{x+t\eta}^{-1}A\leq_{(x,0)}\frac{1}{1-\beta t}.
\]
Since $A^{T}A=I$, we have that $AA^{T}\preceq I$ and hence $A_{x+t\eta}^{T}\leq_{(x,0)}\frac{1}{1-\beta t}$.
Therefore, we have that
\[
A_{x+t\eta}^{T}A_{x+t\eta}\leq_{(x,0)}\frac{1}{\left(1-\beta t\right)^{2}}.
\]
By using the inverse formula again,
\[
\left(A_{x+t\eta}^{T}A_{x+t\eta}\right)^{-1}\leq_{(x,0)}\frac{1}{2-\frac{1}{\left(1-\beta t\right)^{2}}}=\frac{(1-\beta t)^{2}}{2(1-\beta t)^{2}-1}.
\]
Hence, we have that
\[
\left(A_{x+t\eta}^{T}A_{x+t\eta}\right)^{-1}A_{x+t\eta}^{T}\leq_{(x,0)}\frac{(1-\beta t)^{2}}{2(1-\beta t)^{2}-1}\frac{1}{1-\beta t}=\frac{1-\beta t}{2(1-\beta t)^{2}-1}.
\]

Now, we consider the function $H(y)=(A(y+\eta))^{2}$. Note that 
\begin{align*}
\norm{H(y)}_{2} & \leq\norm{A(y+\eta)}_{4}^{2},\\
\norm{DH(y)[d]}_{2} & =2\norm{(A(y+\eta))Ad}_{2}\leq2\norm{A(y+\eta)}_{4}\norm d_{2},\\
\norm{D^{2}H(y)[d,d]}_{2} & \leq2\norm{(Ad)^{2}}_{2}\leq2\norm d_{2}^{2}.
\end{align*}
Therefore, we have that
\[
H\leq_{y=0}\norm{A\eta}_{4}^{2}+2\norm{A\eta}_{4}t+t^{2}=(\norm{A\eta}_{4}+t)^{2}.
\]
Hence, we have that
\[
\overline{F}_{\eta}(y,x,t)=\left(A_{x+t\eta}^{T}A_{x+t\eta}\right)^{-1}A_{x+t\eta}^{T}s_{x+t\eta,y+\eta}^{2}\leq_{(0,x,0)}\frac{(1-\beta t)\left(\norm{A\eta}_{4}+t\right)^{2}}{2(1-\beta t)^{2}-1}\frac{1}{\left(1-\beta t\right)^{2}}.
\]

Let $\phi(t)=\frac{\left(\norm{A\eta}_{4}+t\right)^{2}}{(2(1-\beta t)^{2}-1)(1-\beta t)}$
and we write $\phi(t)=\sum_{k=0}^{\infty}a_{k}t^{k}$. For any complex
$|z|=\frac{1}{8}\min\left(\frac{1}{\beta},\norm{A\eta}_{4}+8\right)$,
we have that
\begin{eqnarray*}
\left|\phi(z)\right| & \leq & \frac{(1-\frac{1}{8})}{(2(1-\frac{1}{8})^{2}-1)\left(1-\frac{1}{8}\right)^{2}}\left(\norm{A\eta}_{4}+\frac{\norm{A\eta}_{4}}{8}+1\right)^{2}\\
 & \leq & 3\left(\norm{A\eta}_{4}+1\right)^{2}\left(\norm{A\eta}_{4}+1\right)^{2}.
\end{eqnarray*}
Theorem \ref{thm:cauchy_estimate} shows that
\begin{eqnarray*}
\left|a_{k}\right| & \leq & 3\left(\norm{A\eta}_{4}+1\right)^{2}\left(8\max\left(\beta,\left(\norm{A\eta}_{4}+8\right)^{-1}\right)\right)^{k}\\
 & \leq & 3\left(\norm{A\eta}_{4}+1\right)^{2}\left(\max\left(8\beta,1\right)\right)^{k}
\end{eqnarray*}
Hence, we can instead bound $\overline{F}_{\eta}$ by 
\[
\overline{F}_{\eta}\leq_{(0,x,0)}\frac{3\left(\norm{A\eta}_{4}+1\right)^{2}}{1-\max\left(8+8\norm{A\eta}_{\infty},1\right)t}.
\]
\end{proof}
Now, we prove the geodesic has large radius of convergence. 
\begin{lem}
\label{lem:geodesic_asm}Under the normalization $A^{T}A=I$, $S_{x}=I$
and Euclidean coordinate, any geodesic starting at $x$ satisfies
the bound 
\[
\norm{\gamma^{(k)}(0)}_{2}\leq k!c^{k}
\]
for all $k\geq2$ where $c=512\norm{A\gamma'(0)}_{4}$.
\end{lem}
\begin{proof}
Recall that under Euclidean coordinate, the geodesic equation is given
by 
\[
\gamma''=\left(A_{\gamma}^{T}A_{\gamma}\right)^{-1}A_{\gamma}^{T}s_{\gamma'}^{2}\defeq F(\gamma',\gamma).
\]
So, we have that $\overline{F}_{\eta}(y,x,t)=\alpha^{2}F(\alpha^{-1}y+\gamma'(0),x+\alpha t\gamma'(0))$
with $\eta=\alpha\gamma'(0).$

Now, we estimate $\overline{F}_{\eta}$. Lemma \ref{lem:geo_F_est}
shows that 
\[
\overline{F}_{\eta}\leq_{(0,x,0)}\frac{3\left(\norm{A\eta}_{4}+1\right)^{2}}{1-\max\left(8+8\norm{A\eta}_{\infty},1\right)t}=\frac{3\left(\alpha\norm{A\gamma'(0)}_{4}+1\right)^{2}}{1-\max\left(8+8\alpha\norm{A\gamma'(0)}_{\infty},1\right)t}.
\]
Setting $\alpha=\norm{A\gamma'(0)}_{4}^{-1}$ and using $\norm{A\gamma'(0)}_{\infty}\leq\norm{A\gamma'(0)}_{4}$,
we have that
\[
\overline{F}_{\eta}\leq_{(0,x,0)}\frac{12}{1-16t}\leq_{0}\frac{16}{1-16t}-1-t.
\]
Lemma \ref{lem:taylor_ODE_2nd} shows that
\[
\norm{\gamma^{(k)}(0)}_{2}\leq\frac{\psi^{(k)}(0)}{\alpha^{k}}
\]
for all $k\geq2$ where $\psi(t)$ is the solution of 
\[
\psi'(t)=\frac{16}{1-16\psi(t)}\text{ with }\psi(0)=0.
\]
Solving it, we get that
\[
\psi(t)=\frac{1}{16}(1-\sqrt{1-512t}).
\]
By Theorem \ref{thm:cauchy_estimate}, we have that 
\[
\left|\psi^{(k)}(0)\right|\leq k!(512)^{k}.
\]
Hence, we have that
\[
\norm{\gamma^{(k)}(0)}_{2}\leq\frac{k!(512)^{k}}{\alpha^{k}}=k!(512)^{k}\norm{A\gamma'(0)}_{4}^{k}
\]
for all $k\geq2$. 
\end{proof}

\subsubsection{Parallel Transport}

Motivated from the equation for parallel transport under Euclidean
coordinate \ref{eq:parallel_log}, we define the following auxiliary
function 
\begin{equation}
F(t)=(A_{\gamma(t)}^{T}A_{\gamma(t)})^{-1}A_{\gamma(t)}^{T}S_{\gamma'(t)}A_{\gamma(t)}.\label{eq:orthgonal_frame_F}
\end{equation}
The derivative bounds on parallel transport rely on the smoothness
of this auxiliary function.
\begin{lem}
\label{lem:orthgonal_F_est}Given a geodesic $\gamma(t)$. Under the
normalization that $A^{T}A=I$, $S_{\gamma(0)}=I$, we have that 
\[
F(t)\leq_{0}\frac{c}{2(1-2ct)^{2}-\left(1-ct\right)^{2}}\frac{1-ct}{1-2ct}
\]
where $F$ is defined in (\ref{eq:orthgonal_frame_F}) and $c=512\norm{A\gamma'(0)}_{4}$.
\end{lem}
\begin{proof}
Lemma \ref{lem:geodesic_asm} shows that $\norm{\gamma^{(k)}(0)}_{2}\leq k!c^{k}$
for all $k\geq2$. Therefore, we have that 
\[
S_{\gamma(t)}\leq_{0}1+t\norm{A\gamma'(0)}_{\infty}+\sum_{k\geq2}(ct)^{k}\leq_{0}\frac{1}{1-ct}.
\]
Using Lemma \ref{lem:calculus_rule}, we have that
\[
A_{\gamma(t)}\leq_{0}\frac{1}{2-\frac{1}{1-ct}}=\frac{1-ct}{1-2ct}
\]
and hence
\[
(A_{\gamma(t)}^{T}A_{\gamma(t)})^{-1}A_{\gamma(t)}^{T}\leq_{0}\frac{\left(1-2ct\right)\left(1-ct\right)}{2(1-2ct)^{2}-\left(1-ct\right)^{2}}.
\]

Now, we note that
\begin{equation}
\Diag(A\gamma'(t))\leq_{0}\norm{A\gamma'(0)}_{\infty}+\sum_{k\geq1}\frac{(k+1)!c^{k+1}t^{k}}{k!}\leq_{0}\frac{c}{(1-ct)^{2}}\label{eq:diag_A_gamma_}
\end{equation}
and hence
\[
S_{\gamma'}=S_{\gamma}^{-1}\Diag(A\gamma'(t))\leq_{0}\frac{1-ct}{1-2ct}\frac{c}{(1-ct)^{2}}=\frac{c}{(1-2ct)(1-ct)}.
\]
This gives the result.
\end{proof}
Now, we prove parallel transport has large radius of convergence. 
\begin{lem}
\label{lem:orthgon_asm}Given a geodesic with $\gamma(0)=x$. Let
$v(t)$ be the parallel transport of a unit vector along $\gamma(t)$.
Under the normalization that $A^{T}A=I$, $S_{\gamma(0)}=I$, we have
that 
\[
\norm{v^{(k)}(0)}_{2}\leq k!\left(16c\right)^{k}
\]
for all $k\geq1$ where $c=512\norm{A\gamma'(0)}_{4}$.
\end{lem}
\begin{proof}
From (\ref{eq:parallel_log}), we have that 
\[
\frac{d}{dt}v(t)=(A_{\gamma(t)}^{T}A_{\gamma(t)})^{-1}A_{\gamma(t)}^{T}S_{\gamma'(t)}A_{\gamma(t)}v(t).
\]
Let $u(t)=v(\alpha t)$, then we have that
\begin{eqnarray*}
u'(t) & = & \alpha v'(\alpha t)\\
 & = & \alpha(A_{\gamma(\alpha t)}^{T}A_{\gamma(\alpha t)})^{-1}A_{\gamma(\alpha t)}^{T}S_{\gamma'(\alpha t)}A_{\gamma(\alpha t)}u(t).
\end{eqnarray*}
Let $F(x,t)=\alpha(A_{\gamma(\alpha t)}^{T}A_{\gamma(\alpha t)})^{-1}A_{\gamma(\alpha t)}^{T}S_{\gamma'(\alpha t)}A_{\gamma(\alpha t)}x$.
Then, Lemma \ref{lem:orthgonal_F_est} shows that 
\[
F(x,t)\leq_{(v,0)}\frac{\alpha c}{2(1-2\alpha ct)^{2}-\left(1-\alpha ct\right)^{2}}\frac{1-\alpha ct}{1-2\alpha ct}(1+t).
\]
Setting $\alpha=\frac{1}{8c}$, we have that
\[
F(x,t)\leq_{(v,0)}\frac{\frac{1}{8}}{2(1-\frac{t}{4})^{2}-\left(1-\frac{t}{8}\right)^{2}}\frac{1-\frac{t}{8}}{1-\frac{t}{4}}(1+t)\leq_{0}\frac{1}{1-t}.
\]

Lemma \ref{lem:taylor_ODE} shows that
\[
\norm{u^{(k)}(0)}_{2}\leq\psi^{(k)}(0)
\]
for all $k\geq1$ where $\psi(t)$ is the solution of 
\[
\psi'(t)=\frac{1}{1-\psi(t)}\text{ with }\psi(0)=0.
\]
Solving it, we get that
\[
\psi(t)=1-\sqrt{1-2t}.
\]
By Theorem \ref{thm:cauchy_estimate}, we have that for any $0\leq t\leq\frac{1}{2}$,
we have that 
\[
\norm{u^{(k)}(0)}_{2}\leq\left|\psi^{(k)}(0)\right|\leq k!2^{k}
\]
for all $k\geq1$. For $k\geq1$, we have that
\[
\norm{v^{(k)}(0)}_{2}\leq k!\left(16c\right)^{k}.
\]
\end{proof}

\subsubsection{Jacobi field}

Motivated from the equation for Jacobi field under orthogonal frame
basis (Lemma \ref{lem:Jacobi_field_log}), we define the following
auxiliary function 
\begin{equation}
F(t)=X^{-1}(A_{\gamma}^{T}A_{\gamma})^{-1}\left(A_{\gamma}^{T}S_{\gamma'}P_{\gamma}S_{\gamma'}A_{\gamma}-A_{\gamma}^{T}diag(P_{\gamma}s_{\gamma'}^{2})A_{\gamma}\right)X.\label{eq:jacobi_F}
\end{equation}
The derivative bounds on Jacobi field rely on the smoothness of this
auxiliary function.
\begin{lem}
\label{lem:jacobi_F_est}Given a geodesic $\gamma(t)$ and an orthogonal
frame $\{x_{i}\}_{i=1}^{n}$. Under the normalization that $A^{T}A=I$,
$S_{\gamma(0)}=I$, we have that 
\[
F(t)\leq_{0}\frac{12c^{2}}{1-64ct}
\]
where $F$ is defined in (\ref{eq:jacobi_F}) and $c=512\norm{A\gamma'(0)}_{4}$.
\end{lem}
\begin{proof}
We first bound the derivatives of $s_{\gamma'}^{2}$. Using $\norm{\gamma^{(k)}(0)}_{2}\leq k!c^{k}$
for all $k\geq2$ (Lemma \ref{lem:geodesic_asm}), we have that 
\[
\gamma'(t)-\gamma'(0)\leq_{0}\frac{d}{dt}\frac{1}{1-ct}=\frac{c}{(1-ct)^{2}}.
\]
Using $A^{T}A=I$ and $\Diag(A\gamma'(t))\leq_{0}\frac{c}{(1-ct)^{2}}$
(\ref{eq:diag_A_gamma_}), we have that
\begin{equation}
\Diag(A\gamma'(t))A(\gamma'(t)-\gamma'(0))\leq_{0}\left(\frac{c}{(1-ct)^{2}}\right)^{2}.\label{eq:jacobi_F_est1}
\end{equation}

Next, we note that 
\[
\Diag(A\gamma'(t))A\gamma'(0)=\Diag(A\gamma'(0))A(\gamma'(t)-\gamma'(0))+\Diag(A\gamma'(0))^{2}.
\]
and hence
\begin{equation}
\Diag(A\gamma'(t))A\gamma'(0)\leq_{0}\frac{c^{2}}{(1-ct)^{2}}+c^{2}.\label{eq:jacobi_F_est2}
\end{equation}

Combining (\ref{eq:jacobi_F_est1}) and (\ref{eq:jacobi_F_est2}),
we get
\[
\Diag(A\gamma'(t))A(\gamma'(t))\leq_{0}\left(\frac{c}{(1-ct)^{2}}\right)^{2}+\frac{c^{2}}{(1-ct)^{2}}+c^{2}\leq_{0}\frac{3c^{2}}{(1-ct)^{4}}
\]
In the proof of Lemma \ref{lem:orthgonal_F_est}, we showed that $S_{\gamma}\leq_{0}\frac{1}{1-ct}$
and hence
\begin{align*}
s_{\gamma'}^{2} & =S_{\gamma}^{-2}\Diag(A\gamma'(t))A(\gamma'(t)).\\
 & \leq_{0}\left(\frac{1}{2-\frac{1}{1-ct}}\right)^{2}\frac{3c^{2}}{(1-ct)^{4}}\\
 & =\frac{3c^{2}}{(1-ct)^{2}(2(1-ct)-1)^{2}}.
\end{align*}
In the proof of Lemma \ref{lem:orthgonal_F_est}, we showed that $A_{\gamma}\leq_{0}\frac{1-ct}{1-2ct}$
and $(A_{\gamma}^{T}A_{\gamma})^{-1}\leq_{0}\frac{\left(1-2ct\right)^{2}}{2(1-2ct)^{2}-\left(1-ct\right)^{2}}$.
Therefore, we have that
\[
P_{\gamma}=A_{\gamma}(A_{\gamma}^{T}A_{\gamma})^{-1}A_{\gamma}^{T}\leq_{0}\frac{\left(1-ct\right)^{2}}{2(1-2ct)^{2}-\left(1-ct\right)^{2}}.
\]
Hence, we have
\begin{align*}
P_{\gamma}s_{\gamma'}^{2} & \leq_{0}\frac{\left(1-ct\right)^{2}}{2(1-2ct)^{2}-\left(1-ct\right)^{2}}\frac{3c^{2}}{(1-ct)^{2}(1-2ct)^{2}}\\
 & =\frac{3c^{2}}{(1-2ct)^{2}(2(1-2ct)^{2}-\left(1-ct\right)^{2})}.
\end{align*}

Let $Y=(A_{\gamma}^{T}A_{\gamma})^{-1}\left(A_{\gamma}^{T}S_{\gamma'}P_{\gamma}S_{\gamma'}A_{\gamma}-A_{\gamma}^{T}diag(P_{\gamma}s_{\gamma'}^{2})A_{\gamma}\right)$.
By a similar proof, we have that
\begin{align*}
Y_{0}\leq_{0} & \frac{\left(1-2ct\right)^{2}}{2(1-2ct)^{2}-\left(1-ct\right)^{2}}\left(\frac{1-ct}{1-2ct}\right)^{2}\\
 & \left(\left(\frac{c}{(1-2ct)(1-ct)}\right)^{2}\frac{\left(1-ct\right)^{2}}{2(1-2ct)^{2}-\left(1-ct\right)^{2}}+\frac{3c^{2}}{(1-2ct)^{2}(2(1-2ct)^{2}-\left(1-ct\right)^{2})}\right)\\
= & \frac{4c^{2}}{(2(1-2ct)^{2}-\left(1-ct\right)^{2})^{2}}\left(\frac{1-ct}{1-2ct}\right)^{2}
\end{align*}

Next, we let $z(t)=X(t)v$ for some unit vector $v$. Since $X(t)$
is a parallel transport of $X(0)$, $z(t)$ is a parallel transport
of $X(0)v$ and hence Lemma \ref{lem:orthgon_asm} shows that $\norm{z^{(k)}(t)}_{2}\leq k!\left(16c\right)^{k}$
for all $i$. For any $k$, there is unit vector $v_{k}$ such that
$\norm{\frac{d}{dt^{k}}X(0)}_{2}=\norm{\frac{d}{dt^{k}}X(0)v_{k}}_{2}.$
Hence, we have that
\[
\norm{\frac{d}{dt^{k}}X(0)}_{2}\leq k!\left(16c\right)^{k}.
\]
Therefore, we have that $X\leq_{0}\frac{1}{1-16ct}$. Since $X(0)=I$,
we have that $X^{-1}\leq_{0}\frac{1}{2-\frac{1}{1-16ct}}=\frac{1-16ct}{1-32ct}$.
Thus, we have that
\begin{eqnarray*}
F(t) & \leq_{0} & \frac{4c^{2}}{(2(1-2ct)^{2}-\left(1-ct\right)^{2})^{2}}\left(\frac{1-ct}{1-2ct}\right)^{2}\frac{1}{1-32ct}\\
 & \leq_{0} & \frac{12c^{2}}{1-64ct}.
\end{eqnarray*}
\end{proof}
Now, we prove Jacobi field has large radius of convergence. 
\begin{lem}
\label{lem:jacobi_asm}Given a geodesic $\gamma(t)$ and an orthogonal
frame $\{X_{i}\}_{i=1}^{n}$ along $\gamma$. Let $V(t)$ be a Jacobi
field along $\gamma(t)$ with $V(0)=0$ and let $U(t)$ be the Jacobi
field under the $X_{i}$ coordinates, namely, $V(t)=X(t)U(t)$. Assume
that $U(0)=0$. Under the normalization that $A^{T}A=I$, $S_{\gamma(0)}=I$,
we have that 
\[
\norm{U^{(k)}(0)}_{2}\leq4(k!)\left(256c\right)^{k-1}\norm{U'(0)}_{2}
\]
for all $k\geq2$ where $c=512\norm{A\gamma'(0)}_{4}$.
\end{lem}
\begin{proof}
Note that $\frac{d^{2}U(t)}{dt^{2}}+F(t)U(t)=0$ where $F(t)$ defined
in (\ref{eq:jacobi_F}). Let $\overline{F}(U,t)=\alpha^{2}F(\alpha t)\left(U+\alpha tU'(0)\right)$
with $\alpha=\frac{1}{64c}$.

Since the differential equation is linear, we can rescale $U$ and
assume that $\norm{U'(0)}_{2}=\frac{1}{\alpha}=64c$. 

Using $\alpha^{2}F(\alpha t)\leq_{0}\frac{12c^{2}\alpha^{2}}{1-64c\alpha t}$
(Lemma \ref{lem:jacobi_F_est}), $U\leq_{U=0}t$ and $\alpha tU'(0)\leq_{t=0}\norm{\alpha U'(0)}_{2}t=t$,
we have that
\begin{eqnarray*}
\overline{F} & \leq_{(0,0)} & \frac{24c^{2}\alpha^{2}t}{1-64c\alpha t}\\
 & \leq_{0} & \frac{t}{1-t}\leq_{0}\frac{2t}{1-t}-1-t.
\end{eqnarray*}
Lemma \ref{lem:taylor_ODE_2nd} shows that that 
\[
\norm{U^{(k)}(0)}_{2}\leq\frac{\psi^{(k)}(0)}{\alpha^{k}}
\]
for all $k\geq2$ where $\psi(t)$ is the solution of 
\[
\psi'(t)=\frac{2}{1-\psi(t)}\text{ with }\psi(0)=0.
\]
Solving it, we get that
\[
\psi(t)=1-\sqrt{1-4t}.
\]
By Theorem \ref{thm:cauchy_estimate}, we have that for any $0\leq t\leq\frac{1}{4}$,
we have that 
\[
\left|\psi^{(k)}(0)\right|\leq k!(4)^{k}.
\]
Hence, we have that
\[
\norm{U^{(k)}(0)}_{2}\leq\frac{k!4^{k}}{\alpha^{k}}\leq k!\left(256c\right)^{k}
\]
for all $k\geq2$. 

Since we have rescaled the equation, we need to rescale it back and
get the result.
\end{proof}

\subsection{Computing Geodesic Equation}

\label{subsec:geodesic_equation_apx}To apply Theorem \ref{thm:Solving_ODE},
we define 
\begin{equation}
F(u,s)=A(A^{T}S^{-2}A)^{-1}A^{T}S^{-3}u^{2}\label{eq:geodesic_F}
\end{equation}
The following lemma bounds the Lipschitz constant of $F$.
\begin{lem}
\label{lem:geo_F_smoothness}Assuming $\frac{1}{2}\leq s_{i}\leq2$
for all $i$. Then, we have
\[
\norm{DF(u,s)[d_{u},d_{s}]}_{2}^{2}\leq10^{4}\left(\norm{d_{s}}_{\infty}^{2}\norm u_{4}^{4}+\norm{d_{u}}_{4}^{2}\norm u_{4}^{2}\right).
\]
\end{lem}
\begin{proof}
Let $A_{s}=S^{-1}A$ and $S_{d}=S^{-1}\Diag(d_{s})$. Then, we have
$F(u,s)=A(A_{s}^{T}A_{s})^{-1}A_{s}^{T}S^{-2}u^{2}$. Hence, we have
that
\begin{eqnarray*}
DF(u,s)[d_{u},d_{s}] & = & 2A(A_{s}^{T}A_{s})^{-1}A_{s}^{T}S_{d}A_{s}(A_{s}^{T}A_{s})^{-1}A_{s}^{T}S^{-2}u^{2}\\
 &  & -3A(A_{s}^{T}A_{s})^{-1}A_{s}^{T}S_{d}S^{-2}u^{2}\\
 &  & +2A(A_{s}^{T}A_{s})^{-1}A_{s}^{T}S^{-2}Ud_{u}
\end{eqnarray*}
Let $P=A_{s}\left(A_{s}^{T}A_{s}\right)^{-1}A_{s}^{T}$, then, we
have that
\begin{eqnarray*}
\norm{DF(u,s)[d_{u},d_{s}]}_{2}^{2} & \leq & 12\left(u^{2}\right)^{T}S^{-2}PS_{d}PS^{2}PS_{d}PS^{-2}u^{2}\\
 &  & +27\left(u^{2}\right)^{T}S^{-2}S_{d}PS^{2}PS_{d}S^{-2}u^{2}\\
 &  & +12d_{u}^{T}S^{-2}UPS^{2}PUS^{-2}d_{u}.
\end{eqnarray*}
Using that $P\preceq I$, we have that
\begin{eqnarray*}
 &  & \norm{DF(u,s)[d_{u},d_{s}]}_{2}^{2}\\
 & \leq & 12\norm S_{\infty}^{2}\left(u^{2}\right)^{T}S^{-2}PS_{d}^{2}PS^{-2}u^{2}+27\norm S_{\infty}^{2}\left(u^{2}\right)^{T}S^{-2}S_{d}^{2}S^{-2}u^{2}+12\norm S_{\infty}^{2}\norm{S^{-1}}_{\infty}^{4}\sum_{i}(d_{u})_{i}^{2}u_{i}^{2}\\
 & \leq & 39\norm S_{\infty}^{2}\norm{S^{-1}}_{\infty}^{4}\norm{S_{d}}_{\infty}^{2}\norm u_{4}^{4}+12\norm S_{\infty}^{2}\norm{S^{-1}}_{\infty}^{4}\norm u_{4}^{2}\norm{d_{u}}_{4}^{2}\\
 & \leq & 39\norm S_{\infty}^{2}\norm{S^{-1}}_{\infty}^{6}\norm{d_{s}}_{\infty}^{2}\norm u_{4}^{4}+12\norm S_{\infty}^{2}\norm{S^{-1}}_{\infty}^{4}\norm u_{4}^{2}\norm{d_{u}}_{4}^{2}.
\end{eqnarray*}
Now, we use that $\frac{1}{2}\leq s_{i}\leq2$ and get
\[
\norm{DF(u,s)[d_{u},d_{s}]}_{2}^{2}\leq9984\norm{d_{s}}_{\infty}^{2}\norm u_{4}^{4}+768\norm u_{4}^{2}\norm{d_{u}}_{4}^{2}.
\]
\end{proof}
Now, we can apply the collocation method to obtain a good approximation
of geodesics.
\begin{lem}
\label{lem:compute_geodesic}Let $\gamma$ be a random geodesic generated
by the geodesic walk with step size $h\leq\frac{1}{10^{20}\sqrt{n}}$.
With probability at least $1-O(\frac{1}{n})$ among $\gamma$, in
time $O(mn^{\omega-1}\log^{2}(n/\varepsilon))$, we can find $\overline{\gamma}$
such that 
\[
\max_{0\leq t\leq\ell}\norm{\gamma(t)-\overline{\gamma}(t)}_{\infty}\leq\varepsilon\quad\text{and}\quad\max_{0\leq t\leq\ell}\norm{\gamma'(t)-\overline{\gamma}'}_{\infty}\leq\varepsilon.
\]
Furthermore, $\overline{\gamma}$ is a $O(\log(1/\varepsilon))$ degree
polynomial.
\end{lem}
\begin{proof}
Let $s(t)=A\gamma(t)-b$. By rotating the space and rescaling the
rows of $A$, we assume that $s(0)_{i}=1$ for all $i$ and $A^{T}A=I$.
We define $F$ as (\ref{eq:geodesic_F}). Then, we have that
\begin{eqnarray*}
s''(t) & = & F(s',s),\\
s'(0) & = & A\gamma'(0),\\
s(0) & = & 1.
\end{eqnarray*}
We let $\alpha=4000\ell$ and 
\begin{eqnarray*}
K & \defeq & \alpha\norm{F(A\gamma'(0),1)}_{4}+\norm{A\gamma'(0)}_{4}\\
 & = & \alpha\norm{A(A^{T}A)^{-1}A^{T}\left(A\gamma'(0)\right)^{2}}_{2}+\norm{A\gamma'(0)}_{4}\\
 & \leq & \alpha\norm{A\gamma'(0)}_{4}^{2}+\norm{A\gamma'(0)}_{4}.
\end{eqnarray*}
Using $\norm{A_{\gamma}\gamma'(0)}_{4}\leq48n^{-1/4}$ with probability
$1-\frac{3}{n}$ (Lemma \ref{lem:V0_bound} and \ref{lem:gamma_est}),
we have that
\begin{align*}
K & \leq4000\ell\left(48n^{-1/4}\right)^{2}+48n^{-1/4}\leq100n^{-1/4},\\
\alpha K & \leq4000\ell\cdot100n^{-1/4}\leq\frac{1}{2}.
\end{align*}

For any $\norm{u-s'(0)}_{4}\leq K\leq100n^{-1/4}$, $\norm{s-1}_{4}\leq\alpha K\leq\frac{1}{2}$,
Lemma \ref{lem:geo_F_smoothness} shows that
\begin{eqnarray*}
\norm{DF(u,s)[d_{u},d_{s}]}_{4} & \leq & \norm{DF(u,s)[d_{u},d_{s}]}_{2}\\
 & \leq & 10^{2}\left(\norm{d_{s}}_{\infty}\norm u_{4}^{2}+\norm{d_{u}}_{4}\norm u_{4}\right)\\
 & \leq & 10^{7}\left(n^{-1/2}\norm{d_{s}}_{\infty}+n^{-1/4}\norm{d_{u}}_{4}\right).
\end{eqnarray*}
Therefore, for any $\norm{u_{1}-s'(0)}_{4}\leq K$, $\norm{s_{1}-1}_{4}\leq\alpha K$,
$\norm{u_{2}-s'(0)}_{4}\leq K$, $\norm{s_{2}-1}_{4}\leq\alpha K$,
we have that
\begin{eqnarray*}
\norm{F(u_{1},s_{1})-F(u_{2},s_{2})}_{4} & \leq & 10^{7}n^{-1/4}\norm{u_{1}-u_{2}}_{4}+10^{7}n^{-1/2}\norm{s_{1}-s_{2}}_{4}\\
 & \leq & \frac{1}{\alpha}\norm{u_{1}-u_{2}}_{4}+\frac{1}{\alpha^{2}}\norm{s_{1}-s_{2}}_{4}.
\end{eqnarray*}

Since $\gamma$ is analytic and $\norm{\gamma^{(k)}(0)}_{2}=O(k!n^{-k/4})$
(Lemma \ref{lem:orthgon_asm}), $\gamma(t)$ is $\varepsilon$ close
to the following polynomial 
\[
\sum_{k=0}^{\Theta(\log(1/\varepsilon))}\frac{1}{k!}\gamma^{(k)}(0)t^{k}
\]
for $0\leq t\leq cn^{1/4}$ for some small constant $c$. Hence, we
can apply Theorem \ref{thm:Solving_ODE_2nd} and find $\overline{\gamma}$
such that $\norm{\gamma-\overline{\gamma}}_{4}\leq\varepsilon$ and
$\norm{\gamma'-\overline{\gamma}'}_{4}\leq\varepsilon$ in $O(n\log^{3}(nK/\varepsilon))$
time plus $O(\log^{2}(K/\varepsilon))$ evaluations of $F$. Note
that each evaluation of $F$ involves solving a linear system and
hence it takes $O(mn^{\omega-1})$. Therefore, the total running time
is $O(mn^{\omega-1}\log^{2}(n/\varepsilon))$.
\end{proof}

\subsection{Computing Parallel Transport}

\label{subsec:parallel_transport_apx}
\begin{lem}
\label{lem:compute_parallel_transport}Given $\gamma$ be a random
geodesic generated by the geodesic walk with step size $h\leq\frac{1}{10^{20}\sqrt{n}}$
and an unit vector $v$. Let $v(t)$ be the parallel transport of
a unit vector along $\gamma(t)$. With probability at least $1-O(\frac{1}{n})$
among $\gamma$, in time $O(mn^{\omega-1}\log^{2}(n/\varepsilon))$,
we can find $\overline{v}$ such that $\max_{0\leq t\leq\ell}\norm{v(t)-\overline{v}(t)}_{\infty}\leq\varepsilon$.
Furthermore, $\overline{v}$ is a $O(\log(1/\varepsilon))$ degree
polynomial.

Similarly, given a basis $\{v_{i}\}_{i=1}^{n}$, with probability
at least $1-O(\frac{1}{n})$, in time $O(mn^{\omega-1}\log^{2}(n/\varepsilon))$,
we can find an approximate parallel transport $\overline{v_{i}}(t)$
of $\{v_{i}\}_{i=1}^{n}$ along $\gamma(t)$ such that $\max_{0\leq t\leq\ell}\norm{v_{i}(t)-\overline{v}_{i}(t)}_{\infty}\leq\varepsilon$
for all $i$.
\end{lem}
\begin{proof}
Recall that the equation for parallel transport (\ref{eq:parallel_log})
is given by
\[
\frac{d}{dt}v(t)=\left(A_{\gamma(t)}^{T}A_{\gamma(t)}\right)^{-1}A_{\gamma(t)}^{T}S_{\gamma'(t)}A_{\gamma(t)}v.
\]
By rotating the space and rescaling the rows of $A$, we assume that
$s(\gamma(0))_{i}=1$ for all $i$ and $A^{T}A=I$. In the proof of
Lemma \ref{lem:compute_geodesic}, we know that $\frac{1}{2}\leq s(\gamma(t))_{i}\leq2$
for all $0\leq t\leq\ell$. For any unit vector $u$, we have
\begin{align*}
\norm{\left(A_{\gamma}^{T}A_{\gamma}\right)^{-1}A_{\gamma}^{T}S_{\gamma'}A_{\gamma}u}_{2} & \leq2\norm{\left(A_{\gamma}^{T}A_{\gamma}\right)^{-1/2}A_{\gamma}^{T}S_{\gamma'}A_{\gamma}u}_{2}\\
 & \leq2\norm{S_{\gamma'}}_{\infty}\norm{A_{\gamma}u}_{2}\\
 & \leq4\norm{S_{\gamma'}}_{\infty}\leq192\left(\sqrt{\frac{\log n}{n}}+\sqrt{h}\right)
\end{align*}
where we used Lemma \ref{lem:gamma_est} in the last line. Using $h\leq\frac{1}{10^{20}\sqrt{n}}$,
we have that
\[
\norm{\left(A_{\gamma}^{T}A_{\gamma}\right)^{-1}A_{\gamma}^{T}S_{\gamma'}A_{\gamma}}_{2}\leq\frac{1}{2000\ell}.
\]

Since $v$ is analytic and $\norm{v^{(k)}(0)}_{2}=O(k!n^{-k/4})$
(Lemma \ref{lem:orthgon_asm}), $v(t)$ is $\varepsilon$ close to
a polynomial with degree $O(\log(1/\varepsilon))$ for $0\leq t\leq cn^{1/4}$
for some small constant $c$. Hence, we can apply Theorem \ref{thm:Solving_ODE}
and find $\overline{v}$ such that $\norm{v-\overline{v}}_{2}\leq\varepsilon$
in $O(n\log^{3}(n/\varepsilon))$ time plus $O(\log^{2}(1/\varepsilon))$
evaluations of $F$. Note that each evaluation of $F$ involves solving
a linear system and hence it takes $O(mn^{\omega-1})$. Therefore,
the total running time is $O(mn^{\omega-1}\log^{2}(n/\varepsilon))$.

For the last result, we note that each evaluation of $F$ becomes
computing matrix inverse and performing matrix multiplication and
they can be done in again $O(mn^{\omega-1})$ time.
\end{proof}

\subsection{Computing Jacobi field}

\label{subsec:jac_equation_apx}
\begin{lem}
\label{lem:compute_jacobi}Given $\gamma$ be a random geodesic generated
by the geodesic walk with step size $h\leq\frac{1}{10^{20}\sqrt{n}}$.
Let $\{X_{i}\}_{i=1}^{n}$ be an orthogonal frame along $\gamma$.
Let $v(t)$ be a Jacobi field along $\gamma(t)$ with $v(0)=0$ and
$\norm{v'(0)}\leq\cdots$. Let $u(t)$ be the Jacobi field under the
$X_{i}$ coordinates, namely, $v(t)=X(t)u(t)$. With probability at
least $1-O(\frac{1}{n})$ among $\gamma$, In time $O(mn^{\omega-1}\log^{2}(n/\varepsilon))$,
we can find $\overline{u}$ such that $\max_{0\leq t\leq\ell}\norm{u(t)-\overline{u}(t)}_{\infty}\leq\varepsilon$.
Furthermore, $\overline{u}$ is a $O(\log(1/\varepsilon))$ degree
polynomial.
\end{lem}
\begin{proof}
Recall that the equation for Jacobi field \ref{lem:Jacobi_field_log}
is given by
\[
\frac{d^{2}u}{dt^{2}}+X^{-1}(A_{\gamma}^{T}A_{\gamma})^{-1}\left(A_{\gamma}^{T}S_{\gamma'}P_{\gamma}S_{\gamma'}A_{\gamma}-A_{\gamma}^{T}\Diag(P_{\gamma}s_{\gamma'}^{2})A_{\gamma}\right)Xu=0.
\]
By rotating the space and rescaling the rows of $A$, we assume that
$s(\gamma(0))_{i}=1$ for all $i$ and $A^{T}A=I$. In the proof of
Lemma \ref{lem:compute_geodesic}, we know that $\frac{1}{2}\leq s(\gamma(t))_{i}\leq2$
for all $0\leq t\leq\ell$. Since $X$ is an orthogonal frame, we
have $XX^{T}=XX^{T}=I$. Hence, for any unit vector $v$, we have
\begin{align*}
 & \norm{X^{-1}(A_{\gamma}^{T}A_{\gamma})^{-1}\left(A_{\gamma}^{T}S_{\gamma'}P_{\gamma}S_{\gamma'}A_{\gamma}-A_{\gamma}^{T}\Diag(P_{\gamma}s_{\gamma'}^{2})A_{\gamma}\right)Xv}_{2}\\
\leq & 2\norm{(A_{\gamma}^{T}A_{\gamma})^{-1/2}\left(A_{\gamma}^{T}S_{\gamma'}P_{\gamma}S_{\gamma'}A_{\gamma}-A_{\gamma}^{T}\Diag(P_{\gamma}s_{\gamma'}^{2})A_{\gamma}\right)Xv}_{2}\\
\leq & 2\norm{S_{\gamma'}P_{\gamma}S_{\gamma'}A_{\gamma}Xv}_{2}+2\norm{\Diag(P_{\gamma}s_{\gamma'}^{2})A_{\gamma}Xv}_{2}.
\end{align*}
Using $\norm{S_{\gamma'}}_{2}\leq48\left(\sqrt{\frac{\log n}{n}}+\sqrt{h}\right)$
and $\norm{\Diag(P_{\gamma}s_{\gamma'}^{2})}_{2}\leq\norm{P_{\gamma}s_{\gamma'}^{2}}_{2}\leq\norm{s_{\gamma'}}_{4}^{2}\leq10^{4}n^{-1/2}$
(Lemma \ref{lem:gamma_est}), we have that
\begin{align*}
 & \norm{X^{-1}(A_{\gamma}^{T}A_{\gamma})^{-1}\left(A_{\gamma}^{T}S_{\gamma'}P_{\gamma}S_{\gamma'}A_{\gamma}-A_{\gamma}^{T}\Diag(P_{\gamma}s_{\gamma'}^{2})A_{\gamma}\right)Xv}_{2}\\
\leq & 2\left(48\left(\sqrt{\frac{\log n}{n}}+\sqrt{h}\right)\right)^{2}\norm{A_{\gamma}Xv}_{2}+4\cdot10^{4}n^{-1/2}\norm{A_{\gamma}Xv}_{2}\\
\leq & \frac{1}{(4000\ell)^{2}}
\end{align*}
where we used $h\leq\frac{1}{10^{20}\sqrt{n}}$ in the last line.

Since $u$ is analytic and $\norm{u^{(k)}(0)}_{2}=O(k!n^{-k/4})$
(Lemma \ref{lem:jacobi_asm}), $u(t)$ is $\varepsilon$ close to
a polynomial with degree $O(\log(1/\varepsilon))$ for $0\leq t\leq cn^{1/4}$
for some small constant $c$. Hence, we can apply Theorem \ref{thm:Solving_ODE}
and find $\overline{u}$ such that $\norm{u-\overline{u}}_{2}\leq\varepsilon$
in $O(n\log^{3}(n/\varepsilon))$ time plus $O(\log^{2}(1/\varepsilon))$
evaluations of $F$. Note that each evaluation of $F$ involves computing
matrix inversions and matrix multiplications and hence it takes $O(mn^{\omega-1})$.
Therefore, the total running time is $O(mn^{\omega-1}\log^{2}(n/\varepsilon))$.
\end{proof}

\subsection{Computing Geodesic Walk}
\begin{proof}[Proof of Theorem \ref{thm:implementation_log}]
To implement the geodesic walk, we use Lemma \ref{lem:compute_geodesic}
to compute the geodesic, Lemma \ref{lem:compute_parallel_transport}
to compute an orthogonal frame along the geodesic and Lemma \ref{lem:compute_jacobi}
to compute the Jacobi field along. Using the Jacobi field, we can
use (\ref{eq:1_step_prof}) and Lemma \ref{lem:formula_dexp} to compute
the probability from $x$ to $y$ and the probability from $y$ to
$x$. Using these probabilities, we can implement the rejection sampling.
It suffices to compute the geodesic and the probability up to $1/n^{O(1)}$
accuracy and hence these operations can be done in time $O(mn^{\omega-1}\log^{2}(n))$.

Note that we only use randomness to prove that $V(\gamma)$ is small
(Lemma \ref{lem:V0_bound}) and they can be checked. When we condition
our walk to that, we only change the distribution by very small amount.
Hence, this result is stated without mentioning the success probability. 
\end{proof}

\bibliographystyle{plain}
\bibliography{acg}

\pagebreak{}

\appendix
\listoftodos[Notes]

\section{Additional proofs}
\begin{proof}[Proof of Lemma \ref{lem:Hessian_normal_map}]
Let $T(y)=\exp_{x}(y)$. Then, we have
\[
F(T(y))=y.
\]
Therefore, 
\[
DF(T(y))[DT(y)[h]]=h
\]
and
\[
D^{2}F(T(y))[DT(y)[h],DT(y)[h]]+DF(T(y))[D^{2}T(y)[h,h]]=0.
\]

By the geodesic equation, we have
\[
\frac{d^{2}x_{k}}{dt^{2}}+\sum_{i,j}\Gamma_{ij}^{k}\frac{dx_{i}}{dt}\frac{dx_{j}}{dt}=0.
\]
Therefore, 
\begin{equation}
T(y)=x+y-\frac{1}{2}\sum_{i,j}\Gamma_{ij}^{k}y_{i}y_{j}+O(\norm y^{3}).\label{eq:T_expansion}
\end{equation}

Putting $y=0$, we have $DT(0)[h]=h$ and hence $DF(x)[h]=h.$ Now,
we note that 
\[
D^{2}F(x)[DT(0)[h],DT(0)[h]]+DF(x)[D^{2}T(0)[h,h]]=0.
\]
Hence,
\[
D^{2}F(x)[h,h]=-D^{2}T(0)[h,h].
\]
Using (\ref{eq:T_expansion}), we have $D^{2}T_{k}(0)[h,h]=-h^{T}\Gamma^{k}h$
and hence $D^{2}F_{k}(x)[h,h]=h^{T}\Gamma^{k}h.$
\end{proof}
\begin{lem}
\label{lem:norm_random_Ax}For $p\geq1$, we have 
\[
P_{x\sim N(0,I)}\left(\norm{Ax}_{p}^{p}\leq\left(\left(\frac{2^{p/2}\Gamma(\frac{p+1}{2})}{\sqrt{\pi}}\sum_{i}\norm{a_{i}}_{2}^{p}\right)^{1/p}+\norm A_{2\rightarrow p}t\right)^{p}\right)\leq1-\exp\left(-\frac{t^{2}}{2}\right).
\]
In particular, we have
\[
P_{x\sim N(0,I)}\left(\norm{Ax}_{4}^{4}\leq\left(\left(3\sum_{i}\norm{a_{i}}_{2}^{4}\right)^{1/4}+\norm A_{2\rightarrow4}t\right)^{4}\right)\leq1-\exp\left(-\frac{t^{2}}{2}\right).
\]
\end{lem}
\begin{proof}
Let $F(x)=\norm{Ax}_{p}$. Since $\left|F(x)-F(y)\right|\leq\norm A_{2\rightarrow p}\norm{x-y}_{2}$,
Gaussian concentration shows that 
\[
P_{x\sim N(0,I)}\left(F(x)\leq\E F(x)+\norm A_{2\rightarrow p}t\right)\leq1-\exp\left(-\frac{t^{2}}{2}\right).
\]
Since $x^{p}$ is convex, we have that 
\begin{eqnarray*}
\E\norm{Ax}_{p} & \leq & \left(\E\norm{Ax}_{p}^{p}\right)^{1/p}=\left(\sum_{i}\E\left|a_{i}^{T}x\right|^{p}\right)^{1/p}\\
 & = & \left(\E_{t\sim N(0,1)}|t|^{p}\sum_{i}\norm{a_{i}}_{2}^{p}\right)^{1/p}\\
 & = & \left(\frac{2^{p/2}\Gamma(\frac{p+1}{2})}{\sqrt{\pi}}\sum_{i}\norm{a_{i}}_{2}^{p}\right)^{1/p}.
\end{eqnarray*}
\end{proof}
\begin{lem}
\label{lem:integrate_interpolation}Let $\psi(x)=\frac{\sqrt{1-y^{2}}\cos\left(d\cos^{-1}x\right)}{d(x-y)}$
where $y=\cos(\frac{2k-1}{2d}\pi)$ for some integer $k\in[d]$. Then,
we have that
\[
\left|\int_{-1}^{t}\psi(x)dx\right|\leq\frac{2000}{d}.
\]
for any $-1\leq t\leq1$.
\end{lem}
\begin{proof}
For $d=1$, we have that $y=0$ and $\psi(x)=1$. Hence, we have the
result. 

From now on, we assume $d\geq2$. Let $z_{j}=\cos\left(j\pi/d\right)$.
Given $j\in[d]$ with $j\neq k$. Let 
\[
\overline{j}=\begin{cases}
j & \text{if }\left|z_{j}-y\right|\leq\left|z_{j-1}-y\right|\\
j-1 & \text{otherwises}
\end{cases}.
\]
Then, we have that
\begin{eqnarray}
\left|\int_{z_{j-1}}^{z_{j}}\psi(x)dx\right| & \leq & \frac{\sqrt{1-y^{2}}}{d}\left|\int_{z_{j-1}}^{z_{j}}\frac{\cos\left(d\cos^{-1}x\right)}{z_{\overline{j}}-y}dx\right|+\frac{\sqrt{1-y^{2}}}{d}\int_{z_{j-1}}^{z_{j}}\left|\frac{\cos\left(d\cos^{-1}x\right)(z_{\overline{j}}-x)}{(z_{\overline{j}}-y)(x-y)}\right|dx\nonumber \\
 & \leq & \frac{\sqrt{1-y^{2}}}{d}\left|\int_{z_{j-1}}^{z_{j}}\frac{\cos\left(d\cos^{-1}x\right)}{z_{\overline{j}}-y}dx\right|+\frac{\sqrt{1-y^{2}}}{d}\frac{(z_{j}-z_{j-1})^{2}}{(z_{\overline{j}}-y)^{2}}\nonumber \\
 & = & \frac{\sqrt{1-y^{2}}}{d(d^{2}-1)}\left|\frac{z_{j}-z_{j-1}}{z_{\overline{j}}-z_{k-\frac{1}{2}}}\right|+\frac{\sqrt{1-y^{2}}}{d}\frac{(z_{j}-z_{j-1})^{2}}{(z_{\overline{j}}-z_{k-\frac{1}{2}})^{2}}.\label{eq:integrate_psi_interpolaton}
\end{eqnarray}

Now, we upper bound the term $(z_{j}-z_{j-1})/(z_{\overline{j}}-z_{k-\frac{1}{2}})$.
By trigonometric formulas, we have 
\[
\left|\frac{z_{j}-z_{j-1}}{z_{\overline{j}}-z_{k-\frac{1}{2}}}\right|=\left|\frac{\cos\left(j\pi/d\right)-\cos\left((j-1)\pi/d\right)}{\cos\left(\overline{j}\pi/d\right)-\cos\left((k-\frac{1}{2})\pi/d\right)}\right|=\left|\frac{\sin\left(\frac{\pi}{2d}\right)\sin\left((j-\frac{1}{2})\frac{\pi}{d}\right)}{\sin\left((\frac{\overline{j}-k+\frac{1}{2}}{2})\frac{\pi}{d}\right)\sin\left((\frac{\overline{j}+k-\frac{1}{2}}{2})\frac{\pi}{d}\right)}\right|.
\]
Note that $\left|(\overline{j}-k+\frac{1}{2})\pi/(2d)\right|\leq\frac{\pi}{2}$
and hence 
\[
\left|\sin\left(\frac{\overline{j}-k+\frac{1}{2}}{2}\frac{\pi}{d}\right)\right|\geq\frac{2}{\pi}\left|\frac{\overline{j}-k+\frac{1}{2}}{2}\frac{\pi}{d}\right|=\left|\frac{\overline{j}-k+\frac{1}{2}}{d}\right|.
\]
By symmetric, we can assume $k-\frac{1}{2}\leq\frac{d}{2}$ and hence
$\left|(\overline{j}-k+\frac{1}{2})\pi/(2d)\right|\leq\frac{3\pi}{4}$
and
\[
\left|\sin\left(\frac{\overline{j}+k-\frac{1}{2}}{2}\frac{\pi}{d}\right)\right|\geq\frac{4}{3\sqrt{2}\pi}\left|\frac{\overline{j}+k-\frac{1}{2}}{2}\frac{\pi}{d}\right|=\frac{\sqrt{2}}{3}\left|\frac{\overline{j}+k-\frac{1}{2}}{d}\right|.
\]
Therefore, we have that
\[
\left|\frac{z_{j}-z_{j-1}}{z_{\overline{j}}-z_{k-\frac{1}{2}}}\right|\leq\frac{\frac{\pi}{2d}\left|(j-\frac{1}{2})\frac{\pi}{d}\right|}{\left|\frac{\overline{j}-k+\frac{1}{2}}{d}\right|\frac{\sqrt{2}}{3}\left|\frac{\overline{j}+k-\frac{1}{2}}{d}\right|}=\frac{3\pi^{2}}{2\sqrt{2}}\frac{\left|j-\frac{1}{2}\right|}{\left|\overline{j}-k+\frac{1}{2}\right|\left|\overline{j}+k-\frac{1}{2}\right|}.
\]
Note that $\left|\overline{j}+k-\frac{1}{2}\right|\geq j-\frac{1}{2}\geq0$
and hence
\[
\left|\frac{z_{j}-z_{j-1}}{z_{\overline{j}}-z_{k-\frac{1}{2}}}\right|\leq\frac{3\pi^{2}}{2\sqrt{2}}\frac{1}{\left|\overline{j}-k+\frac{1}{2}\right|}.
\]

Putting it into (\ref{eq:integrate_psi_interpolaton}), we get
\begin{align*}
\left|\int_{z_{j-1}}^{z_{j}}\psi(x)dx\right| & \leq\frac{1}{d(d^{2}-1)}\frac{3\pi^{2}}{2\sqrt{2}}\frac{1}{\left|\overline{j}-k+\frac{1}{2}\right|}+\frac{1}{d}\left(\frac{3\pi^{2}}{2\sqrt{2}}\frac{1}{\left|\overline{j}-k+\frac{1}{2}\right|}\right)^{2}\\
 & \leq\frac{21}{d(d^{2}-1)}+\frac{110}{d}\frac{1}{\left|\overline{j}-k+\frac{1}{2}\right|^{2}}.
\end{align*}
Using $d\geq2$, we have

\begin{eqnarray}
\left|\int_{-1}^{t}\psi(x)dx\right| & \leq & \sum_{k\neq j}\left|\int_{z_{j-1}}^{z_{j}}\psi(x)dx\right|+\int_{z_{k-1}}^{z_{k}}\left|\psi(x)\right|dx\nonumber \\
 & \leq & \frac{21}{d^{2}-1}+\frac{220}{d}\left(\frac{1}{0.5^{2}}+\frac{1}{1.5^{2}}+\frac{1}{2.5^{2}}+\cdots\right)+\int_{z_{k-1}}^{z_{k}}\left|\psi(x)\right|dx\nonumber \\
 & \leq & \frac{1200}{d}+\int_{z_{k-1}}^{z_{k}}\left|\psi(x)\right|dx.\label{eq:integrate_psi_interpolaton2}
\end{eqnarray}

To bound $\psi$ over $[z_{k},z_{k-1}]$, we write $x=\cos(\frac{2k-1-\theta}{2d}\pi)$
with $-1\leq\theta\leq1$. We have that
\begin{align*}
\psi(x) & =\frac{\sqrt{1-\cos(\frac{2k-1}{2d}\pi)^{2}}\cos\left(d\cos^{-1}\cos(\frac{2k-1-\theta}{2d}\pi)\right)}{d(\cos(\frac{2k-1-\theta}{2d}\pi)-\cos(\frac{2k-1}{2d}\pi))}\\
 & =\frac{-\sin(\frac{2k-1}{2d}\pi)\sin\left(\frac{\theta}{2}\pi\right)}{2d(\sin(\frac{\theta}{2d}\pi)\sin(\frac{2k-1-\theta/2}{2d}\pi))}.
\end{align*}
Since $-1\leq\theta\leq1$ and $d\geq2$, we have
\begin{align*}
\left|\psi(x)\right| & \leq\frac{\sin(\frac{2k-1}{2d}\pi)\left|\frac{\theta}{2}\pi\right|}{2d\frac{\sqrt{2}\left|\theta\right|}{d}\sin(\frac{2k-1-\theta/2}{2d}\pi)}=\frac{\pi}{4\sqrt{2}}\frac{\sin(\frac{2k-1}{2d}\pi)}{\sin(\frac{2k-1-\theta/2}{2d}\pi)}\\
 & \leq\frac{\pi}{4\sqrt{2}}\frac{\sin(\frac{1}{2d}\pi)}{\sin(\frac{1}{4d}\pi)}\leq2.
\end{align*}
Putting it into (\ref{eq:integrate_psi_interpolaton2}), we get
\[
\left|\int_{-1}^{t}\psi(x)dx\right|\leq\frac{1200}{d}+2\left|z_{k}-z_{k-1}\right|\leq\frac{2000}{d}.
\]
\end{proof}
\begin{xca}
Prove Fact \ref{fact:basic_RG} for Hessian manifolds using Lemma
\ref{lem:Hessian_formula} as a definition.
\end{xca}
\begin{proof}
We ignore the proof for (1) and (7) since we use Lemma \ref{lem:Hessian_formula}
as the definition here. (3) and (4) are immediate from the definition.

(5) By the definition of $D_{t}$, we have that $D_{t}=\nabla_{c'}$.
Using the definition of $\nabla_{v}u$, we have
\begin{align*}
 & \left\langle D_{t}u,w\right\rangle _{c(t)}+\left\langle u,D_{t}w\right\rangle _{c(t)}\\
= & \left\langle \sum_{ik}c'_{i}\frac{\partial u_{k}}{\partial x_{i}}e_{k}+\sum_{ijk}c'_{i}u_{j}\Gamma_{ij}^{k}e_{k},w\right\rangle _{c(t)}+\left\langle u,\sum_{ik}c'_{i}\frac{\partial w_{k}}{\partial x_{i}}e_{k}+\sum_{ijk}c'_{i}w_{j}\Gamma_{ij}^{k}e_{k}\right\rangle _{c(t)}\\
= & \left\langle \frac{du(c(t))}{dt},w(c(t))\right\rangle _{c(t)}+\left\langle u(c(t)),\frac{dw(c(t))}{dt}\right\rangle _{c(t)}\\
 & +\frac{1}{2}\left\langle \sum_{ijkl}\gamma'_{i}u_{j}g^{kl}\phi_{ijl}e_{k},w\right\rangle _{c(t)}+\frac{1}{2}\left\langle u,\sum_{ijkl}\gamma'_{i}w_{j}g^{kl}\phi_{ijl}e_{k}\right\rangle _{c(t)}\\
\overset{\diamondsuit}{=} & \left\langle \frac{du(c(t))}{dt},w(c(t))\right\rangle _{c(t)}+\left\langle u(c(t)),\frac{dw(c(t))}{dt}\right\rangle _{c(t)}\\
 & +\frac{1}{2}\sum_{ijklp}\gamma'_{i}u_{j}g^{kl}\phi_{ijl}g_{kp}w_{p}+\frac{1}{2}\sum_{ijklp}\gamma'_{i}w_{j}g^{kl}\phi_{ijl}g_{kp}u_{p}\\
\overset{\clubsuit}{=} & \left\langle \frac{du(c(t))}{dt},w(c(t))\right\rangle _{c(t)}+\left\langle u(c(t)),\frac{dw(c(t))}{dt}\right\rangle _{c(t)}\\
 & +\frac{1}{2}\sum_{ijklp}\gamma'_{i}u_{j}\phi_{ijp}w_{p}+\frac{1}{2}\sum_{ijklp}\gamma'_{i}w_{j}\phi_{ijp}u_{p}\\
= & \left\langle \frac{du(c(t))}{dt},w(c(t))\right\rangle _{c(t)}+\left\langle u(c(t)),\frac{dw(c(t))}{dt}\right\rangle _{c(t)}+u(t)^{T}\frac{d}{dt}g(\gamma(t))w(t)\\
= & \frac{d}{dt}\left\langle u,w\right\rangle _{c(t)}
\end{align*}
where we used $\left\langle a,b\right\rangle _{c(t)}=\sum a_{i}g_{ij}b_{j}$
on $\diamondsuit$ and $g^{kl}$ is the inverse of $g_{kp}$ on $\clubsuit$.

(6) For any map $c(t,s)$ on $M$, we have that
\begin{align*}
D_{s}\frac{\partial c}{\partial t} & =\frac{\partial}{\partial s}\frac{\partial c}{\partial t}+\sum_{ijk}\frac{\partial c_{i}}{\partial s}\frac{\partial c_{j}}{\partial t}\Gamma_{ij}^{k}e_{k}\\
 & =\frac{\partial}{\partial t}\frac{\partial c}{\partial s}+\sum_{ijk}\frac{\partial c_{i}}{\partial s}\frac{\partial c_{j}}{\partial t}\Gamma_{ij}^{k}e_{k}\\
 & =D_{t}\frac{\partial c}{\partial s}.
\end{align*}

(2) Recall a curve is geodesic if $\norm{\frac{d}{dt}\gamma(t)}_{\gamma(t)}$
is constant and $\left.\frac{d}{ds}\right|_{s=0}\int_{a}^{b}\norm{\frac{d}{dt}\gamma(t,s)}_{\gamma(t,s)}dt=0$
for any variation of $\gamma(t)$. We first prove that if $D_{t}\gamma'=0$
then it is a geodesic. For the first criteria, we use (5) and get
\[
\frac{d}{dt}\norm{\gamma'(t)}_{\gamma(t)}^{2}=\left\langle D_{t}\gamma',\gamma'\right\rangle _{\gamma(t)}+\left\langle \gamma',D_{t}\gamma'\right\rangle _{\gamma(t)}=0.
\]
Hence, $\norm{\frac{d}{dt}\gamma(t)}_{\gamma(t)}$ is a constant.
For the second criteria, we again use (5) and then (6) to get
\begin{align}
 & \frac{d}{ds}\int_{a}^{b}\norm{\frac{d}{dt}\gamma(t,s)}_{\gamma(t,s)}dt\nonumber \\
\overset{\diamondsuit}{=} & \int_{a}^{b}\frac{1}{\norm{\frac{d}{dt}\gamma(t,s)}_{\gamma(t,s)}}\left\langle D_{s}\frac{d}{dt}\gamma(t,s),\frac{d}{dt}\gamma(t,s)\right\rangle _{\gamma(t,s)}dt\nonumber \\
\overset{\clubsuit}{=} & \int_{a}^{b}\frac{1}{\norm{\frac{d}{dt}\gamma(t,s)}_{\gamma(t,s)}}\left\langle D_{t}\frac{d}{ds}\gamma(t,s),\frac{d}{dt}\gamma(t,s)\right\rangle _{\gamma(t,s)}dt\nonumber \\
\overset{\spadesuit}{=} & \frac{1}{\norm{\frac{d}{dt}\gamma(t,s)}_{\gamma(t,s)}}\int_{a}^{b}\frac{d}{dt}\left\langle \frac{d}{ds}\gamma(t,s),\frac{d}{dt}\gamma(t,s)\right\rangle _{\gamma(t,s)}-\left\langle \frac{d}{ds}\gamma(t,s),D_{t}\frac{d}{dt}\gamma(t,s)\right\rangle _{\gamma(t,s)}dt\nonumber \\
= & \frac{1}{\norm{\frac{d}{dt}\gamma(t,s)}_{\gamma(t,s)}}\left(\left\langle \frac{d}{ds}\gamma(b,s),\frac{d}{dt}\gamma(b,s)\right\rangle _{\gamma(b,s)}-\left\langle \frac{d}{ds}\gamma(a,s),\frac{d}{dt}\gamma(a,s)\right\rangle _{\gamma(b,s)}\right)\label{eq:geodesic_Dgamma}
\end{align}
where we used (5) on $\diamondsuit$, (6) on $\clubsuit$ and used
(5) and $\norm{\frac{d}{dt}\gamma(t,s)}_{\gamma(t,s)}$ is a constant
with respect to $t$ on $\spadesuit$. In the last line, we use the
fact that $D_{t}\frac{d}{dt}\gamma(t,s)=0$. Since $\frac{d}{ds}\gamma(a,s)=\frac{d}{ds}\gamma(b,s)=0$,
we have that $\frac{d}{ds}\int_{a}^{b}\norm{\frac{d}{dt}\gamma(t,s)}_{\gamma(t,s)}dt=0$.

To prove that any geodesic satisfies $D_{t}\gamma'=0$, by the calculation
in (\ref{eq:geodesic_Dgamma}), we have that
\[
0=\int_{a}^{b}\left\langle \frac{d}{ds}\gamma(t,s),D_{t}\frac{d}{dt}\gamma(t,s)\right\rangle _{\gamma(t,s)}dt
\]
for any variation $\gamma(t,s)$ of $\gamma(t)$. Since $\frac{d}{ds}\gamma(t,s)$
is chosen by us, we can put $\frac{d}{ds}\gamma(t,s)=D_{t}\frac{d}{dt}\gamma(t,s)$
and this gives that $\int_{a}^{b}\norm{D_{t}\frac{d}{dt}\gamma(t,s)}_{2}^{2}dt=0$.
Hence, we have the result.
\end{proof}
\begin{xca}
Prove the first part of Theorem \ref{thm:Jacobi_equation}.
\end{xca}
\begin{proof}
Note that
\begin{align*}
D_{t}D_{t}u & =\nabla_{c'}\nabla_{c'}\frac{\partial c}{\partial s}=\nabla_{c'}\nabla_{\frac{\partial c}{\partial s}}\frac{\partial c}{\partial t}\\
 & =\nabla_{\frac{\partial c}{\partial s}}\nabla_{c'}\frac{\partial c}{\partial t}-R(\frac{\partial c}{\partial s},\frac{\partial c}{\partial t})\frac{\partial c}{\partial t}\\
 & =-R(\frac{\partial c}{\partial s},\frac{\partial c}{\partial t})\frac{\partial c}{\partial t}
\end{align*}
where we used Fact \ref{fact:basic_RG} in the first equality, Fact
\ref{fact:formula_R} in the second equality and $\nabla_{c'}\frac{\partial c}{\partial t}=0$
in the last equality. Once we fix $s=0$, the partial derivatives
become derivatives in $t$.
\end{proof}

\end{document}